\documentclass[12pt,a4paper,reqno]{amsart}

\usepackage[english]{babel}
\usepackage{mathtools}
\usepackage{titlesec}
\usepackage[T1]{fontenc}
\usepackage{mathrsfs}
\usepackage{enumitem}
\usepackage{amsmath,amsfonts,amssymb,mathrsfs}

\usepackage{thmtools}

\usepackage{bigints} 
\usepackage{algorithmicx} 
\usepackage{algpseudocode} 
\usepackage[colorlinks=true, urlcolor=tealgreen,citecolor=tealgreen]{hyperref}
\usepackage{parskip}
\usepackage{parskip}
\usepackage{algorithm}
\usepackage[margin=1in]{geometry}
\usepackage{newtxtext} 
\usepackage{xcolor}
\usepackage{ragged2e}
\usepackage{lipsum}
\usepackage{xfrac}
\usepackage[autostyle,french=guillemets]{csquotes}
\usepackage[square, sort, numbers]{natbib}
\usepackage{float}
\usepackage{subfigure}

\usepackage{comment}
\usepackage{fancyhdr}
\renewcommand{\refname}{References}
\makeatletter
\renewcommand\bibsection{%
  \section*{{\normalfont\large\bfseries\sc\refname}\@mkboth{\MakeUppercase{\refname}}{\MakeUppercase{\refname}}}%
}%
\makeatother

\ifcsname c@theorem\endcsname\else
  \newtheorem{theorem}{Theorem}[section]

  \theoremstyle{definition}
  \newtheorem{assumption}[theorem]{Assumption}
  \newtheorem{lemma}[theorem]{Lemma}
  \newtheorem{proposition}[theorem]{Proposition}
  \theoremstyle{remark}
  \newtheorem{remark}[theorem]{Remark}
  \theoremstyle{plain}

\providecommand{\dif}{\mathrm{d}}

\providecommand{\Rzeros}{\mathcal{R}_{0}^{s}}

\providecommand{\indic}{\mathbf{1}}

\usepackage{tikz}
\usepackage{tikz-cd}
\usetikzlibrary{matrix,positioning,shapes,calc,petri,arrows,
babel,decorations.pathreplacing,arrows.meta,decorations.pathmorphing}

\usepackage{thmtools}
\usepackage[unq]{unique}

\declaretheoremstyle[headindent=0pt,headfont=\bfseries,notefont=\normalfont\itshape,spaceabove=1.5em]{mystyle1}
\declaretheoremstyle[headindent=0pt,headfont=\bfseries,notefont=\normalfont,bodyfont=\itshape,spaceabove=1.5em]{mystyle2}
\declaretheoremstyle[headfont=\itshape,notefont=\normalfont\itshape]{remark1}

\declaretheorem[name=Corollary,style=mystyle2,numberwithin=section]{corollary}

\definecolor{cinnabar}{rgb}{0.89, 0.26, 0.2}
\definecolor{dukebluep}{rgb}{0.0, 0.0, 0.61}
\definecolor{tealgreen}{rgb}{0.0, 0.51, 0.5}
\definecolor{bluep}{rgb}{0.2, 0.2, 0.6}
\definecolor{pinegreen}{rgb}{0.0, 0.47, 0.44}

\renewcommand{\footnotesize}{\normalsize}

\usepackage{xpatch}
\xpatchcmd{\author}{\relax#1\relax}{\relax\detokenize{#1}\relax}{}{}

\tikzset{
    ultra thin/.style= {line width=0.1pt},
    very thin/.style=  {line width=0.2pt},
    thin/.style=       {line width=0.4pt},
    semithick/.style=  {line width=0.6pt},
    thick/.style=      {line width=0.8pt},
    very thick/.style= {line width=1.2pt},
    ultra thick/.style={line width=1.6pt},
    multra thick/.style={line width=1pt},
}

\renewenvironment{abstract}
{\begin{quote}
\noindent\rule{\linewidth}{.5pt}\par\vspace*{\dimexpr-1em+7pt\relax}{\scshape\abstractname.}}
{\par\vspace*{-1em}\noindent \rule{\linewidth}{.5pt}\vspace*{1.5em}
\end{quote}
}

\let\svoverline\overline
\newcommand{\uc}[1]{\mkern 5mu\svoverline{\mkern-1mu#1\mkern-0.3mu \rule{0pt}{3.5mm}}\mkern 0mu}
\let\overline\uc

\titlespacing{\thetitle}
{0pt}
{-20pt}
{0pt}

\titleformat{\section}
{\centering\normalfont\large\scshape\bfseries}
{\thesection.}
{0.7em}
{}

\titlespacing{\section}
{0em}
{1.5em}
{1em}

\titleformat{\subsection}
{\normalfont\large\bfseries}
{\thesubsection}
{0.7em}
{}

\titlespacing{\subsection}
{0em}
{1.5em}
{1em}

\titleformat{\subsubsection}[runin]
{\normalfont\scshape}
{\thesubsubsection}
{0.7em}
{}

\titlespacing{\subsubsection}
{0em}
{1.5em}
{0.4em}
\renewcommand{\footnotesize}{\scriptsize}
\allowdisplaybreaks[1]

\title{\large\scshape Spread of Chronic Wasting Disease under Stochastic Environmental Conditions and its Control using Deep Reinforcement Learning}
\date{}

\begin{document}

\maketitle


{
\begin{tabular}{@{}c@{}}
Wei {\scshape Yin}\textsuperscript{1}, Wesley J. {\scshape Marrero}\textsuperscript{2},
Kamal {\scshape Jnawali}\textsuperscript{3},
Lale {\scshape Asik}\textsuperscript{4},\\
Michael G.  {\scshape Tyshenko}\textsuperscript{5}, and Tamer {\scshape Oraby}\textsuperscript{1,*}\\
\\
1. School of Mathematical and Statistical Sciences, \\ The University of Texas Rio Grande Valley, Edinburg, TX, USA. \\
2. Thayer School of Engineering, Dartmouth College, Hanover, NH, USA.\\
3. Department of Mathematics, State University of New York at Oswego, Oswego, NY, USA.\\
4. Mathematical Science Department, Wenzhou-Kean University, Zhejiang, China.  \\
5.  Risk Sciences International, Ottawa, ON, Canada.\\
\\
{}\textsuperscript{*} Corresponding Author: \texttt{tamer.oraby@utrgv.edu}
\\
\end{tabular}\par\bigskip}


\begin{abstract}
Chronic wasting disease (CWD) is a fatal neurodegenerative prion condition affecting deer, elk, moose, reindeer, muntjac, and potentially other members of the cervid family (Cervidae). Because free-ranging cervid populations experience environmental variability and demographic randomness, deterministic descriptions may miss important dynamics such as stochastic fade-out. We develop a stochastic compartmental Susceptible-Infectious-Environmental (SIU) model, formulated as a system of stochastic differential equations with normal (Skorokhod) reflection that keeps the susceptible class nonnegative. We investigate how environmental variability shapes cervid populations as CWD pressure intensifies and as control measures are applied. For the deterministic version, we characterize the disease-free equilibria, derive the basic reproduction number as the sum of a direct and an environmental contribution, and show that the endemic branch emerges through a forward transcritical bifurcation at $R_0=1$. For the reflected stochastic system, we establish local well-posedness and positivity, and obtain the disease-free invariant law. We derive the top Lyapunov exponent for invasion and show that it is unaffected by reflection. We then examine CWD mitigation using a memory-based deep reinforcement learning agent trained with Proximal Policy Optimization over a hybrid discrete-continuous action space, contrasting adjustments to the hunting rate, environmental decontamination, and integrated approaches. In the deterministic setting, hunting alone can suppress the disease but at the cost of a substantial (about $58\%$) decline in the total population, and decontamination alone can control the disease but only under sustained maximal effort, an approach usually limited to cervid farms. Whereas the combined policy more than doubles the cervid population while driving infection prevalence and environmental contamination to near-zero levels. In the stochastic setting, the learned policy contains the disease in roughly $80\%$ of independent realizations, with about $10\%$ producing large outbreaks, and its reliability degrades as noise intensity increases. Across all scenarios, the agent consistently raises and sustains environmental decontamination, identifying it as the dominant control lever. 

\end{abstract}
\newpage
\section{Introduction}\label{Sec:Intro}

Chronic wasting disease (CWD) is a transmissible spongiform encephalopathy (TSE) caused by misfolded prion protein \cite{williams1980chronic}. TSEs include diseases like bovine spongiform encephalopathy (BSE) in cattle, scrapie in sheep, transmissible mink encephalopathy (TME) in mink, Feline spongiform encephalopathy (FSE), exotic ungulate spongiform encephalopathy (EUE), and Creutzfeldt–Jakob disease (CJD) and variant Creutzfeldt–Jakob disease (vCJD) in humans \cite{imran2011overview}.

CWD can be transmitted between deer by various methods, including direct contact \cite{williams2002chronic}, vertical transmission from mother to offspring \cite{nalls2013mother}, and indirectly through contamination of shed prions in the environment \cite{miller2004environmental}. The incubation period of CWD ranges from 12–34 months \cite{kahn2004chronic}, and the disease can be transmitted during this period \cite{mathiason2009infectious}. At the end of the clinical disease, CWD-affected deer show changes in behavior and symptoms of physical wasting, increased thirst, increased urination, excessive salivation, difficulty swallowing, motor coordination impairment, and drooping of ears \cite{gilch2011chronic}.

Unmanaged CWD reduces the size of free-ranging deer populations due to high disease prevalence \cite{edmunds2016chronic}. The persistence of prions shed into the environment has been suggested to be a key driver of transmission. Hot spots frequented by deer for water, food, or shelter complicate management efforts \cite{almberg2011modeling}. Deer and elk field survival studies have shown that CWD-infected animals, as expected, have lower survival rates, which reduce population growth rates and have implications for genetic diversity over time, as population connectivity can influence disease transmission and spread \cite{devivo2017endemic, rogers2011diversity}.

There is no prophylactic or therapeutic treatment to prevent or treat CWD \cite{joly2006spatial, belay2004chronic}. Some cervid genotypes appear to be more resistant to CWD transmission, but even these genotypes can still become infected. There are concerns that cervid herds in North America will become extinct due to CWD’s continued geographic spread and environmental contamination if not controlled through effective deer population management.

Wildlife services and researchers in Canada and the US recognize that eradication of CWD is likely unrealistic at this time, with efforts focused on controlling its geographic spread \cite{rivera2019chronic}. A review of several CWD management plans includes several common objectives such as surveillance, actions to reduce the rate of spread and prevalence of CWD within the managed areas (targeted culling and harvest management); providing information and management support to areas where new CWD foci are detected; communication to other stakeholders (e.g., the public, hunter associations); and supporting CWD scientific research \cite{conner2021relationship}. Farmed deer CWD management strategy includes depopulation of the infected herd and restocking after two years \cite{mysterud2019review}. Free-range deer CWD management strategy relies mainly on population reduction to reduce disease transmission. Selective culling of deer in CWD endemic areas is used to control CWD prevalence \cite{conner2021relationship}.

Short-term disease management used by many states and provinces includes intentional, local reduction of infected animal numbers using culling that targets hot-spot areas and reduces population density. The strategy is an effective way to reduce the local prevalence of CWD. However, researchers have shown that such localized culling strategies have yielded mixed results and had little impact on long-term changes in CWD transmission. Long-term management of endemic CWD, with increasing prevalence and environmental contamination, requires different strategies to be more effective \cite{wolfe2018evaluation, conner2021relationship}.

Mathematical modeling has been widely used to study various disease dynamics. Many CWD modeling studies have focused on direct transmission and how subject-to-subject contacts relate to host density \cite{gross2001chronic, schauber2003chronic, wasserberg2009host}. Some other studies suggest that CWD is transmitted both directly and indirectly \cite{miller2004environmental, mathiason2009infectious, miller1998epidemiology}, through contact with saliva \cite{haley2009detection, mathiason2006infectious}, urine \cite{haley2009detection}, feces \cite{tamguney2009asymptomatic}, infected carcasses, and prion-contaminated environments \cite{david2011soil, schramm2006potential}. The simulation model proposed in \cite{almberg2011modeling} suggests that disease prevalence and the severity of population decline are driven by the duration that prions remain infectious in the environment, and highlights that environmental prion persistence is a critical factor in the timing and extent of effective control efforts.

However, unlike captive (farmed) deer populations, free-range deer interact with a number of unpredictable environmental conditions that have complex processes affecting their population dynamics, ultimately modulating CWD transmission and spread. This “noise in the system” or stochasticity needs to be considered when using adaptive management strategies. Besides environmental stochasticity, there is the influence of demographic stochasticity, defined as the variation in population dynamics due to the probabilistic nature of individual processes such as birth, death, or disease transmission – all of which can affect small populations \cite{lambert2018demographic}. Finally, stakeholder (human behavior) and adaptive management actions (or inaction) can also introduce stochasticity that affects cervid population dynamics. To date, CWD decision-making, as a form of stochasticity, has not been well integrated into modeling approaches.

Stochastic modeling can help to better understand the underlying processes that affect disease and has been proven helpful in assessing disease management strategies for species conservation \cite{wasserberg2009host}. Deterministic modeling has previously been used to predict optimized cervid harvesting strategies \cite{potapov2016chronic, belsare2021getting}. Due to the effects of environmental variability on population dynamics, modeling approaches indicate that structured population demographics are better captured by stochastic rather than deterministic approaches \cite{boyce2006demography, belsare2021getting}. In addition, stochastic processes can significantly affect population dynamics and may result in a 'stochastic fadeout' that leads to population extinction \cite{jnawali2022stochasticity}. This potential for extinction must be taken into account when using adaptive management practices with CWD that continue to spread geographically and increase in prevalence over time. A mathematical model incorporating stochasticity allows us to compare management strategies and more practically predict the effects of CWD harvesting.

Reinforcement learning (RL) is the subfield of machine learning that focuses on optimal decision-making under uncertainty \cite{sutton1998reinforcement, powell2007approximate, wiering2012reinforcement}. In RL, an agent is trained by interacting with an environment to gain maximal rewards. RL has achieved remarkable success in managing sequential decision-making and control in complex dynamic systems, such as Go and others \cite{fackeldey2022approximative, silver2016mastering}. In recent years, the application of RL to control tasks has advanced greatly, particularly for systems where traditional mathematical optimization or control approaches may be intractable. 

These developments have enabled promising applications across diverse fields, including video games \cite{mnih2015human}, robotics \cite{le2024comprehensive}, autonomous driving \cite{kiran2021deep}, epidemiology and healthcare \cite{abdellatif2023reinforcement, gemignani2025reinforcement, weltz2022reinforcement}.

Since 2020, reinforcement learning (RL) methods have been increasingly adopted in public health applications \cite{weltz2022reinforcement}. During the COVID-19 pandemic, several studies employed RL frameworks to inform epizootic control and policy design. For instance, \cite{weltz2022reinforcement} utilized a traditional SEIRD model to simulate disease spread and used RL to determine optimal mobility restriction levels of 0\%, 25\%, and 75\% \cite{ohi2020exploring}. Similarly, \cite{canto2024computing} proposed an RL-based approach to compute epizootic containment policies and minimize both public health impacts and economic losses. A simulation-deep RL framework was introduced in \cite{bushaj2023simulation} to identify optimal interventions under varying vaccination strategies, while \cite{guo2022pacar} developed a deep RL decision-making framework to evaluate non-pharmaceutical interventions, vaccination, and variant dynamics. Additional studies have applied RL to epizootic control in related contexts \cite{chadi2022reinforcement, wan2021multi, kwak2021deep}. Beyond epizootic modeling, RL techniques have also been employed to optimize sequential treatment strategies for a wide range of chronic diseases \cite{kosorok2019precision, tsiatis2019dynamic}. The connection between stochastic optimal control and reinforcement learning is investigated in \cite{quer2024connecting}, where importance sampling is applied to capture rare events, the problem is reformulated as an optimal control problem, and a parameterized approach recasts it as a stochastic optimization problem sharing the same Markov Decision Process framework as reinforcement learning.

In this study, we investigate epizootic decision-making by examining the interventions necessary to control a CWD outbreak among cervids as a function of the decision-maker's purpose. Our contribution is twofold, spanning both the analysis of the model and the design of control. On the modeling side, we adopt a stochastic compartmental Susceptible-Infectious-Environmental (SIU) model that incorporates both environmental and demographic stochasticity, and we regulate the susceptible class at the origin through normal (Skorokhod) reflection so that recruitment noise tied to the total population cannot drive the susceptible count negative. For the deterministic version, we identify the disease-free equilibria, apply the next-generation matrix method to obtain the basic reproduction number $R_0$ as the sum of a direct-transmission contribution and an environmental contribution, and show that the disease-free and endemic states exchange stability through a forward transcritical bifurcation at $R_0=1$. For the reflected stochastic system, we establish local existence, uniqueness, and positivity of solutions, derive the disease-free invariant distribution of the resulting stochastic logistic dynamics, and obtain a diffusion-corrected next-generation proxy $\mathcal{R}_0^s$ together with a top-Lyapunov-exponent criterion for disease invasion; we further show that the reflection leaves the disease-free measure and the transverse invasion linearization unchanged. On the control side, we compare intervention strategies, including hunting rate adjustment policies, environmental decontamination, and a combination of both, and we contrast the resulting policies under deterministic and stochastic dynamics. We also investigate how varying the basic reproduction number $R_0$ influences the optimal control strategies derived by the RL agent in a stochastic environment, and how the decision-maker's priorities, encoded in the reward weights, reshape the learned strategies. This analysis provides insights into the robustness and adaptability of the learned policies under different epizootic intensities. By explicitly incorporating stochasticity and parameter variability, our framework demonstrates its potential ability to apply to a broader class of infectious disease models beyond the one considered here.

To determine these interventions, we propose a deep reinforcement learning (DRL) framework comprising a simulation model and a DRL agent that applies interventions in the simulated environment. A memory-based agent is employed, since memory-based architectures such as recurrent neural networks and long short-term memory units allow the agent to retain and process historical information \cite{bakker2001reinforcement, goodfellow2016deep, hochreiter1997long}, which is valuable in epizootic settings where the dynamics depend not only on current states but also on prior intervention histories and latent environmental factors. Because our interventions combine discrete decontamination decisions with continuous hunting-rate adjustments, we adopt a hybrid action space and train the agent using Proximal Policy Optimization (PPO), a policy-gradient method well suited to mixed action spaces \cite{schulman2017proximal}. The model is parameterized for Texas white-tailed deer, with the hunting-rate bounds informed by reported statewide harvest rates and the maximum rate capped at the demographic extinction threshold identified by the deterministic stability analysis, so that the agent is never permitted to select a policy predicted to extirpate the host population.

The remainder of this paper is organized as follows. Section \ref{Sec:Model} introduces the stochastic SIU model and develops the mathematical analysis, including the deterministic equilibria, the basic reproduction number, and the transcritical bifurcation. It is followed by the well-posedness, disease-free invariant measure, and Lyapunov-based invasion criterion for the reflected stochastic system. It concludes with the reinforcement learning environment and training algorithm. Section \ref{Sec:Res} presents the experimental results for the deterministic and stochastic settings, including the three intervention strategies, sensitivity to initial conditions and reward weights, and robustness to noise and to the transmission intensity. Section \ref{Sec:Disc} provides the discussion, limitations, and concluding remarks. Technical details on the reflected model, the numerical scheme, parameter values, and auxiliary lemmas are collected in the Appendix.

\section{Mathematical model and analysis}\label{Sec:Model}
\subsection{Model}
Animal populations experience fluctuating environmental conditions, demographic heterogeneity, and random disease-transmission events, so deterministic models often miss important sources of variation. A stochastic approach offers a more realistic framework that more explicitly captures environmental variability and inherent randomness that shapes population dynamics. For example, stochastic processes can capture disease fadeout, in which a disease unexpectedly goes extinct even when deterministic models predict persistence \cite{lambert2018demographic}. This is especially relevant for CWD, whose spatial spread and increasing prevalence require careful consideration of extinction probabilities when evaluating management interventions. However, a deterministic simplification of the stochastic model is computationally less expensive and mathematically easier to analyze. It also offers a baseline for comparison with stochastic simulations, helping to identify the effects of randomness and environmental variability on disease transmission and control strategies.

Our work adopts a stochastic compartmental Susceptible-Infectious-Environmental (SIU) model based on the previously established framework designed for studying CWD \cite{oraby2014modeling, oraby2016using, vasilyeva2015aggregation}. We use a susceptible-infected (SI) model of CWD transmission between cervids with a transmission rate $\beta$. Population dynamics are governed by the birth rate $\nu$ and the death rate $\mu$ from natural causes and recreational hunting. We assume from the outset that $\nu>\mu$, see equation \eqref{population}. Let the CWD-specific death rate be denoted by $\gamma$. Environmental transmission is incorporated into the model through a compartment $U$ with a transmission rate $\beta_e$. Infectious prions are shed into the environment at a rate of $\xi$ and naturally cleared at a rate of $\epsilon$. Stochasticity in the birth, death, and environmental degradation is represented by white noise terms $\dot W_{1},\dot W_{2},\dot W_{3}$ with magnitudes proportional to $\sigma_1$, $\sigma_2$, and $\sigma_3$, respectively. The processes $W_{1},W_{2},W_{3}$ are independent standard Brownian motions. In numerical studies of SDE models, noise intensities are usually assumed to be nonnegative and sufficiently small to preserve biologically meaningful dynamics \cite{gray2011stochastic}. 

In this CWD model, the demographic driver $\sigma_{1}N\,\dif W_{1}$ in the susceptible equation is a fluctuating recruitment term whose noise is tied to the total population $N$ rather than to $S$ itself. Thus, $\sigma_{1}>0$ does not vanish on the set $\{S=0\}$. Since a negative susceptible count is not biologically feasible, we regulate $S$ at the origin by normal (Skorokhod) reflection, injecting the minimal recruitment flux needed to keep $S\geq 0$. To do so, let $L$ be a local time of $S$ at $0$ which is a continuous, nondecreasing, adapted process with
\begin{equation}\label{eq:localtime}
L_{0}=0,\qquad \dif L_{t}\ge 0,\qquad
\int_{0}^{t}\indic_{\{S_{s}>0\}}\,\dif L_{s}=0 .
\end{equation}
The reflection direction is the inward normal $+\mathbf{e}_{S}$, so $L$ acts only when $S=0$ and injects the least susceptible flux consistent with $S\ge 0$.  This preserves nondegenerate demographic noise; see also Corollary \ref{cor:dfe-invariance}.

Equation \eqref{M1} illustrates the resulting model:
\begin{align}\label{M1}
\frac{dS}{dt} &= \nu N\left(1-\frac{N}{K}\right) - \frac{\beta S I}{N} - \beta_e S U - \mu S\left(1-\frac{N}{K}\right) - \alpha_1 S + \sigma_1 N \dot{W}_1 + \sigma_2 S \dot{W}_2 + \dot{L} \nonumber\\
\frac{dI}{dt} &= \frac{\beta S I}{N} + \beta_e S U - \mu I\left(1-\frac{N}{K}\right) - \gamma I - \alpha_1 I + \sigma_2 I \dot{W}_2 \\
\frac{dU}{dt} &= \xi I - \epsilon U - \alpha_2 U + \sigma_3 U \dot{W}_3 \nonumber
\end{align}
where $S+I=N$ and the population size is regulated by a carrying capacity $K$. Here, $S$ denotes the number of susceptible individuals, $I$ is the number of infected individuals, and $U$ is the environmental load of the CWD prion. Because births, deaths, and infections occur dynamically, the population size $N$ is not fixed. 

\begin{assumption}\label{ass:params}
We assume that all rate constants $\nu,\mu,\beta,\beta_{e},\gamma,\alpha_{1},\alpha_{2},
\xi,\epsilon,K$ and noise intensities $\sigma_{1},\sigma_{2},\sigma_{3}$ are
nonnegative, with $K>0$. We also assume that the initial state $(S_{0},I_{0},U_{0})\in D$ satisfies $N_{0}=S_{0}+I_{0}>0$.
\end{assumption}

\subsection{The CWD Deterministic Model.} We begin by studying the stability of the deterministic version of the model, given by the following system
\begin{align}\label{M2deterministic}
\frac{dS}{dt} &= \nu N\left(1-\frac{N}{K}\right) - \frac{\beta S I}{N} - \beta_e S U - \mu S\left(1-\frac{N}{K}\right) - \alpha_1 S  \nonumber\\
\frac{dI}{dt} &= \frac{\beta S I}{N} + \beta_e S U - \mu I\left(1-\frac{N}{K}\right) - \gamma I - \alpha_1 I \\
\frac{dU}{dt} &= \xi I - \epsilon U - \alpha_2 U   \nonumber
\end{align}
where $S+I=N$.

\paragraph*{\textbf{Disease-Free Equilibrium States.}}

In the disease-free case, $I=0$ and $U=0$, so $N=S$ and the model reduces to
\begin{align}\label{population}
\frac{dS}{dt}&=S\left[(\nu-\mu)\left(1-\frac{S}{K}\right)-\alpha_1\right].
\end{align}
Recall that we assume that $\nu>\mu$. This yields two disease-free equilibria
\begin{enumerate}
    \item $\mathcal{E}_{1}\equiv(0,0,0)$.
    \item $\mathcal{E}_2\equiv\left(K\left(1-\frac{\alpha_1}{\nu-\mu}\right),0,0\right)$, which exists only if $0<\alpha_1<\nu-\mu$.
\end{enumerate}

To study local asymptotic stability of \eqref{M2deterministic}, we compute the Jacobian of the system of ODEs with $N=S+I$,

$$
\mathcal{J}(S,I,U)=
{\Small
\begin{pmatrix}
\nu\left(1-\frac{2N}{K}\right)-\mu\left(1-\frac{N}{K}\right)+\frac{\mu S}{K}-\alpha_1-\frac{\beta I^2}{N^2}-\beta_e U
&
\nu\left(1-\frac{2N}{K}\right)+\frac{\mu S}{K}-\frac{\beta S^2}{N^2}
&
-\beta_e S
\\[6pt]
\frac{\beta I^2}{N^2}+\beta_e U+\frac{\mu I}{K}
&
\frac{\beta S^2}{N^2}-\gamma-\alpha_1-\mu\left(1-\frac{N}{K}\right)+\frac{\mu I}{K}
&
\beta_e S
\\[6pt]
0 & \xi & -(\epsilon+\alpha_2)
\end{pmatrix}
}.
$$
At any disease-free equilibrium $(S,0,0)$, and so $N=S$, the first column below the diagonal vanishes, and the matrix reduces to
$$
\mathcal{J}(S,0,0)=
\begin{pmatrix}
(\nu-\mu)\left(1-\frac{2S}{K}\right)-\alpha_1
&
\nu\left(1-\frac{2S}{K}\right)+\frac{\mu S}{K}-\beta
&
-\beta_e S
\\[4pt]
0
&
\beta-\gamma-\alpha_1-\mu\left(1-\frac{S}{K}\right)
&
\beta_e S
\\[4pt]
0 & \xi & -(\epsilon+\alpha_2)
\end{pmatrix}.
$$
Because this matrix is block upper-triangular, one eigenvalue is the $(1,1)$ entry. The remaining two are those of the lower-right block governing the infected-environmental subsystem $(I,U)$, and a way in which the $(1,2)$ entry does not affect the spectrum.

\begin{enumerate}
\item For $\mathcal{E}_{1}\equiv(0,0,0)$, taking the limit along the disease-free axis ($S/N\to1$),
$$
\mathcal{J}(0,0,0)=
\begin{pmatrix}
\nu-\mu-\alpha_1 & \nu-\beta & 0\\[4pt]
0 & \beta-\gamma-\alpha_1-\mu & 0\\[4pt]
0 & \xi & -(\epsilon+\alpha_2)
\end{pmatrix},
$$
with eigenvalues $\nu-\mu-\alpha_1$, $\beta-\gamma-\alpha_1-\mu$, and $-(\epsilon+\alpha_2)<0$. Hence $\mathcal{E}_1$ is locally asymptotically stable if $\max\left(\nu, \beta-\gamma\right)<\mu+\alpha_1$.

\item For $\mathcal{E}_2\equiv\left(S_2^{*}=K\left(1-\frac{\alpha_1}{\nu-\mu}\right),0,0\right)$, thus
$1-\frac{S_2^{*}}{K}=\frac{\alpha_1}{\nu-\mu}$, and
$$
\mathcal{J}(\mathcal{E}_2)=
\begin{pmatrix}
\alpha_1-(\nu-\mu)
&
\nu\left(1-\frac{2S_2^{*}}{K}\right)+\frac{\mu S_2^{*}}{K}-\beta
&
-\beta_e S_2^{*}
\\[4pt]
0
&
\beta-\gamma-\frac{\nu\alpha_1}{\nu-\mu}
&
\beta_e S_2^{*}
\\[4pt]
0 & \xi & -(\epsilon+\alpha_2)
\end{pmatrix}.
$$
The first eigenvalue is $\alpha_1-(\nu-\mu)<0$, which holds automatically whenever $\mathcal{E}_2$ exists. The remaining two are the eigenvalues of
$$
B_2=
\begin{pmatrix}
\beta-\gamma-\frac{\nu\alpha_1}{\nu-\mu} & \beta_e S_2^{*}\\[4pt]
\xi & -(\epsilon+\alpha_2)
\end{pmatrix},
$$
so $\mathcal{E}_2$ is locally asymptotically stable if and only if $\operatorname{tr}B_2<0$ and
$\det B_2>0$, i.e.
$$
\beta-\gamma-\frac{\nu\alpha_1}{\nu-\mu}<\epsilon+\alpha_2
\qquad\text{and}\qquad
(\epsilon+\alpha_2)\left(\gamma+\frac{\nu\alpha_1}{\nu-\mu}-\beta\right)>\beta_e\,\xi\,S_2^{*}.
$$
\end{enumerate}

\paragraph{\textbf{The Basic Reproduction Number.}}
To obtain the basic reproduction number, we apply the next-generation matrix method \cite{van2002reproduction} at $\mathcal{E}_2$. Order the infected compartments as $(I,U)$ and write, for brevity, let
$$
a\equiv\gamma+\alpha_1+\mu\left(1-\frac{S_2^{*}}{K}\right)
=\gamma+\frac{\nu\alpha_1}{\nu-\mu},
\qquad
d\equiv\epsilon+\alpha_2,
\qquad
S_2^{*}=K\left(1-\frac{\alpha_1}{\nu-\mu}\right).
$$
At $\mathcal{E}_2$, we have $S/N=1$, so we decompose the infected subsystem into new infections $\mathcal{F}$ and transitions $\mathcal{V}$. The environmental shedding $\xi I$ is a transition, not a primary infection. That gives
$$
F=\begin{pmatrix}\beta & \beta_e S_2^{*}\\[2pt]0 & 0\end{pmatrix},
\qquad
V=\begin{pmatrix}a & 0\\[2pt]-\xi & d\end{pmatrix},
\qquad
FV^{-1}=\begin{pmatrix}\dfrac{\beta}{a}+\dfrac{\beta_e\xi S_2^{*}}{ad} & \dfrac{\beta_e S_2^{*}}{d}\\[8pt]0 & 0\end{pmatrix}.
$$
The reproduction number is the spectral radius $R_0=\rho(FV^{-1})$, which is
\begin{equation}\label{R0}
R_0=R_{\mathrm{dir}}+R_{\mathrm{env}}:=\frac{\beta}{a}+\frac{\beta_e\,\xi\,S_2^{*}}{a\,d}
=\frac{\beta}{\gamma+\frac{\nu\alpha_1}{\nu-\mu}}
+\frac{\beta_e\,\xi\,K\left(1-\frac{\alpha_1}{\nu-\mu}\right)}
{\left(\gamma+\frac{\nu\alpha_1}{\nu-\mu}\right)(\epsilon+\alpha_2)},
\end{equation}
The reproduction number is the sum of a direct-transmission contribution $R_{\mathrm{dir}}$ and an environmental (indirect) contribution $R_{\mathrm{env}}$. In this case, an infectious host transmits directly at rate $\beta$ over the mean infectious period $1/a$ and additionally sheds $\xi/a$ units of environmental load. Each of them persists for time $1/d$ and generates new infections at rate $\beta_e S_2^{*}$.

The reproduction number \eqref{R0} controls the stability of $\mathcal{E}_2$ established above. Indeed, with $\det B_2=ad-\beta d-\beta_e\xi S_2^{*}$ one has the exact identity $\det B_2 = a\,d\,(1-R_0)$. Since $ad>0$, the determinant condition $\det B_2>0$ is equivalent to $R_0<1$. Furthermore $\det B_2>0$ implies $a>\beta$ and hence $\operatorname{tr}B_2<0$ automatically, so the trace condition is redundant. Therefore, the disease-free equilibrium $\mathcal{E}_2$ is locally asymptotically stable if and only if $R_0<1$, and unstable when $R_0>1$.

\paragraph{\textbf{Disease-Endemic Equilibrium States.}}

When the disease is endemic ($I^{*}>0$, $U^{*}>0$), setting all derivatives in \eqref{M2deterministic} to zero gives. First, setting $\dot U=0$ and the total-population balance $\dot N=0$,
$$
U^{*}=\frac{\xi I^{*}}{\epsilon+\alpha_2},
\qquad
N^{*}=S^{*}+I^{*}=K\left(1-\frac{\alpha_1+\gamma P^{*}}{\nu-\mu}\right),
\qquad P^{*}\equiv\frac{I^{*}}{S^{*}+I^{*}} \text{ (prevalence)}.
$$
Second, setting $\dot I=0$ gives
equation
\begin{equation}\label{Pequation}
\begin{aligned}
&\frac{\beta_e \xi K}{(\epsilon+\alpha_2)(\nu-\mu)}
\Bigl[\gamma (P^{*})^{2}-\bigl(\gamma+\nu-\mu-\alpha_1\bigr)P^{*}+\bigl(\nu-\mu-\alpha_1\bigr)\Bigr]\\
&\qquad-\left(\frac{\nu\gamma}{\nu-\mu}+\beta-\gamma\right)P^{*}+\beta-\gamma-\frac{\nu\alpha_1}{\nu-\mu}=0.
\end{aligned}
\end{equation}
Solving \eqref{Pequation} for $P^{*}$ then determines the endemic state in the original variables,
\begin{equation}\label{NUequation}
\begin{aligned}
I^{*}=P^{*}N^{*},\qquad S^{*}=(1-P^{*})\,N^{*},\qquad U^{*}=\frac{\xi I^{*}}{\epsilon+\alpha_2}.
\end{aligned}
\end{equation}
Let $\kappa\equiv\dfrac{\beta_e\xi K}{(\epsilon+\alpha_2)(\nu-\mu)}$ and $m\equiv\nu-\mu$, hence equation \eqref{Pequation} is the quadratic equation $f(P^{*}):=A(P^{*})^{2}+BP^{*}+C=0$ with
$$
A=\kappa\gamma,\quad
B=-\kappa(\gamma+m-\alpha_1)-\frac{\nu\gamma}{m}-(\beta-\gamma),\quad
C=\kappa(m-\alpha_1)+\beta-\gamma-\frac{\nu\alpha_1}{m},
$$
and discriminant
$$
\Delta=B^{2}-4\kappa\gamma\,C
=\Big[\kappa(\gamma+m-\alpha_1)+\frac{\nu\gamma}{m}+(\beta-\gamma)\Big]^{2}
-4\kappa\gamma\Big[\kappa(m-\alpha_1)+\beta-\gamma-\frac{\nu\alpha_1}{m}\Big].
$$
Using $\kappa(m-\alpha_1)=\beta_e\xi S_2^{*}/(\epsilon+\alpha_2)$ and $a=\gamma+\frac{\nu\alpha_1}{m}$, the constant term factors through the reproduction number \eqref{R0},
$$
C=a\,(R_0-1),
$$
so $\operatorname{sign}C=\operatorname{sign}(R_0-1)$. Since $f(1)=-\nu(\gamma+\alpha_1)/m<0$, then $P=1$ always separates the two roots, and therefore the larger one is inadmissible. With $A>0$, a unique biologically admissible endemic prevalence therefore exists if and only if $R_0>1$, given by the lower root
$$
P^{*}=\frac{-B-\sqrt{\Delta}}{2A}\in(0,1),\qquad R_0>1 .
$$

Figure \ref{fig:R0} (a) illustrates how the reproduction number $R_0$ varies over the intervention plane. $R_0$ attains its highest values when both the hunting rate $\alpha_1$ and the decontamination rate $\alpha_2$ are small, and it declines monotonically as either of these rates increases. Mechanistically, hunting shortens the effective infectious period via the removal term $a=\gamma+\nu\alpha_1/(\nu-\mu)$ while at the same time reducing the disease-free host density $S_2^{*}$. In parallel, decontamination accelerates the removal of pathogens from the environment, as described by the term $\epsilon+\alpha_2$. Across most of the parameter space, the environmental transmission route dominates, $R_{\mathrm{env}}\gg R_{\mathrm{dir}}$ (see Figure \ref{fig:R0} (b)), because the slow intrinsic disease dynamics (large $1/\gamma$) render direct transmission relatively inefficient. As a result, the $R_0=1$ contour is much more responsive to changes in $\alpha_2$ than it would be for a pathogen transmitted solely by direct contact. Moreover, for intermediate levels of hunting, decontamination acts as a more powerful control lever for driving the system below the epizootic threshold. The contour intersects the demographic threshold $\alpha_1=\nu-\mu$ tangentially as $S_2^{*}$ diminishes, since once hunting alone extirpates the host population, $R_{\mathrm{env}}$ collapses and the epizootic threshold becomes irrelevant.

\begin{figure}[H]
    \centering
    \subfigure[]{\includegraphics[width=7cm]{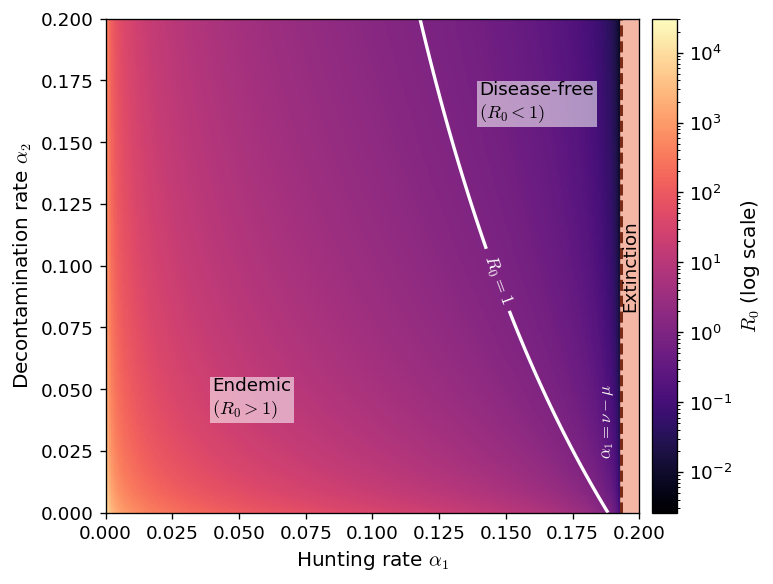}}
    \subfigure[]{\includegraphics[width=7cm]{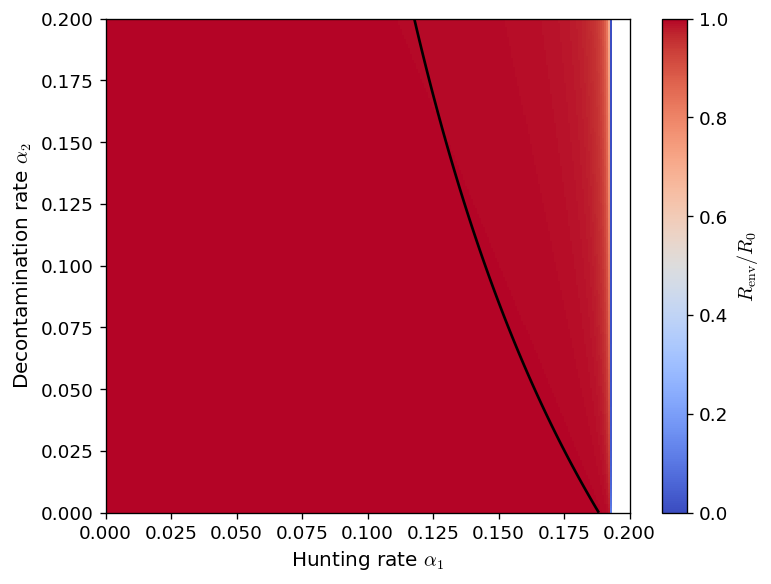}}
    \caption{Basic reproduction number $R_0$ shown on a logarithmic scale over the $(\alpha_1,\alpha_2)$ parameter plane in (a), and the relative contribution of the environmental component $R_\mathrm{env}/R_0$ in (b), including the contour where $R_0=1$. $R_0$ attains its highest values when both hunting and decontamination rates are low, and declines as either of these rates increases. Across most of the parameter space, the environmental transmission component $R_{\mathrm{env}}$ exceeds the direct transmission component $R_{\mathrm{dir}}$. The hatched region indicates host extinction, occurring once the extinction threshold $\alpha_1\geq\nu-\mu\approx0.1931$ is reached. In all of those calculations $K=1000$.}
    \label{fig:R0}
\end{figure}

Figure \ref{fig:heatmapNSIU} decomposes the endemic equilibrium into its constituent components. The total population $N^{*}$ is highest when hunting pressure is low and then declines steadily as $\alpha_1$ increases. Besides natural death, deer are removed both directly via hunting and indirectly through disease-induced mortality $\gamma I^{*}$. To the right of the $R_0=1$ contour, the infection is eliminated and $N^{*}=S^{*}=S_2^{*}$. The susceptible density $S^{*}$ is non-monotonic, achieving its maximum \emph{along} the $R_0=1$ ridge because as decontamination or hunting approaches the threshold, transmission is curtailed and individuals who would otherwise have been infected instead remain in the susceptible class. Once past this threshold, $S^{*}$ simply mirrors the decline in $S_2^{*}$. The infected density $I^{*}$ is greatest when both control measures are weak and drops to zero upon crossing the $R_0=1$ boundary, consistent with the analytically derived forward bifurcation (below). The environmental contamination level $U^{*}=\xi I^{*}/(\epsilon+ \alpha_2)$ is maximized in the absence of decontamination and is much more sensitive to changes in $\alpha_2$ than in $\alpha_1$. This shows that even relatively low decontamination rates markedly reduce environmental contamination, whereas hunting only mildly and indirectly reduces $U^{*}$ by lowering the density of infectious shedders. Overall, these panels suggest that combining a high decontamination rate with a moderate hunting rate most effectively suppresses infection and environmental contamination while maintaining a sustainable cervid population.
\begin{figure}[H]
   \centering
    \subfigure[]{\includegraphics[width=7cm]{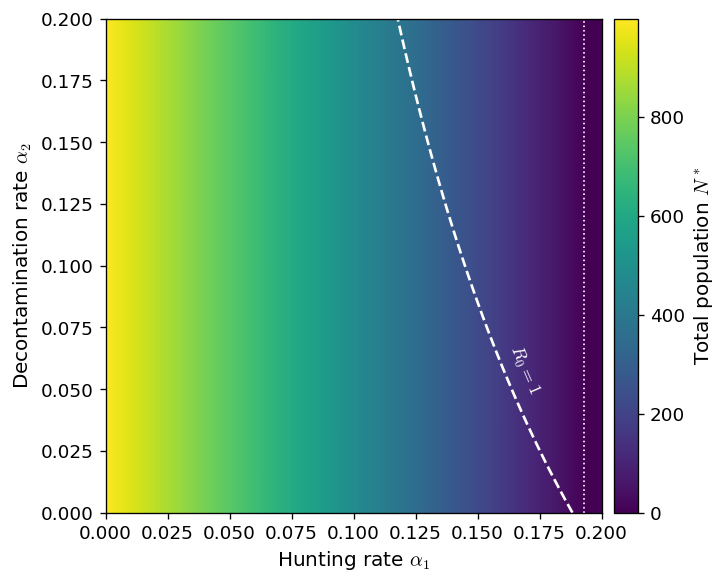}}
    \subfigure[]{\includegraphics[width=7cm]{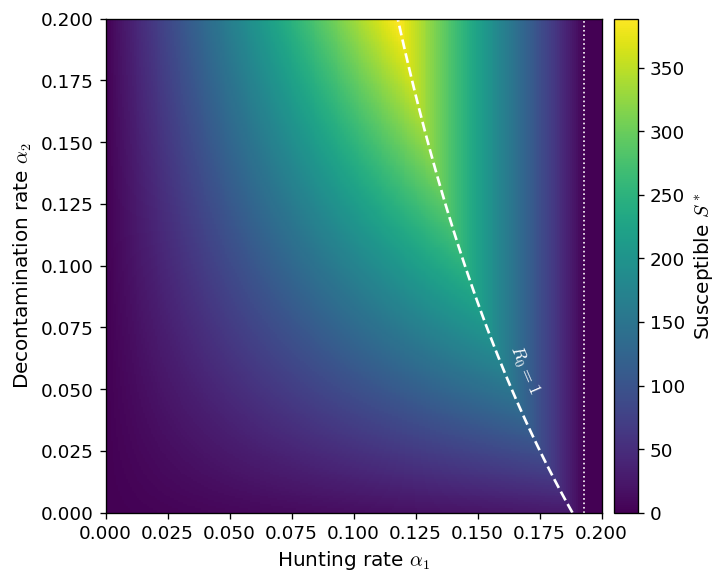}}
    \subfigure[]{\includegraphics[width=7cm]{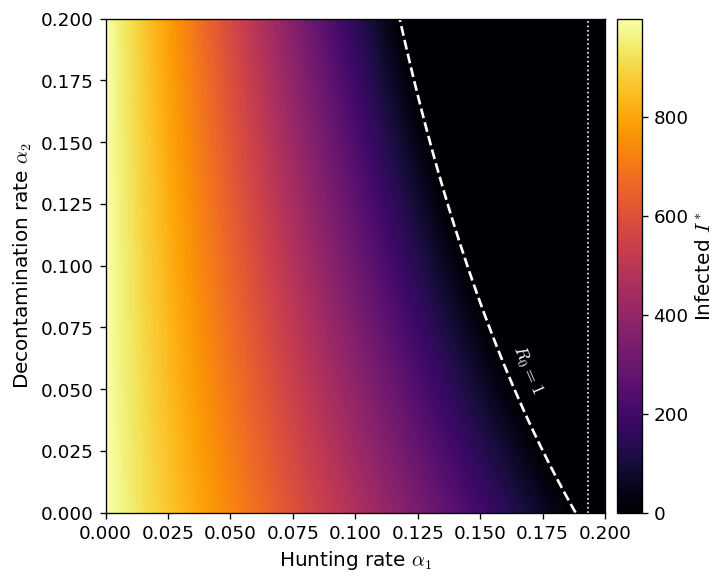}}
    \subfigure[]{\includegraphics[width=7cm]{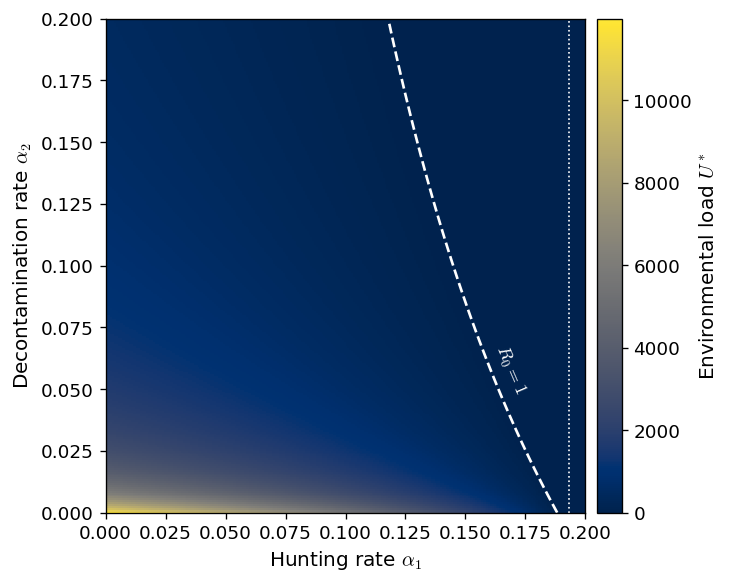}}
    \caption{Endemic equilibrium states under intervention strategies $(\alpha_1,\alpha_2)$ for (a) total population size $N^{*}$; (b) susceptible individuals $S^{*}$; (c) infected individuals $I^{*}$; (d) environmental pathogen concentration $U^{*}$, as determined by the equilibrium expressions \eqref{NUequation}. The white dashed curve corresponds to $R_0=1$. To the right of that curve, the infection dies out, so that $I^{*}=0$ and $N^{*}=S^{*}=S_2^{*}$. The dotted vertical line marks the host-extinction boundary $\alpha_1=\nu-\mu$. In all of those calculations $K=1000$.}
    \label{fig:heatmapNSIU}
\end{figure}

\paragraph{\textbf{Transcritical bifurcation.}}
At $R_0=1$ we have $C=0$, so equation \eqref{Pequation} simplifies to $P^{*}(AP^{*}+B)=0$, which yields the solutions $P^{*}=0$ and $P^{*}=-B/A$. The solution $P^{*}=0$ corresponds to the disease-free equilibrium $\mathcal{E}_2$, meaning that the endemic equilibrium branch emerges from $\mathcal{E}_2$ precisely at $R_0=1$, characteristic of a transcritical bifurcation. A Taylor expansion of the biologically relevant root near this threshold gives
$$
P^{*}=\frac{a\,(R_0-1)}{-B}+O\big((R_0-1)^{2}\big).
$$
Since $B<0$ whenever $\mathcal{E}_2$ exists, this branch is forward (supercritical). Together with $f(1)<0$, this rules out the occurrence of a backward bifurcation. Therefore, $\mathcal{E}_2$ and the endemic equilibrium exchange stability at $R_0=1$, with the disease-free equilibrium $\mathcal{E}_2$ being stable for $R_0<1$ and the endemic equilibrium being stable for $R_0>1$.

Figure \ref{fig:bifurcation} shows the endemic prevalence $P^{*}=I^{*}/N^{*}$ along a horizontal slice at fixed $\alpha_2=0.05$. As the hunting rate increases, the system sequentially traverses three distinct regimes. First, there is an endemic regime in which $P^{*}>0$ and the disease-free equilibrium $\mathcal{E}_2$ is unstable. Next comes a disease-free regime, which appears once the system crosses $R_0=1$; in this regime, $P^{*}=0$ and $\mathcal{E}_2$ is stable. Finally, at $\alpha_1=\nu-\mu\approx0.1931$, the system enters a population-extinction regime. The endemic branch bifurcates continuously from $\mathcal{E}_2$ at $R_0=1$, and its prevalence increases as $\alpha_1$ decreases. This pattern corresponds to a forward (supercritical) transcritical bifurcation, featuring a clean exchange of stability and no backward branch. In practical terms, hunting can drive the infection to extinction before it threatens population persistence, but only over a limited parameter range. Beyond that demographic threshold, the same control mechanism that suppresses the disease will, if intensified further, eliminate the host population itself.
\begin{figure}[H]
    \centering
    \subfigure[]{\includegraphics[width=0.49\textwidth]{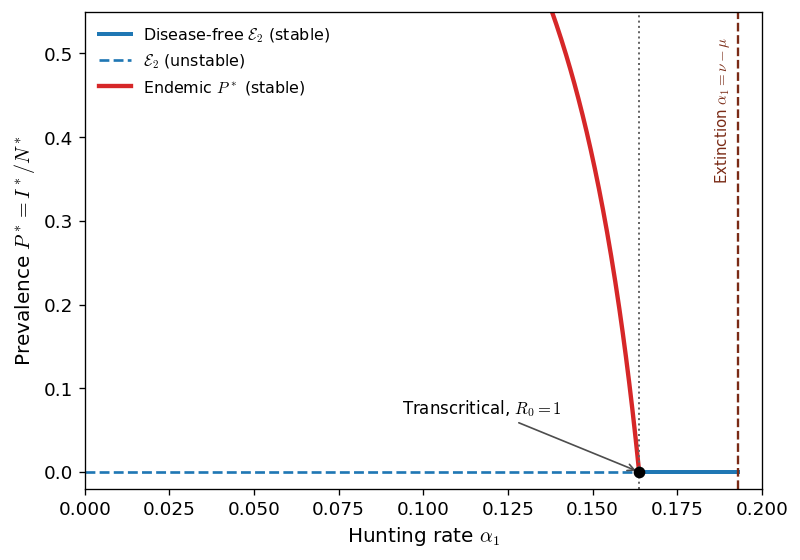}}
    \caption{Forward transcritical bifurcation in the hunting rate $\alpha_1$ with $\alpha_2$ fixed at $0.05$. The endemic prevalence $P^{*}$ (solid line, stable) branches from the disease-free equilibrium $\mathcal{E}_2$ at $R_0=1$; $\mathcal{E}_2$ is unstable (dashed line) for $R_0>1$ and stable (solid line) for $R_0<1$, up to the point of host extinction at $\alpha_1=\nu-\mu$. In all of those calculations $K=1000$.}
    \label{fig:bifurcation}
\end{figure}

\subsection{The CWD Stochastic Model.} We follow by studying the well-posedness and stability of the stochastic version of the CWD model in \eqref{M1}, rewritten as 
\begin{subequations}\label{eq:reflected}
\begin{align}
\dif S &= f_{S}\,\dif t
        + \sigma_{1} N\,\dif W_{1}
        + \sigma_{2} S\,\dif W_{2}
        + \dif L, \label{eq:reflected-S}\\
\dif I &= f_{I}\,\dif t
        + \sigma_{2} I\,\dif W_{2}, \label{eq:reflected-I}\\
\dif U &= f_{U}\,\dif t
        + \sigma_{3} U\,\dif W_{3}, \label{eq:reflected-U}
\end{align}
\end{subequations}
with drifts
\begin{align*}
f_{S} &= \nu N\left(1-\frac{N}{K}\right)
        - \frac{\beta S I}{N} - \beta_{e} S U
        - \mu S\left(1-\frac{N}{K}\right) - \alpha_{1} S,\\
f_{I} &= \frac{\beta S I}{N} + \beta_{e} S U
        - \mu I\left(1-\frac{N}{K}\right) - \gamma I - \alpha_{1} I,\\
f_{U} &= \xi I - (\epsilon+\alpha_{2}) U.
\end{align*}

\paragraph{\textbf{Local existence, uniqueness, and positivity.}}
For each $m \geq 1$, define the stopping times
$$ \tau_m = \inf\left\{ t \ge 0 : N_t \le m^{-1} \ \text{or}\ \|(S_t, I_t, U_t)\| \ge m \right\}, $$
with the convention that $\inf \varnothing = \infty$. The explosion time is then defined by
$$ \tau_e := \lim_{m \to \infty} \tau_m.$$
We say that a solution is global and nonexplosive if $\tau_e = \infty$ almost surely. In what follows, we assume that the maximal reflected solution is nonexplosive, that is, $\tau_e = \infty$ almost surely. Equivalently, we take as a standing hypothesis that the full $(S,I,U)$ system does not explode, and every subsequent statement is to be interpreted on $[0,\tau_e)$ or, after localization, on $[0,\tau_m]$. 

\begin{theorem}[Local reflected well-posedness]\label{thm:wellposed}
Under Assumption \ref{ass:params}, for every admissible initial state, there exists a pathwise-unique maximal strong solution $(S,I,U,L)$ of \eqref{eq:reflected}, defined on a stochastic interval $[0,\tau_e)$. Almost surely, for every $t<\tau_e$,
$$S_t\geq0,\qquad I_t>0,\qquad U_t>0,\qquad N_t>0.$$
\end{theorem}
\begin{proof} First, we prove the pathwise-unique maximal strong solution of $(S,I,U,L)$. Let $\mathcal D = \{(s,i,u)\in\mathbb R^3 : s\geq 0, i>0,u>0\}$. The drift vector and the diffusion matrix are locally Lipschitz on $\mathcal D$. This is true since all of the coefficients are polynomials, except for the incidence term $\beta SI/N=\beta I - \frac{\beta I^2}{N}$. But it is still a smooth function on the open set $\{N>0\}$. Hence, the coefficients and their first derivatives are bounded on any compact subset of $\mathcal D$. We then select a smooth cutoff function supported away from $\{N=0\}$ so that the truncated coefficients coincide with the original ones on
$$
\left\{N\geq m^{-1},\ \|(S,I,U)\|\leq m\right\}
$$
for every $m \ge 1$, while being globally Lipschitz on the reflecting half-space
$$
H = \{(s,i,u)\in\mathbb R^3 : s\ge0\}.
$$
Because $H$ is a convex set whose Skorokhod map with normal reflection is Lipschitz, the classical Tanaka–Lions–Sznitman theory guarantees both pathwise uniqueness and the existence of a strong solution for each truncated reflected SDE; see \cite[Theorem 5.1]{Saisho1987}. Pathwise uniqueness implies that any two truncated solutions agree up to their respective exit times from the compact sets. Hence, we concatenate these solutions to obtain a pathwise unique maximal strong solution $(S,I,U,L)$ on the interval $[0,\tau_e)$; see \cite{Tanaka1979,LionsSznitman1984}. A similar reasoning applies to the regulator process $L_t$, due to its assumptions in \eqref{eq:localtime}.

Next, we establish the strict positivity of $I$ and $U$, since the bound $S_t\ge 0$ is enforced by reflection and then $N_t=S_t+I_t>0$ follows once $I_t>0$ is proven. Both $I$ and $U$ satisfy linear equations with a \emph{nonnegative} source, so we use the method of variation of constants. On $[0,\tau_m]$ let
$$
b^{I}_t=\frac{\beta S_t}{N_t}-\mu\left(1-\frac{N_t}{K}\right)-\gamma-\alpha_1,
\qquad
b^{U}_t=-(\epsilon+\alpha_2).
$$
The process $b^{I}$ is bounded since $0\le S_t/N_t\le 1$, $\|(S,I,U)\|\leq m$ and $N_t\ge m^{-1}$. Introduce the strictly positive stochastic exponentials
$$
\Phi_t=\exp\left(\int_0^{t}\Big(b^{I}_s-\frac12\sigma_2^2\Big)\,ds
      +\sigma_2 W_{2,t}\right),
\qquad
\Psi_t=\exp\left(-\Big(\epsilon+\alpha_2+\frac12\sigma_3^2\Big)t
      +\sigma_3 W_{3,t}\right),
$$
which solve $d\Phi=\Phi\,b^{I}\,dt+\sigma_2\Phi\,dW_2$ and $d\Psi=\Psi\,b^{U}\,dt+\sigma_3\Psi\,dW_3$, respectively. Applying It\^o's product rule to $\Phi^{-1}I$ and $\Psi^{-1}U$ removes the linear and martingale parts and gives
$$
I_t=\Phi_t\left(I_0+\beta_e \int_0^{t}\Phi_s^{-1}\, S_s U_s\,ds\right),
\qquad
U_t=\Psi_t\left(U_0+\xi \int_0^{t}\Psi_s^{-1}\, I_s\,ds\right).
$$
for $t<\tau_e$. Let $\tau_0=\inf\{t\ge 0:\min(I_t,U_t)\le 0\}$. On $[0,\tau_0)$, we have $S_t,I_t,U_t\ge 0$, so the sources $\beta_e S_sU_s$ and $\xi I_s$ are nonnegative. In that case, $I_t\ge\Phi_tI_0>0$ and $U_t\ge\Psi_tU_0>0$ and by continuity these persist at $t=\tau_0$. Therefore $\tau_0\ge\tau_e$ implies that $I_t>0$, $U_t>0$, and hence $N_t>0$ for all $t<\tau_e$.
\end{proof}

The disease-free analysis presented below does not depend on the assumption that $\tau_e = \infty$ almost surely. In the disease-free case, the dynamics collapse to the one-dimensional stochastic logistic diffusion \eqref{eq:logistic}, whose explicit solution remains strictly positive and nonexplosive for all finite times, as we will see next.

\paragraph{\textbf{Disease-free measure and invariance of the threshold.}}
On the disease-free set
$$\mathcal F_{+}=\{(S,I,U):S>0,\ I=U=0\},$$
the processes $I_t\equiv0$ and $U_t\equiv0$ satisfy their respective stochastic equations. By pathwise uniqueness, this implies that $\mathcal F_{+}$ is invariant. Restricted to $\mathcal F_{+}$, we have $N=S$, and the model reduce to population dynamics given by
\begin{equation}\label{eq:logistic}
\dif N = N(r-cN)\,\dif t+\sigma_NN\,\dif W,
\end{equation}
where $r=(\nu-\mu)-\alpha_1$, $c=\frac{\nu-\mu}{K}$, $\sigma_N^2=\sigma_1^2+\sigma_2^2$ and $W_t=\frac{\sigma_1W_{1,t}+\sigma_2W_{2,t}}{\sigma_N}$, with $\sigma_N>0$. Since $W_1$ and $W_2$ are independent, the resulting process $W$ is a standard Brownian motion.

For $N_0>0$, the solution to \eqref{eq:logistic} is given by
$$N_t= \frac{ N_0\exp\left[\left(r-\frac12\sigma_N^2\right)t+\sigma_NW_t
\right]}{1+cN_0\displaystyle\int_0^t\exp\left[\left(r-\frac12\sigma_N^2\right)s+\sigma_NW_s \right]\dif s},$$
which shows that $N_t=S_t>0$ for every finite $t\ge0$ almost surely.
Hence, the reflecting boundary $\{S=0\}$ is never hit in finite time, and the regulator satisfies $L_t\equiv0$ on $\mathcal F_{+}$.

\begin{lemma}\label{lem:dfe}
Assume $c>0$, $\sigma_N>0$, and $N_0>0$. If $r>\frac12\sigma_N^2$, then the stochastic logistic diffusion \eqref{eq:logistic} is positive recurrent on $(0,\infty)$ and admits a unique invariant probability distribution $\pi_0$ given by the $\mathrm{Gamma}(k,\theta)$ density
$$\pi_0(n)=\frac{1}{\Gamma(k)\theta^k} n^{k-1}e^{-n/\theta},$$
for $n>0$, where $k=\frac{2r}{\sigma_N^2}-1>0$ and $\theta=\frac{\sigma_N^2}{2c}$. In that case, $\overline N=\mathbb E_{\pi_0}[N]=K\left(1-\frac{\alpha_1+\frac12\sigma_N^2}{\nu-\mu}\right)$, and $\operatorname{Var}_{\pi_0}(N)=\overline N\,\frac{\sigma_N^2}{2c}$. Moreover, $\pi_0$ is ergodic.
\end{lemma}
\begin{proof}
Write the diffusion equation \eqref{eq:logistic} as $dN = b(N)\,dt + \sigma(N)\,dW$ with drift $b(n) = n(r - c n)$ and diffusion coefficient $\sigma(n) = \sigma_N n$. The corresponding scale density, up to a positive multiplicative constant, is
$$
s'(n)
= \exp\left(-\int_1^n \frac{2b(y)}{\sigma^2(y)}\,dy\right)
= n^{-2r/\sigma_N^2} e^{2c n/\sigma_N^2}.
$$
Thus, the speed density is
$$
m(n) = \frac{2}{\sigma^2(n) s'(n)} \propto n^{2r/\sigma_N^2 - 2} e^{-2c n/\sigma_N^2}.
$$
Because $c > 0$, this speed density is integrable at infinity. It is integrable at zero when $\frac{2r}{\sigma_N^2} - 2 > -1$ which is equivalent to $r > \frac{1}{2}\sigma_N^2$. Under this condition, both boundary scale integrals diverge, while the total speed measure is finite. Therefore, the diffusion is positive recurrent, and its normalized speed measure serves as its unique invariant probability distribution \cite{karlin1981second}. Ergodicity then follows from positive recurrence together with the uniqueness of this invariant distribution
\cite{khasminskii2012stochastic}.
\end{proof}

\begin{lemma}\label{cor:dfe-invariance}
Assume $c>0$, $\sigma_N>0$, and $N_0>0$. Reflection does not modify the restriction of \eqref{eq:reflected} to the disease-free region $\mathcal F_+=\{S>0,\ I=U=0\}$. In particular, it leaves unchanged the disease-free invariant measure $\pi_0$, its mean $\overline N$, and the transverse infection variational equation along the disease-free trajectory. Consequently, the following derived quantities, built from $\pi_0$ and the formal linearization, coincide between the reflected and unreflected systems
\end{lemma}
\begin{proof}
On $\mathcal F_+$ we have $I_t=U_t=0$ and $N_t=S_t>0$ at all finite times $t$, with $L_t\equiv0$. Thus, the reflected and unreflected systems coincide when restricted to $\mathcal F_+$, and hence share the same invariant measure $\pi_0$ and the same mean $\overline N$. A transverse linearization of the $(I,U)$ subsystem along the disease-free trajectory yields
\begin{align}\label{traverse}
dJ
&=
\left[
\left(
\beta-\mu\left(1-\frac{N_t}{K}\right)
-\gamma-\alpha_1
\right)J
+\beta_eN_tV
\right]dt
+\sigma_2J\,dW_2, \nonumber \\
& \\
dV
&=\left(\xi J-(\epsilon+\alpha_2)V \right)\,dt+\sigma_3V\,dW_3. \nonumber
\end{align}
No reflection term appears in this variational system. Therefore, the transverse variational cocycle, and hence its top Lyapunov exponent, whenever it exists, is the same for both the reflected and unreflected models. The derived quantities are likewise unchanged, as they are derived from the same disease-free mean and removal rates.
\end{proof}

In summary, reflection preserves both the disease-free dynamics and the associated transverse infection linearization, while additionally ensuring that $S_t\geq 0$ holds for all times in the maximal lifetime of the reflected solution. 

The following theorem establishes the condition for disease invasion.

\begin{theorem}\label{thm:main}
Assume that $r>\frac12\sigma_N^2$ and that the projective process $(N,Z)$, with $Z=V/J$, is positive Harris recurrent and possesses a unique invariant probability measure $\rho$ on $(0,\infty)^2$ satisfying the moment conditions
\begin{equation}\label{eq:rho-moments}
\int nz\,\rho(dn,dz)<\infty
\qquad\text{and}\qquad
\int z^{-1}\,\rho(dn,dz)<\infty.
\end{equation}
Then, the top Lyapunov exponent associated with the transverse infection linearization is
$$
\lambda
=
\beta-a_s-\frac12\sigma_2^2
+\beta_e\int nz\,\rho(dn,dz),
$$
where
$$a_s = \mu\frac{\alpha_1+\frac12\sigma_N^2}{\nu-\mu}+\gamma+\alpha_1.
$$
Consequently, the disease-free equilibrium is linearly resistant to invasion if $\lambda<0$, and it is linearly susceptible to invasion if $\lambda>0$.
\end{theorem}

\begin{proof}
Let $c_1(N):=\beta-\mu\left(1-\frac{N}{K}\right)-\gamma-\alpha_1$ and $b:=\epsilon+\alpha_2$, so the transverse system \eqref{traverse} reads
$$
dJ=\big[c_1(N)J+\beta_e NV\big]dt+\sigma_2 J\,dW_2,
\qquad
dV=(\xi J-bV)\,dt+\sigma_3 V\,dW_3 .
$$

First, we show invariance of the positive cone of $Z=V/J$. Since the off-diagonal drift couplings $\beta_e N\ge0$ and $\xi\ge0$, the system is cooperative. By the variation of constants representation used in the positivity part of Theorem \ref{thm:wellposed} and the sources $\beta_e N_sV_s$ and $\xi J_s$ are nonnegative on the cone, any initial condition with $J_0>0$, $V_0>0$ yields $J_t>0$ and $V_t>0$ for all $t$. Hence, the process $Z_t:=V_t/J_t\in(0,\infty)$ is well defined for all $t$.

Second, we derive the diffusion equation of $Z=V/J$. By It\^o's quotient rule, using that $W_2$ and $W_3$ are independent 
$$
dZ=\big[\xi-bZ-c_1(N)Z-\beta_e NZ^2+\sigma_2^2 Z\big]dt
   +\sigma_3 Z\,dW_3-\sigma_2 Z\,dW_2 .
$$
Thus $(N,Z)$ is a Markov diffusion on $(0,\infty)^2$, matching the hypothesis. Since $N$ evolves autonomously through the disease-free logistic \eqref{eq:logistic}, the $N$-marginal of any invariant law of $(N,Z)$ is the measure $\pi_0$ of Lemma \ref{lem:dfe}; that is, the $N$-marginal of $\rho$ is $\pi_0$.

Third, we derive the diffusion equation of $\log J$. Applying It\^o formula to $\log J$ and using $V/J=Z$,
$$
d\log J=\Big[c_1(N)+\beta_e NZ-\frac12\sigma_2^2\Big]dt+\sigma_2\,dW_2
=:g(N,Z)\,dt+\sigma_2\,dW_2 .
$$
Integrating,
$$
\frac1t\log J_t
=\frac1t\log J_0
+\frac1t\int_0^t g(N_s,Z_s)\,ds
+\frac{\sigma_2}{t}W_{2,t}.
$$

Fourth, we use ergodic averaging for the integral of $g$. The function $g$ is affine in $n$ and linear in the product $nz$; the hypothesis $\int nz\,\rho(dn,dz)<\infty$, together with the finite mean of $\pi_0$, gives $g\in L^1(\rho)$. Positive Harris recurrence and uniqueness of $\rho$ yield the ergodic theorem
$$
\lim_{t\to\infty} \frac1t\int_0^t g(N_s,Z_s)\,ds
=\int g\,d\rho
\qquad\text{a.s.},
$$
while $\lim_t \frac{\sigma_2}{t}W_{2,t}=0$ a.s. by the law of large numbers for Brownian motion. Hence $\lim_t \frac1t\log J_t=\int g\,d\rho$ almost surely.

Splitting $\int g\,d\rho$ and using that the $N$-marginal of $\rho$ is $\pi_0$,
$$
\int g\,d\rho
=\int c_1(n)\,\pi_0(dn)
+\beta_e\int nz\,\rho(dn,dz)
-\frac12\sigma_2^2 .
$$
By Lemma \ref{lem:dfe}, $\overline N=\mathbb E_{\pi_0}[N]
=K\left(1-\frac{\alpha_1+\frac12\sigma_N^2}{\nu-\mu}\right)$, so
$1-\overline N/K=\frac{\alpha_1+\frac12\sigma_N^2}{\nu-\mu}$ and
$$
\int c_1(n)\,\pi_0(dn)
=\beta-\mu\left(1-\frac{\overline N}{K}\right)-\gamma-\alpha_1
=\beta-\mu\,\frac{\alpha_1+\frac12\sigma_N^2}{\nu-\mu}-\gamma-\alpha_1
=\beta-a_s.
$$
Therefore
$$
\int g\,d\rho=\beta-a_s-\frac12\sigma_2^2+\beta_e\int nz\,\rho(dn,dz)=\lambda .
$$

Finally, we derive the Lyapunov exponent. From $\|(J_t,V_t)\|=\sqrt{J_t^2+V_t^2}=J_t\sqrt{1+Z_t^2}$, we have
\begin{equation}\label{eq:norm-projective}
\frac{1}{t}\log\|(J_t,V_t)\|
=
\frac{1}{t}\log J_t
+
\frac{1}{2t}\log\big(1+Z_t^2\big).
\end{equation}
Hence, given the steps above, it is enough to show that
\begin{equation}\label{eq:projective-subexp}
\lim_t \frac{1}{t}\log\big(1+Z_t^2\big)= 0
\qquad\text{a.s.}
\end{equation}
Since $1+Z_t^2\ge 1$, then $0\le\log(1+z^2)\le\log 2+2\max\{\log z,0\}$ for every $z>0$. Thus, the claim in equation \eqref{eq:projective-subexp} reduces to
proving
\begin{equation}\label{eq:logZ-sublinear}
\lim_t \frac{1}{t}\log Z_t= 0
\qquad\text{a.s.}
\end{equation}

Since $Z_t>0$, It\^o's formula applied to $\log Z_t$, gives
\begin{equation}\label{eq:logZ}
d\log Z_t
=
h(N_t,Z_t)\,dt
+\sigma_3\,dW_{3,t}-\sigma_2\,dW_{2,t},
\end{equation} where
$$
h(n,z)
=
\frac{\xi}{z}-b-c_1(n)-\beta_e nz+\frac12\sigma_2^2-\frac12\sigma_3^2.
$$
The function $h$ is affine in $n$, linear in the product $nz$, and linear in $z^{-1}$. Since the $N$-marginal of $\rho$ is $\pi_0$ with $\int n\,\rho(dn,dz)=\overline N<\infty$ by Lemma \ref{lem:dfe}, the two moment conditions in \eqref{eq:rho-moments} give $h\in L^1(\rho)$.

Integrating \eqref{eq:logZ} and dividing by $t$,
$$
\frac{1}{t}\log Z_t
=
\frac{1}{t}\log Z_0
+\frac{1}{t}\int_0^t h(N_s,Z_s)\,ds
+\frac{\sigma_3W_{3,t}-\sigma_2W_{2,t}}{t}.
$$
The martingale term vanishes a.s. by the strong law of large numbers for Brownian motion. By the ergodic theorem and positive Harris recurrence with the uniqueness of $\rho$, we have
\begin{equation}\label{eq:logZ-kappa}
\lim_t\frac{1}{t}\log Z_t=\int h\,d\rho =:\kappa
\qquad\text{a.s.}
\end{equation}

Since $\rho$ is a probability measure on $(0,\infty)^2$, then there exists an $m_0$ such that the set compact $C=[m_0^{-1},m_0]^2\subset(0,\infty)^2$ has $\rho(C)>0$. The ergodic theorem applied to the bounded function $\mathbf 1_C$ gives
\begin{equation}\label{eq:compact-occupation}
\lim_t \frac{1}{t}\int_0^t\mathbf 1_C(N_s,Z_s)\,ds=\rho(C)>0
\qquad\text{a.s.}
\end{equation}
On one hand, if $\kappa>0$, then \eqref{eq:logZ-kappa} forces $\lim_t Z_t=\infty$, so eventually $Z_t>m_0$ and $(N_t,Z_t)\notin C$. On the other hand, if $\kappa<0$, then $\lim_tZ_t= 0$, so eventually $Z_t<m_0^{-1}$ and again $(N_t,Z_t)\notin C$. In either case the occupation fraction in \eqref{eq:compact-occupation} would tend to $0$, contradicting $\rho(C)>0$. Hence $\kappa=0$, which establishes \eqref{eq:logZ-sublinear} and therefore \eqref{eq:projective-subexp}.

Substituting \eqref{eq:projective-subexp} into \eqref{eq:norm-projective}, we have
$$
\lambda
:=
\lim_{t\to\infty}\frac{1}{t}\log\|(J_t,V_t)\|
=
\lim_{t\to\infty}\frac{1}{t}\log J_t
=
\beta-a_s-\frac12\sigma_2^2
+\beta_e\int nz\,\rho(dn,dz)
\qquad\text{a.s.}
$$
The limit does not depend on the initial transverse direction $(J_0,V_0)$ in the open cone, so $\lambda$ is the top Lyapunov exponent of the transverse cocycle. Therefore, the results of $\lambda<0$ versus $\lambda>0$ follow at once.
\end{proof}

The following corollary gives a lower bound for $\mathbb E_\rho(NZ)=\int nz\,\rho(dn,dz)$ and for the Lyapunov exponent.

\begin{corollary}
\label{cor:lambda-lower-bound}
Under the assumptions of Theorem \ref{thm:main}, assume additionally that
\begin{equation}\label{eq:aY-positive}
a_Y
:=
\nu-\beta+\gamma-\epsilon-\alpha_2
-\frac{\nu}{K}\overline N
-\frac12(\sigma_1^2+\sigma_3^2)
>0.
\end{equation}
Then
\begin{equation}\label{eq:nz-lower-bound}
\beta_e\int_{(0,\infty)^2} nz\,\rho(dn,dz)
\geq a_Y,
\end{equation}
and consequently the transverse invasion exponent satisfies
\begin{equation}\label{eq:lambda-lower-bound}
\lambda
\geq
-\left(
\epsilon+\alpha_2+\frac12\sigma_3^2
\right).
\end{equation}
\end{corollary}

\begin{proof}
Using It\^o's product rule for $Y_t:=N_tZ_t$, we get
$$
dY_t=\left[\xi N_t+\left(\nu-\beta+\gamma-\epsilon-\alpha_2
-\frac{\nu}{K}N_t\right)Y_t-\beta_eY_t^2\right]dt
+\sigma_1Y_t\,dW_{1,t}+\sigma_3Y_t\,dW_{3,t}.
$$
Since $Y_t>0$, It\^o's formula gives
$$
d\log Y_t=\left[\frac{\xi}{Z_t}+\nu-\beta+\gamma-\epsilon-\alpha_2-\frac{\nu}{K}N_t-\beta_eY_t-\frac12(\sigma_1^2+\sigma_3^2)
\right]dt+\sigma_1\,dW_{1,t}+\sigma_3\,dW_{3,t}.
$$
Because $\xi/Z_t>0$, integration gives
$$
\log Y_t \geq a(t)-\beta_e\int_0^tY_s\,ds,
$$
where
$$
a(t)=\log Y_0+\left[\nu-\beta+\gamma-\epsilon-\alpha_2
-\frac12(\sigma_1^2+\sigma_3^2)\right]t
-\frac{\nu}{K}\int_0^tN_s\,ds+\sigma_1W_{1,t}+\sigma_3W_{3,t}.
$$
By ergodicity of the stochastic logistic process $N_t$,
$$
\lim_t \frac1t\int_0^tN_s\,ds
=
\overline N
\qquad\text{a.s.},
$$
while $W_{1,t}/t\to0$ and $W_{3,t}/t\to0$ almost surely. Hence,
$$
\lim_t \frac{a(t)}{t}= a_Y
\qquad\text{a.s.}
$$
Since $a_Y>0$, Lemma \ref{Ap1}, applied with
$f(t)=Y_t$ and $b=\beta_e$, we get
$$
\liminf_{t\to\infty}
\frac1t\int_0^tY_s\,ds
\geq
\frac{a_Y}{\beta_e}
\qquad\text{a.s.}
$$
But since, $Y_t=N_tZ_t$, and the ergodic theorem for
$(N,Z)$ gives
$$
\lim_t \frac1t\int_0^tY_s\,ds
=\lim_t \frac1t\int_0^tN_sZ_s\,ds= \int_{(0,\infty)^2}nz\,\rho(dn,dz)
\qquad\text{a.s.}
$$
Therefore,
$$
\beta_e\int_{(0,\infty)^2}nz\,\rho(dn,dz)
\geq a_Y.
$$
Also, the Lyapunov exponent
$$
\lambda=\beta-a_s-\frac12\sigma_2^2+\beta_e\int_{(0,\infty)^2}nz\,\rho(dn,dz) \geq
\beta-a_s-\frac12\sigma_2^2+a_Y.
$$
Substituting
$$
a_s
=
\mu\frac{\alpha_1+\frac12\sigma_N^2}{\nu-\mu}
+\gamma+\alpha_1,
\qquad
\overline N
=
K\left(
1-\frac{\alpha_1+\frac12\sigma_N^2}{\nu-\mu}
\right),
$$
and using
$$
\sigma_N^2=\sigma_1^2+\sigma_2^2,
$$
after simplification, we get
$$
\lambda
\geq
-\epsilon-\alpha_2-\frac12\sigma_3^2.
$$
This proves the result.
\end{proof}

\begin{remark}\label{rem:environmental-noise-bound}
The lower bound
$$
\lambda
\geq
-\left(\epsilon+\alpha_2+\frac12\sigma_3^2\right)
$$
has a direct interpretation in terms of the environmental component of the transverse system. Indeed, since $V_t>0$,
$$
dV_t
=
\left(\xi J_t-(\epsilon+\alpha_2)V_t\right)dt
+\sigma_3V_t\,dW_{3,t},
$$
and It\^o's formula gives
$$
d\log V_t
=
\left[
\xi\frac{J_t}{V_t}
-(\epsilon+\alpha_2)
-\frac12\sigma_3^2
\right]dt
+\sigma_3\,dW_{3,t}.
$$
Because $J_t/V_t>0$,
$$
d\log V_t
\geq
-\left(\epsilon+\alpha_2+\frac12\sigma_3^2\right)dt
+\sigma_3\,dW_{3,t}.
$$
Consequently,
$$
\liminf_{t\to\infty}\frac1t\log V_t
\geq
-\left(\epsilon+\alpha_2+\frac12\sigma_3^2\right)
\qquad\text{a.s.}
$$
Since $V_t=J_tZ_t$ and
$$
\lim_t \frac1t\log Z_t=0
\qquad\text{a.s.},
$$
the components $J_t$ and $V_t$ have the same asymptotic exponential
growth rate. Hence,
$$
\lambda
\geq
-\left(\epsilon+\alpha_2+\frac12\sigma_3^2\right).
$$

The quantity
$$
\epsilon+\alpha_2+\frac12\sigma_3^2
$$
can thus be viewed as the effective almost-sure exponential decay rate of the environmental component in the absence of contributions from infected individuals. The quantity $\frac12\sigma_3^2$ represents the It\^o correction arising from multiplicative noise in the environment. Hence, larger values of $\sigma_3$ reduce the lower bound on $\lambda$, rendering the bound less stringent. 
\end{remark}

\paragraph{\textbf{The deterministic invasion exponent.}}
A deterministic analogue of the stochastic invasion condition in Theorem \ref{thm:main} is derived by linearizing the infected subsystem $(I,U)$ in \eqref{M2deterministic} around the disease-free host equilibrium $\mathcal E_2=(S_2^{*},0,0)$, where $\overline N=S_2^{*}=K\left(1-\frac{\alpha_1}{\nu-\mu}\right)$. The transverse block associated with $(I,U)$ is
\begin{equation}\label{eq:B2-transverse}
B_2=
\begin{pmatrix}
c_1 & \beta_e\,\overline N\\[2pt]
\xi & -d
\end{pmatrix},
\qquad
c_1=\beta-\mu\left(1-\frac{\overline N}{K}\right)-\gamma-\alpha_1=\beta-a,
\qquad
d=\epsilon+\alpha_2,
\end{equation}
where $a=\gamma+\alpha_1+\mu\left(1-\frac{\overline N}{K}\right)=\gamma+\frac{\nu\alpha_1}{\nu-\mu}$. The deterministic invasion exponent $\lambda_{\det}$ is the dominant eigenvalue, with the largest real part, of $B_2$. The dominant eigenvalue is actually given by
\begin{equation}\label{eq:lambda-det}
\lambda_{\det}=\frac{(c_1-d)+\sqrt{(c_1+d)^{2}+4\beta_e\xi\,\overline N}}{2} \in\mathbb R.
\end{equation}

The exponent $\lambda_{\det}$ and the reproduction number $R_0$ of \eqref{R0} give the same threshold. Using $c_1=\beta-a$, the determinant of $B_2$ is given by $\det B_2=a\,d\,(1-R_0)$, and since $a,d>0$, then
\begin{equation}\label{eq:sign-equivalence}
\operatorname{sign}\bigl(\lambda_{\det}\bigr)
=\operatorname{sign}\bigl(R_0-1\bigr).
\end{equation}

Indeed, $R_0>1$ gives $\det B_2<0$, so $B_2$ has one positive and one negative eigenvalue and $\lambda_{\det}>0$. That means that the disease-free equilibrium $\mathcal E_2$ is a saddle point, and the infection invades. Conversely, $R_0<1$ gives $\det B_2>0$, forcing $a>\beta$ and so $c_1<0$, together with $\operatorname{tr}B_2=c_1-d<0$, then both eigenvalues have negative real part. Thus, $\lambda_{\det}<0$ implies $\mathcal E_2$ is locally asymptotically stable. Therefore, $\lambda_{\det}$ is the growth rate of an infinitesimal outbreak, and $R_0$ is its per-generation multiplier. 

Lastly, $\lambda_{\det}$ is obtained as the zero-noise end of the stochastic invasion exponent $\lambda$ in Theorem \ref{thm:main} when $\sigma_1,\sigma_2,\sigma_3= 0$. In this regime, the diffusion terms disappear, with $a_s=a$ and $\frac12\sigma_2^{2}=0$. Moreover, the invariant distribution $\rho$ collapses onto the deterministic transverse fixed point, implying that $\lambda=\lambda_{\det}$ when $\sigma_1,\sigma_2,\sigma_3= 0$. Hence, $\lambda_{\det}$ serves as the deterministic reference for quantifying the noise-driven shift in the invasion threshold.

\paragraph{\textbf{The stochastic reproduction number proxy.}} By Theorem \ref{thm:main}, the transverse invasion exponent can be written in the threshold form
\begin{align}
\mathcal R_0^s:=
\frac{\beta}{a_s+\frac12\sigma_2^2}
+\frac{\beta_e \, \mathbb E_\rho(NZ)}
{a_s+\frac12\sigma_2^2},
\end{align}
in which $\mathbb E_\rho(NZ)=\int nz\,\rho(dn,dz)$ depends on the invariant law $\rho$ of the projective process $(N,Z)$. Since $a_s+\frac12\sigma_2^2>0$, then $\operatorname{sign}\big(\mathcal R_0^s-1\big)=\operatorname{sign}(\lambda)$. Hence, the disease-free equilibrium is linearly resistant to invasion when $\mathcal R_0^s<1$ and linearly vulnerable to invasion when $\mathcal R_0^s>1$. We therefore regard $\mathcal R_0^s$ as a diffusion-corrected stochastic reproduction number proxy. It shares the critical value $1$ with the deterministic reproduction number $R_0$, but it is not a quantitative continuation of $R_0$ as $\lambda$ is for $\lambda_{\det}$.

In particular, $\mathcal R_0^s$ does not reduce to $R_0$ in the noiseless case. When $\sigma_1=\sigma_2=\sigma_3= 0$, the invariant law $\rho$ concentrates at the point mass $\delta_{(\overline N,\,Z^{*})}$, where $\overline N=S_2^{*}$ and $Z^{*}=V^{*}/J^{*}$ which is the positive root of the transverse balance
\begin{equation}\label{eq:Zstar}
\beta_e\,\overline N\,(Z^{*})^{2}
+\big(\epsilon+\alpha_2+c_1(\overline N)\big)\,Z^{*}-\xi=0,
\end{equation} where
$$c_1(\overline N)=\beta-\mu\left(1-\frac{\overline N}{K}\right)-\gamma-\alpha_1,$$
and not the uncoupled value $\xi/(\epsilon+\alpha_2)$. In that case, the transverse ratio $Z^{*}$ is shifted by the infection growth $c_1$ and the environmental feedback $\beta_e\overline N$ along the marginal direction. Consequently, in the noiseless end $\mathbb E_\rho(NZ)=\overline N\,Z^{*}$ and $a_s+\frac12\sigma_2^{2}= a$, give
$$
\mathcal R_0^s\ =\
\frac{\beta+\beta_e\,\overline N\,Z^{*}}{a}\ \neq\ R_0.
$$
Nevertheless, in the absence of noise we have $\operatorname{sign}\big(\mathcal R_0^s-1\big)=\operatorname{sign}(\lambda)=\operatorname{sign}(\lambda_{\det})=\operatorname{sign}(R_0-1)$. In particular, $\mathcal R_0^s=1$ holds exactly when $R_0=1$. Hence, $R_0$ and $\mathcal R_0^s$ agree on the epizootic threshold, although they generally differ away from it.

Figure \ref{fig:lambda} contrasts the stochastic invasion threshold with the associated deterministic threshold. The exponent $\lambda$ in Theorem \ref{thm:main} is obtained by averaging $\dif\log J$ over time along a reflected Euler-Maruyama simulated path. The sign of $\lambda$ determines whether the disease-free equilibrium prevents invasion when $\lambda<0$ or is open to invasion when $\lambda>0$. Figure \ref{fig:lambda} (a) illustrates the dependence of $\lambda$ on the noise intensity $\sigma$ for three management policies spanning a range of basic reproduction numbers $R_0$. At $\sigma=0$, each curve begins at the deterministic value $\lambda_{\det}$, confirming that the numerically estimated exponent matches the deterministic exponent in the absence of noise. As $\sigma$ grows, all curves decline. This indicates that environmental fluctuations reduce the transverse invasion exponent, so higher environmental stochasticity tends to impede CWD invasion on average.

The magnitude of this reduction depends on how far the deterministic exponent sits above the threshold. With the clearly super-threshold case with $R_0=5.18$, $\lambda$ remains positive across the entire noise range, allowing invasion. In contrast, under the two near-thresholds given by $R_0=2.11$ and $R_0=2.37$, $\lambda$ drops toward the cutoff line $\lambda=0$ as $\sigma$ increases. The $95\%$ interval starts to extend across zero; see Figure \ref{fig:lambda} (a). The model demonstrates that this behavior in the transverse exponent is characteristic of the stochastic fadeout it is designed to capture. The same contrast can be seen on the intervention plane, Figure \ref{fig:lambda} (b). The solid curve marks the deterministic threshold $R_0=1$, separating the endemic zone on the left side with warmer colors) from the disease-free zone on the right side. The dashed curve denotes the stochastic threshold $\lambda=0$ at $\sigma=0.2$. It falls to the left of $R_0=1$ as, at this noise level, there is a set of interventions deemed endemic by the deterministic analysis that already lie in the non-invasion region under the stochastic regime. In line with Figure \ref{fig:lambda} (a), noise shifts the epizootic boundary toward lower control effort. This is made possible because environmental stochasticity reduces the hunting and decontamination efforts required to cross the threshold. The two thresholds nevertheless retain the same qualitative shape, both curving toward smaller $\alpha_1$ as the decontamination rate $\alpha_2$ rises, and both ending close to the host-extinction boundary $\alpha_1=\nu-\mu$.

\begin{figure}[H]
    \centering
    \subfigure[]{\includegraphics[width=7cm]{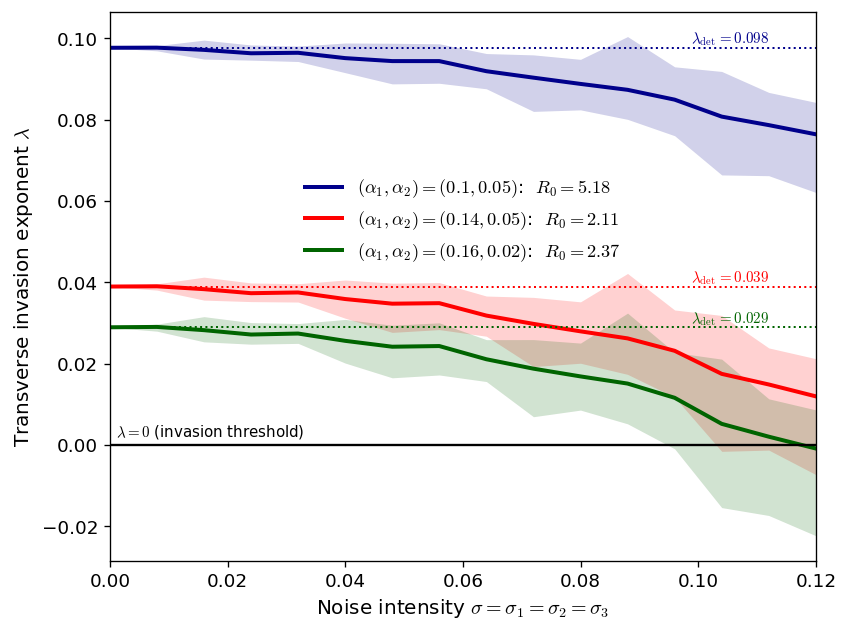}}
    \subfigure[]{\includegraphics[width=7cm]{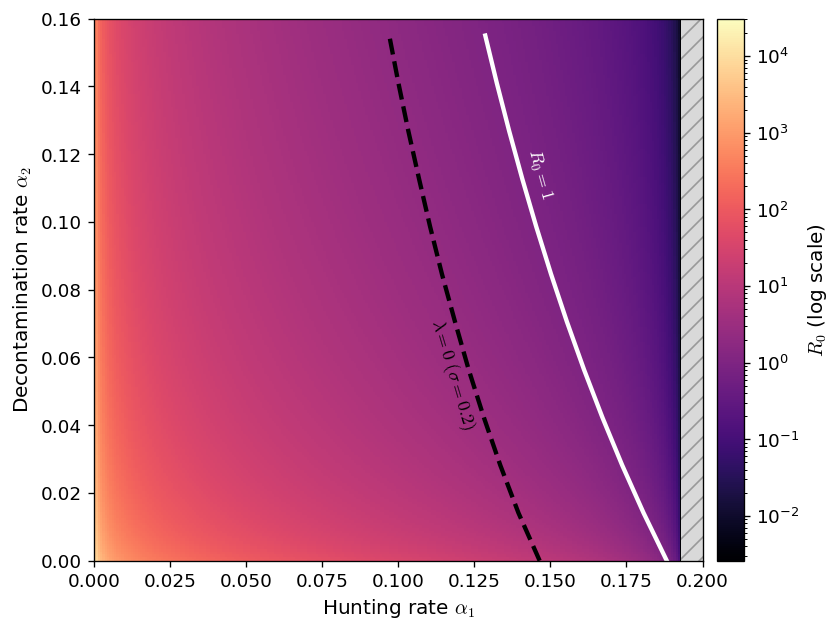}}
    \caption{Stochastic invasion threshold and its deterministic benchmark.
    (a) The transverse invasion exponent $\lambda$ as a function of the noise
    intensity $\sigma=\sigma_1=\sigma_2=\sigma_3$, for three intervention
    policies $(\alpha_1,\alpha_2)$. Solid curves are the mean and
    shaded regions the $95\%$ band over $20$ random simulations. Dotted lines
    mark the deterministic exponent $\lambda_{\det}$ recovered as $\sigma= 0$.
    The horizontal line $\lambda=0$ is the invasion threshold. (b) The
    deterministic basic reproduction number $R_0$ (log scale) over the
    $(\alpha_1,\alpha_2)$ plane, with the deterministic threshold $R_0=1$ (solid
    white) and the stochastic threshold $\lambda=0$ at $\sigma=0.2$ (dashed
    black) overlaid. The hatched region marks host extinction, $\alpha_1\geq\nu-\mu\approx0.1931$. All computations use $K=1000$.}
    \label{fig:lambda}
\end{figure}

Figure \ref{fig:sim_I_N} compares the deterministic dynamics with stochastic realizations of the full reflected system at the fixed intervention efforts $(\alpha_1,\alpha_2)=(0.15,0.04)$, which lies between the two curves in Figure \ref{fig:lambda} (b). Figure \ref{fig:sim_I_N} (a) and (b) show a single stochastic path/realization at a moderate noise level of $\sigma=0.12$. The infected population $I$ exhibits a small early transient before decaying, whereas the deterministic $I$ rises to a substantial endemic level; see Figure \ref{fig:sim_I_N} (a). Meanwhile, the total population $N$ fluctuates around, and somewhat below, the deterministic curve; see Figure \ref{fig:sim_I_N} (b). A single realization, however, is only one draw from a broad distribution, so panels (c) and (d) summarize $1000$ independent realizations at the same noise level through their mean and a $95\%$ band.

At $\sigma=0.12$, two points are particularly notable. First, the stochastic average of $N$ stabilizes at a value below the deterministic equilibrium; see Figure \ref{fig:sim_I_N} (d). In a nonlinear model with multiplicative noise, the expected value of the stochastic trajectory generally differs from the deterministic solution. In particular, the disease-free stationary mean includes the It\^o correction $\overline N=K\left(1-\frac{\alpha_1+\frac12\sigma_N^2}{\nu-\mu}\right)$, which is smaller than $S_2^{*}$. Second, the mean of $I$ remains well under the deterministic endemic level, and the $95\%$ interval is markedly right-skewed; see Figure \ref{fig:sim_I_N} (c). Although the total population never becomes negative, it can attain large values, so the lower tail pulls the mean even as most sample paths stay small. At a higher noise intensity $\sigma=0.25$, the mean total population size is driven toward extinction over the simulation horizon, while the $95\%$ band shrinks; see \ref{fig:sim_I_N} (e) and (f). According to the model, the stochastic fadeout indicates that strong environmental variability can drive the host population to very low levels at some fixed intervention levels, even when the deterministic model predicts a persistent endemic state and population abundance.

\begin{figure}[H]
   \centering
    \subfigure[]{\includegraphics[width=7cm]{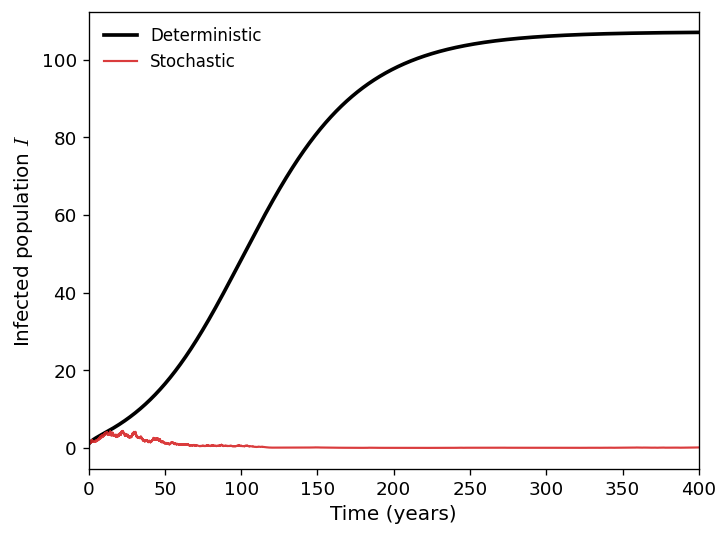}}
    \subfigure[]{\includegraphics[width=7cm]{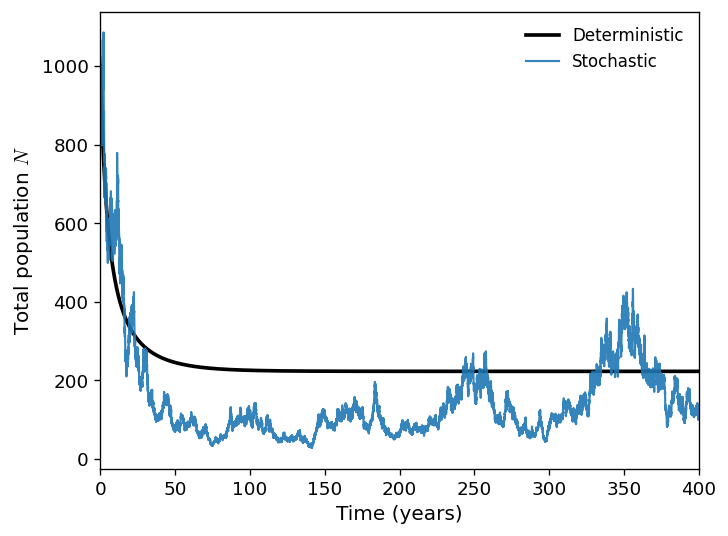}}
    \subfigure[]{\includegraphics[width=7cm]{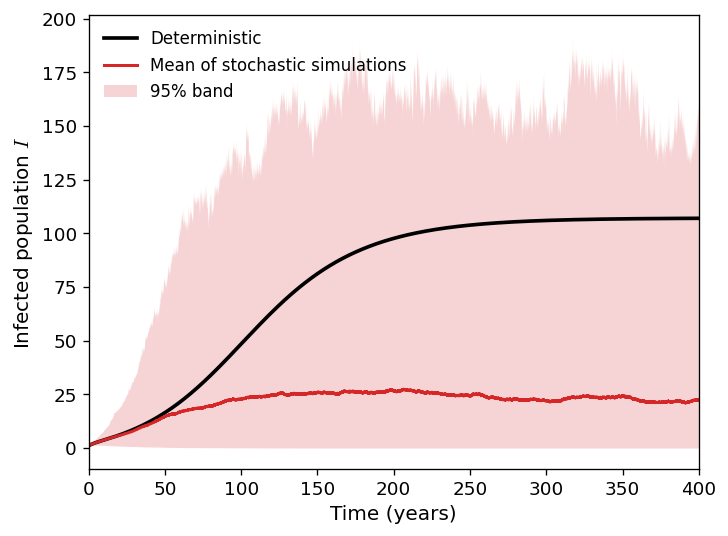}}
    \subfigure[]{\includegraphics[width=7cm]{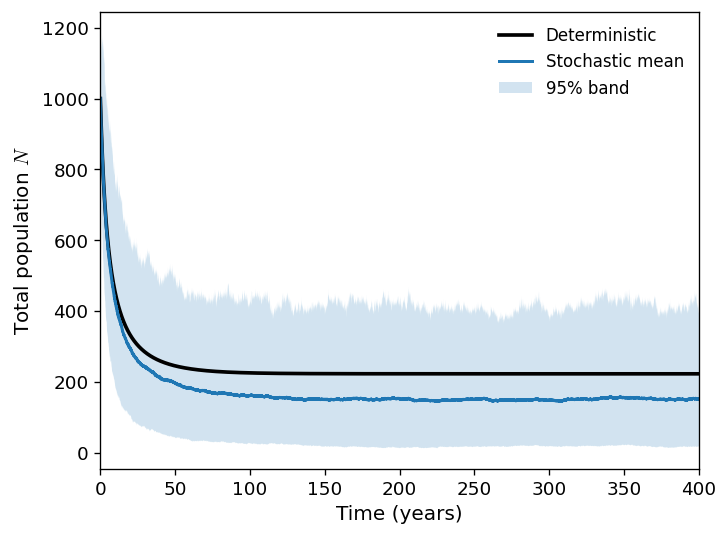}}
    \subfigure[]{\includegraphics[width=7cm]{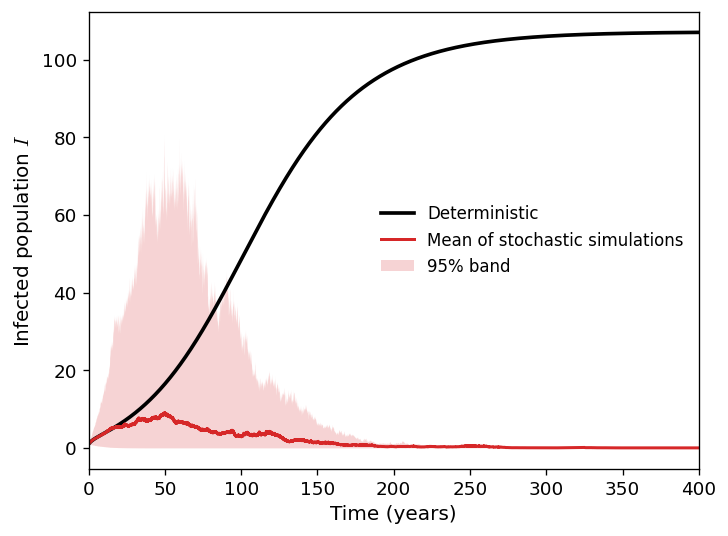}}
    \subfigure[]{\includegraphics[width=7cm]{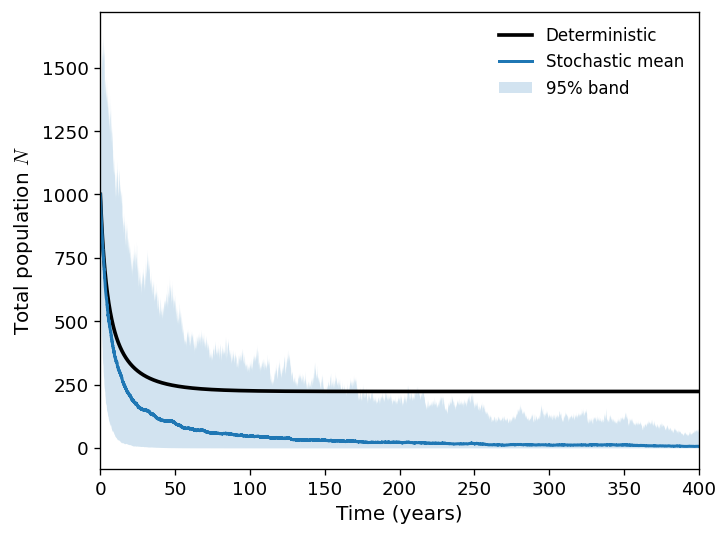}}
    \caption{Deterministic versus stochastic trajectories of the infected
    population $I$ (left column) and total population $N$ (right column) under
    the combined intervention $(\alpha_1,\alpha_2)=(0.15,0.04)$, integrated with
    the reflected Euler-Maruyama scheme (Algorithm \ref{alg:reflected-euler}).
    In every panel, the solid black curve is the deterministic solution.
    (a),(b) A single stochastic realization at noise intensity
    $\sigma_1=\sigma_2=\sigma_3=0.12$. (c),(d) The mean over $1000$
    independent realizations at $\sigma=0.12$, with the shaded $95\%$ band.
    (e),(f) The same summary at the higher noise intensity $\sigma=0.25$. Under moderate noise, the stochastic mean tracks the deterministic curve but settles below it due to the It\^o correction to the population mean. Under severe noise, the mean population is driven toward extinction while the $95\%$ band shrinks. All computations use $T=400$ years and $K=1000$.}
    \label{fig:sim_I_N}
\end{figure}

\subsection{Finite-Horizon Reflection Diagnostic}
The regulator $L$ tracks the total adjustment required to keep $ S_t \geq 0$. It thus serves as an indicator of the strength of the reflecting boundary's influence on the system’s dynamics. Since global nonexplosion has not been proven, we define this indicator in terms of the localized solution.

With $\tau_m$ as defined above, and given a finite time horizon $T>0$, set
\begin{equation}\label{eq:ell-localized}
\ell_{T,m}=\frac1T\mathbb E\left[L_{T\wedge\tau_m}\right].
\end{equation}
We refer to $\ell_{T,m}$ as the localized finite-horizon boundary-correction rate. When $\ell_{T,m}$ is small, reflection plays only a minor role, while large values signal that the system’s behavior is substantially influenced by the boundary correction.

\begin{proposition}
\label{prop:reflection-rate}
For every $T<\infty$ and $m\ge 1$, we have $0\le \ell_{T,m}<\infty$. In addition,
\begin{equation}\label{eq:reflection-balance}
\mathbb E[L_{T\wedge\tau_m}]=\mathbb E[S_{T\wedge\tau_m}]-S_0-\mathbb E\int_0^{T\wedge\tau_m} f_S(S_t,I_t,U_t)\,dt.
\end{equation}
If the initial condition belongs to the disease-free set $\mathcal F_+=\{S>0,\ I=U=0\}$, then $L_t\equiv 0$, and therefore $\ell_{T,m}=0$ for all $T$ and $m$.
\end{proposition}

\begin{proof}
Integrating the reflected $S$-equation up to time $T\wedge\tau_m$ gives
$$
S_{T\wedge\tau_m}=S_0
+\int_0^{T\wedge\tau_m} f_S(S_t,I_t,U_t)\,dt
+\int_0^{T\wedge\tau_m}\sigma_1N_t\,dW_{1,t}
+\int_0^{T\wedge\tau_m}\sigma_2S_t\,dW_{2,t}
+L_{T\wedge\tau_m}.
$$
On the time interval prior to $\tau_m$, the process remains in a compact subset of $\{N>0\}$. Hence, $f_S$, $N$, and $S$ are bounded on the stopped interval. As a result, the corresponding stopped stochastic integrals are square-integrable martingales with zero mean. Taking expectations then leads to \eqref{eq:reflection-balance}. The terms on the right-hand side of \eqref{eq:reflection-balance} are finite because the coefficients and state variables are bounded before $\tau_m$. Hence
$$
\mathbb E[L_{T\wedge\tau_m}]<\infty.
$$
Nonnegativity follows because $L$ is nondecreasing and $L_0=0$.

If the initial condition lies in $\mathcal F_+$, then $I_t = U_t = 0$ and $N_t = S_t > 0$ for all finite times. Hence, the disease-free stochastic logistic trajectory never hits the reflecting boundary. Using the support property
$$
\int_0^t \mathbf 1_{\{S_s>0\}}\,dL_s = 0,
$$
we conclude that $L_t \equiv 0$.
\end{proof}

\section{Optimal Control Using Reinforcement Learning}\label{Sec:Env}
The compartmental dynamics of the cervid population are governed by the ODE model in Equation \ref{M1}. A simulation state at each time step is defined by the numbers of susceptible and infected individuals, together with the environmental pathogen load. In this simulation environment, we identify optimal intervention strategies using a memory-based deep reinforcement learning (DRL) agent.

\subsection{Episodes}

We define an episode as a 200-year simulation. Interventions are applied at the beginning of each year. Each simulation year $T$ consists of 365 days and is divided into 128 time steps $t$, with the disease compartments updated approximately every 2-3 days. The resulting sequence of 128 observations per year is used as the input history for the DRL agent. This planning horizon was chosen to ensure convergence of optimal CWD control strategies.

\subsection{State Space}

The state is represented by a one-dimensional array describing the current compartmental status of the epizootic. Specifically, a state $\mathbf{s}_t = (S, I, U, N) \in \mathcal{S}$ consists of the numbers of susceptible deer ($S$), infected deer ($I$), environmental pathogen load ($U$), and total deer population ($N$). These components collectively define the state space $\mathcal{S} \subseteq \mathbb{R}_{\geq 0}^{4}$.

\subsection{Action Space}

At the beginning of each year $T$, the agent selects an intervention action based on the current state. In practice, disease dynamics can be influenced by the annual hunting rate and the environmental pathogen load. Such strategies are commonly employed in CWD management programs \cite{AGFC_CWD}.

To avoid large fluctuations in the deer population, the agent operates in a continuous action space, constrained to adjust the hunting rate by at most $10\%$ each year $T$. The Texas Parks and Wildlife Department estimates that approximately $430,000$ to $500,000$ white-tailed deer are harvested annually in Texas \cite{TexasHarvestData}. Given an estimated statewide population of approximately $5.3$ million \cite{TPWDDeerPopulation}, the corresponding annual harvesting rate ranges from $8.1\%$ to $9.4\%$. Based on these estimates, we set the baseline and minimum harvesting rates to $10\%$ and $4\%$, respectively.

In addition, the environmental accumulation of prions contributes significantly to CWD transmission \cite{almberg2011modeling}. Management strategies that reduce prion loads in cervid hotspots, such as bait piles, backyard feeders, stored forage, and grain bins, can help reduce localized contamination and disease transmission \cite{wvjen}. However, a standardized method for quantifying CWD prions in environmental materials is currently unavailable \cite{yuan2022sensitive}. In this study, environmental decontamination is represented by a control level ranging from $0$ to $10$. At the beginning of each year $T$, the agent increases or decreases the decontamination level by $0.5$, thereby dynamically modifying the environmental pathogen load over time.

At each year $T$, the agent selects an action from the action space
\begin{equation}
    a_t = (a_1, a_2),\quad a_1 \in (-1,1),\ a_2 \in \{0,1\}.
\end{equation}
The continuous component $a_1$ does not directly set the hunting rate. Rather, it determines the year-over-year \emph{adjustment} applied to the existing rate
\begin{equation}\label{HuntingUpdate}
    \alpha_1(T) = \mathrm{clip}\Big(\alpha_1(T-1) + 0.10\times a_1,\ \alpha_1^{\min},\ \alpha_1^{\max}\Big),
\qquad \alpha_1^{\min}=4\%,\ \alpha_1^{\max}=19.31\%.
\end{equation}
Thus $a_1=1$ corresponds to the largest permitted annual increase of $10$ percentage points, and $a_1=-1$ the largest permitted decrease. The realized rate is never allowed to fall below the floor $\alpha_1^{\min}$. The ceiling $\alpha_1^{\max}$ is set equal to the extinction threshold identified in Section \ref{Sec:Model}. The agent is therefore never permitted to select a hunting rate that the deterministic stability analysis predicts would drive the population extinct.

\subsection{Reward Structure}

Our reward function incorporates both the hunting rate and the environmental decontamination level. During each year $T$, we track the variables $S$, $I$, $N$, and $U$ and compute their annual averages. The agent maximizes the following reward

\begin{equation}\label{Rfn}
\begin{split}
r(T)=&
\;\omega_1\left(
\frac{S_{(T,\mathrm{mean})}}{N_{(T,\mathrm{mean})}}
-
\frac{S_{(T-1,\mathrm{mean})}}{N_{(T-1,\mathrm{mean})}}
\right)
\\
&+\omega_2\left(
\frac{I_{(T-1,\mathrm{mean})}}{N_{(T-1,\mathrm{mean})}}
-
\frac{I_{(T,\mathrm{mean})}}{N_{(T,\mathrm{mean})}}
\right)
\\
&+\omega_3\left(
\frac{U_{(T-1,\mathrm{mean})}}{N_{(T-1,\mathrm{mean})}}
-
\frac{U_{(T,\mathrm{mean})}}{N_{(T,\mathrm{mean})}}
\right).
\end{split}
\end{equation}

The reward accounts for changes in the proportions of each compartment $(S, I, U)$. It assigns positive rewards to increases in the proportion of susceptible individuals, while penalizing higher levels of infection and environmental contamination. The weights $\omega_1$, $\omega_2$, and $\omega_3$ determine the relative importance of the corresponding components. This structure encourages policies that promote population recovery while reducing both disease prevalence and environmental pathogen load.

\subsection{Training Algorithm}

We use a Proximal Policy Optimization (PPO) algorithm to learn optimal intervention strategies \cite{schulman2017proximal}. It employs an actor network that learns a policy for selecting actions and a critic network that evaluates the quality of states. PPO improves training stability and sample efficiency. It optimizes a surrogate objective while limiting large policy updates.

In general, policy-gradient methods optimize a parametrized policy by estimating the policy gradient and applying stochastic gradient ascent. The most widely used gradient estimator is
\begin{equation}
\hat{g} = \hat{\mathbb{E}}_t \left[ \nabla_{\theta} \log \pi_{\theta}(a_t \mid s_t)\hat{A}_t \right],
\end{equation}
where $\pi_{\theta}$ denotes a stochastic policy and $\theta$ is the vector of weights and biases of the policy network. Here $\hat{A}_t$ is an estimator of the advantage function at timestep $t$. The expectation $\hat{\mathbb{E}}_t[\cdot]$ is the empirical average over a finite batch of samples, in an algorithm that alternates between sampling and optimization. The gradient estimator $\hat{g}$ is obtained by differentiating the objective
\begin{equation}
L^{PG}(\theta) = \hat{\mathbb{E}}_t \left[ \log \pi_{\theta}(a_t \mid s_t)\hat{A}_t \right].
\end{equation}

Updating $\theta$ with policy-gradient objectives can lead to excessively large policy updates. The Trust Region Policy Optimization (TRPO) approach addresses this by maximizing a surrogate objective subject to a constraint on the update size \cite{schulman2015trust}
\begin{equation}
\max_{\theta} \ \hat{\mathbb{E}}_t \left[
\frac{\pi_{\theta}(a_t \mid s_t)}
{\pi_{\theta_{\text{old}}}(a_t \mid s_t)}
\hat{A}_t
\right]
\end{equation}
subject to
\begin{equation}
\hat{\mathbb{E}}_t \left[
KL\left(
\pi_{\theta_{\text{old}}}(\cdot \mid s_t)\|
\pi_{\theta}(\cdot \mid s_t)
\right)
\right]
\leq \delta.
\end{equation}
Here $\theta_{\text{old}}$ is the vector of policy parameters before the most recent update. The Kullback-Leibler (KL) divergence measures the difference between two probability distributions $p(x)$ and $q(x)$. It is defined as
\begin{equation}
KL(p \,\|\, q)
=
\sum_x p(x)
\log\left(
\frac{p(x)}{q(x)}
\right)
\end{equation}
for discrete distributions, and
\begin{equation}
KL(p \,\|\, q)
=
\int p(x)
\log\left(
\frac{p(x)}{q(x)}
\right)\,dx
\end{equation}
for continuous distributions. PPO replaces the KL-divergence constraint in TRPO with a KL penalty. The resulting KL-penalized surrogate objective is
\begin{equation}
\max_{\theta} \ \hat{\mathbb{E}}_t \left[
\frac{\pi_{\theta}(a_t \mid s_t)}
{\pi_{\theta_{\text{old}}}(a_t \mid s_t)}
\hat{A}_t
- \beta \,
KL\left[
\pi_{\theta_{\text{old}}}(\cdot \mid s_t),
\pi_{\theta}(\cdot \mid s_t)
\right]
\right],
\end{equation}
where $\beta \in [0, \infty)$ controls the strength of the penalty. Larger values of $\beta$ impose a greater penalty on the divergence between $\pi_{\theta}$ and $\pi_{\theta_{\text{old}}}$.

Evaluating and enforcing KL-based constraints or penalties can increase computational complexity. PPO further simplifies the optimization by introducing a clipped surrogate objective that maintains training stability
\begin{equation}
L^{CLIP}(\theta)
=
\hat{\mathbb{E}}_t
\left[
\min
\left(
\rho_t(\theta)\hat{A}_t,
\;
\text{clip}\bigl(\rho_t(\theta), 1-\epsilon, 1+\epsilon\bigr)\hat{A}_t
\right)
\right],
\end{equation}
where $\epsilon \in (0, 1)$ is a hyperparameter and $\rho_t(\theta)=
\frac{
\pi_{\theta}(a_t \mid s_t)
}{
\pi_{\theta_{\text{old}}}(a_t \mid s_t)
}$ is the probability ratio between the current and old policies. The first term inside the minimum is the unclipped surrogate objective. The clip function restricts $\rho_t(\theta)$ to $[1-\epsilon,1+\epsilon]$, bounding the objective between $(1-\epsilon)\hat{A}_t$ and $(1+\epsilon)\hat{A}_t$. Taking the minimum yields a lower bound on the unclipped objective. This prevents large policy updates without the overhead of explicit KL evaluation.

Figure \ref{fig:PPO} illustrates the DRL training workflow. The simulation and RL environment are initialized by specifying the total and initial infected populations, the initial environmental load, the simulation duration, and the Long Short-Term Memory (LSTM) sequence length. For all training runs, we assume a carrying capacity of $K=10,000$ individuals, an initial deer population of $1,000$, one initially infected individual, and an initial environmental load of $1$. The agent's memory sequence length is set to $128$ to capture sufficient historical context.

Once initialized, the simulation advances one year. At the end of a step, the compartmental states are passed to the DRL agent, which selects an intervention for the following year. The simulation then advances another year and returns updated statistics. The agent's performance is evaluated using the reward function in Equation \ref{Rfn}. This loop repeats until the final timestep.

\begin{figure}[H]
    \centering
    \subfigure{\includegraphics[width=15cm]{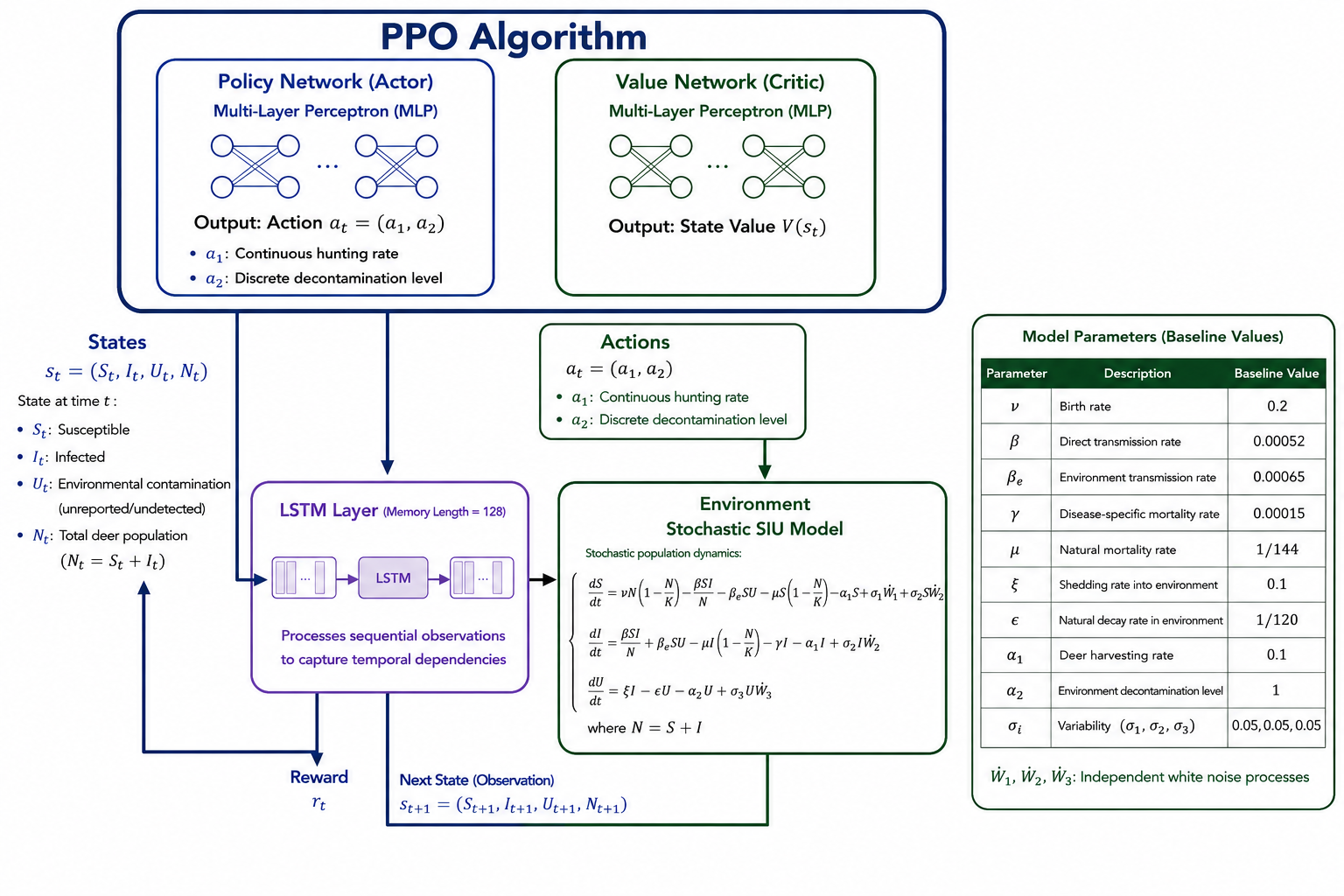}}
    \caption{DRL framework for CWD control}
    \label{fig:PPO}
\end{figure}

During training, we use the \texttt{EvalCallback} function in Stable-Baselines3 to select the best model by highest mean evaluation reward. At a predefined frequency, \texttt{EvalCallback} runs a fixed number of episodes, computes the mean reward, and saves the model whenever this mean exceeds the previous best. The saved model therefore reflects the highest reward achieved during training. This model is then used to predict $S$, $I$, and $U$ over a fixed time period. The agent uses a learning rate of $3 \times 10^{-4}$, discount factor $\gamma=0.99$, Generalized Advantage Estimation (GAE) parameter $\lambda=0.95$, PPO clipping parameter $\epsilon =0.2$, batch size of 128, rollout length of 128 steps, and 10 optimization epochs per update.

Training consists of approximately $10,000$ episodes, each a $200$-year simulation. The policy was evaluated every 200 time steps. The learned policy is then evaluated in a single environment for the deterministic scenario and in 30 independent environments for the stochastic scenario. From these results, we determine whether the trained agent can control disease transmission, maintain population persistence, and achieve high cumulative rewards.

\section{Results of RL Control}\label{Sec:Res}
We evaluate two modeling approaches, a deterministic model and a stochastic model. For each, we consider three intervention strategies: adjusting the hunting rate, implementing environmental decontamination, or applying both. We also investigate the impact of different initial conditions, and how varying the reward weights, which reflect different decision-maker priorities, shapes the resulting strategies. By comparing results, we show how changes in the reward function yield qualitatively distinct strategies tailored to specific management objectives. All results in this section are obtained from a single independent testing phase.

\subsection{Deterministic Optimal Control}
In this section, we compare the dynamics of $S$, $I$, and $U$ under different strategies before and after intervention.

\textbf{Scenario 1: Optimal Hunting Rates}

Here the RL agent adjusts the hunting rate while the decontamination rate $\alpha_2$ is fixed at $1$. Figure \ref{fig:DE1_SIUcompare} shows the effect on disease dynamics. Panels (a), (b), and (c) compare the uncontrolled epizootic (gray curves) to the hunting-rate policy (colored curves) under reward weights $(\omega_1, \omega_2, \omega_3) = (5, 5, 2)$. Under intervention, the infected population $I$ remains very low. The uncontrolled epizootic, by contrast, peaks above $3,000$ individuals (Figure \ref{fig:DE1_SIUcompare} (b)). The environmental pathogen load $U$ stays below $10$ units under intervention. In the uncontrolled scenario, $U$ rises above $300$ units and persists at an elevated level (Figure \ref{fig:DE1_SIUcompare} (c)). The hunting rates applied during testing are shown in Figure \ref{fig:DE1_SIUcompare} (d). The optimal policy increases the hunting rate to a maximum of $0.193$ over the following century. Since hunting is the only available control here, the agent learns that increasing harvest intensity is the most effective way to reduce disease prevalence.

\begin{figure}[H]
   \centering
    \subfigure[]{\includegraphics[width=7cm]{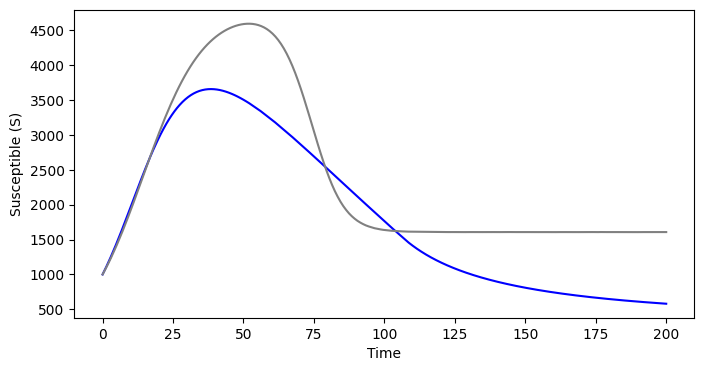}}
    \subfigure[]{\includegraphics[width=7cm]{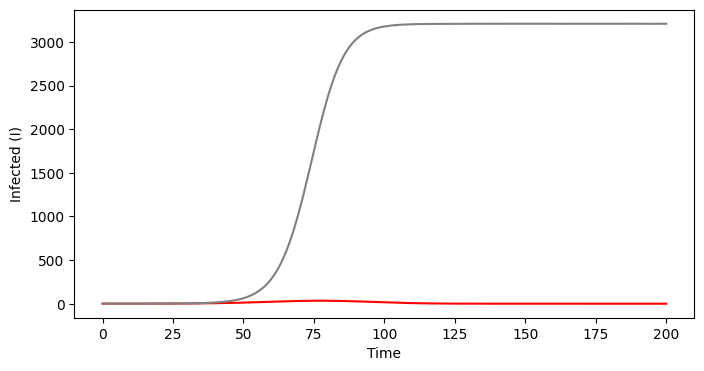}}
    \subfigure[]{\includegraphics[width=7cm]{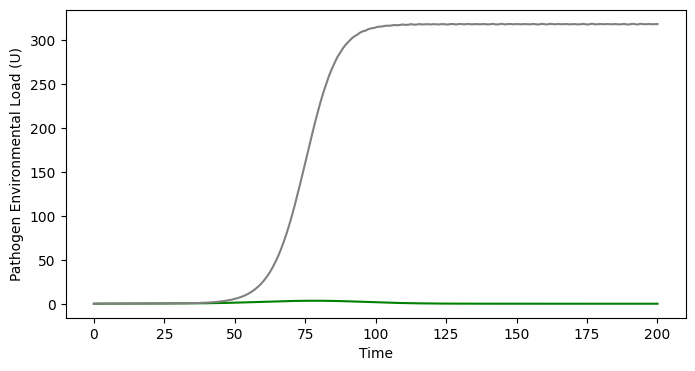}}
    \subfigure[]{\includegraphics[width=7cm]{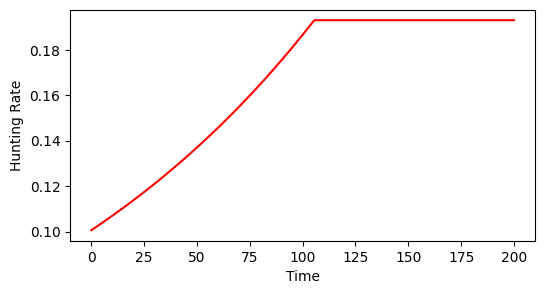}} \caption{(a) Comparison of uncontrolled (gray) and controlled (colored) dynamics under the optimal hunting rate policy with $(\omega_1, \omega_2, \omega_3) = (5, 5, 2)$. (a) Susceptible population $S$ with uncontrolled (grey) vs.\ intervention (blue). (b) Infected population $I$ with uncontrolled (grey) vs.\ intervention (red). (c) Environmental pathogen load $U$ with uncontrolled (grey) vs.\ intervention (green). (d) Hunting rates applied during the testing phase.}
    \label{fig:DE1_SIUcompare}
\end{figure}

Table \ref{tab:hunting_results} compares the mean predicted population before and after intervention. The policy reduces the infected population and environmental pathogen load by more than $99\%$, so targeted culling suppresses both individual-level infection and environmental contamination. However, it also reduces the total and susceptible populations, by $57.83\%$ and $22.31\%$ respectively.

\begin{table}[htbp]
\centering
\small
\caption{Mean population characteristics before and after intervention under the hunting-rate adjustment policy with reward weights $(5,5,2)$.}
\label{tab:hunting_results}

\begin{tabular}{p{0.32\linewidth}ccc}
\hline
\textbf{Population} &
\textbf{Pre-} &
\textbf{Post-} &
\textbf{Change} \\
\textbf{Characteristic} &
\textbf{intervention} &
\textbf{intervention} &
\textbf{(\%)} \\
\hline
Total population ($N$)
& 4402.5235 & 1856.2682 & $-57.83$ \\

Susceptible population ($S$)
& 2378.5226 & 1848.0202 & $-22.31$ \\

Infected population ($I$)
& 2024.0009 & 8.2479 & $-99.59$ \\

Environmental pathogen load ($U$)
& 199.1911 & 0.8180 & $-99.59$ \\
\hline
\end{tabular}
\end{table}

\textbf{Scenario 2: Optimal Environmental Decontamination}

Here the hunting rate $\alpha_1$ is fixed at $10\%$, and the RL agent adjusts the decontamination rate. It starts from an initial value of $1$ and increments or decrements by $0.5$ each year $T$. Reward weights are $(\omega_1, \omega_2, \omega_3) = (5, 5, 2)$, as in Scenario 1. Figure \ref{fig:DE2_SIUcompare} shows the evolution of $S$, $I$, and $U$ before and after intervention. Panels (a), (b), and (c) display each compartment, and panel (d) the decontamination policy. The agent consistently selects the maximum decontamination rate of $10$ and tends to increase effort in the next term. This suggests that reducing the environmental pathogen load is one of the most effective ways to prevent transmission. Since CWD prions can survive in the environment for extended periods, contaminated environments act as long-term reservoirs of infection.

\begin{figure}[H]
    \centering
    \subfigure[]{\includegraphics[width=7cm]{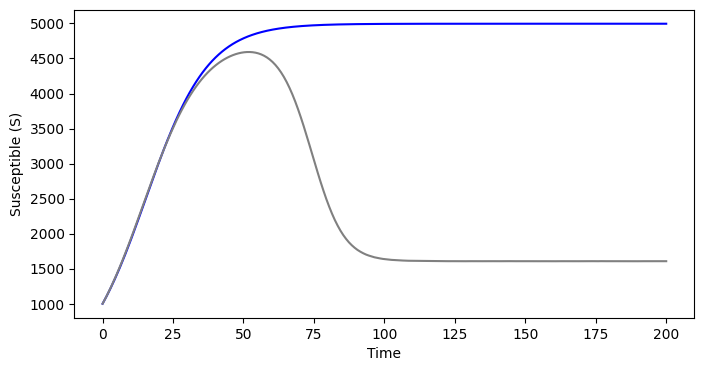}}
    \subfigure[]{\includegraphics[width=7cm]{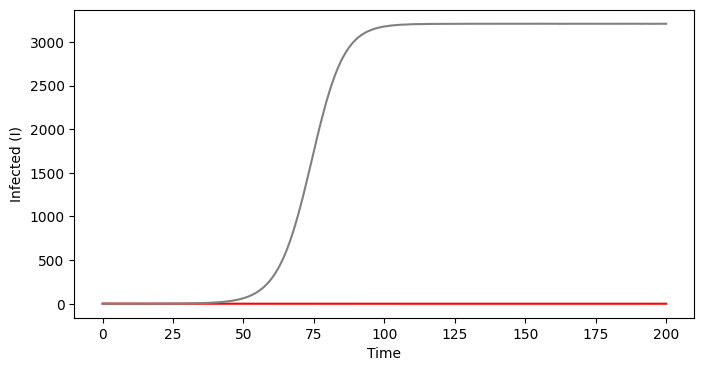}}
    \subfigure[]{\includegraphics[width=7cm]{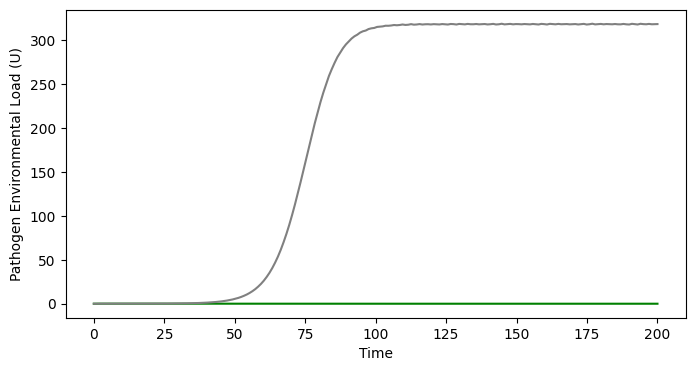}}
    \subfigure[]{\includegraphics[width=7cm]{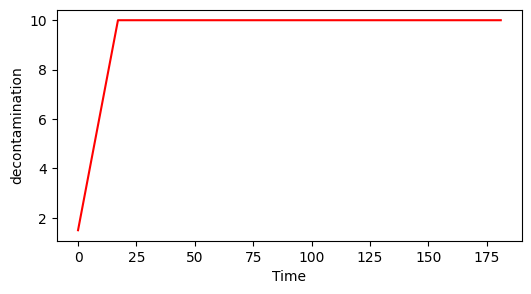}}
    \caption{(a) Comparison of uncontrolled (gray) and controlled (colored) dynamics under the optimal environmental decontamination policy with $(\omega_1, \omega_2, \omega_3) = (5, 5, 2)$. (a) Susceptible population $S$ with uncontrolled (grey) vs.\ intervention (blue). (b) Infected population $I$ with uncontrolled (grey) vs.\ intervention (red). (c) Environmental pathogen load $U$ with uncontrolled (grey) vs.\ intervention (green). (d) Decontamination rates were applied during the testing phase.}
    \label{fig:DE2_SIUcompare}
\end{figure}

Table \ref{tab:env_intervention} compares the mean population values before and after intervention. The susceptible population nearly doubles, the infected population is almost entirely eliminated, and the environmental pathogen load is reduced to near zero.

\begin{table}[htbp]
\centering
\small
\caption{Mean predicted values of population variables before and after intervention under the environmental decontamination policy with reward weights $(\omega_1,\omega_2,\omega_3)=(5,5,2)$.}
\label{tab:env_intervention}

\begin{tabular}{p{0.32\linewidth}ccc}
\hline
\textbf{Population} &
\textbf{Pre-} &
\textbf{Post-} &
\textbf{Change} \\
\textbf{Characteristic} &
\textbf{intervention} &
\textbf{intervention} &
\textbf{(\%)} \\
\hline

Total population ($N$)
& 4402.52 & 4551.79 & $+3$ \\

Susceptible population ($S$)
& 2378.52 & 4551.72 & $+91$ \\

Infected population ($I$)
& 2024.00 & 0.07 & $-100$ \\

Environmental pathogen load ($U$)
& 199.19 & 0.002 & $-100$ \\
\hline
\end{tabular}
\end{table}

\textbf{Scenario 3: Optimal Hunting with Environmental Decontamination}

When both $\alpha_1$ and $\alpha_2$ are adjustable, the RL agent updates both each year $T$. As shown in Figure \ref{fig:DE3_SIUcompare} (a) and (b), the susceptible population stays relatively high while the infected population is kept low.

The corresponding control actions are shown in Figure \ref{fig:DE3_SIUcompare} (c) and (d). The agent gradually reduces the hunting rate, reaching a minimum after about $30$ years. Meanwhile it rapidly increases the decontamination rate to its maximum of $10$ and sustains it thereafter. This suggests that early, aggressive environmental intervention can quickly suppress the prion load, and that sustained high-level decontamination maintains control over the long term.

\begin{figure}[H]
    \centering
    \subfigure[]{\includegraphics[width=6.6cm]{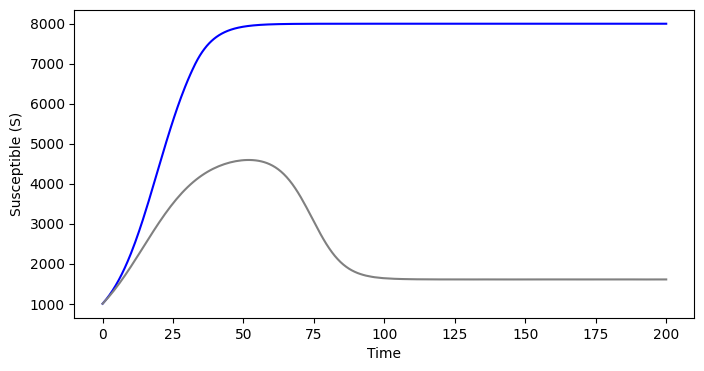}}
    \subfigure[]{\includegraphics[width=6.6cm]{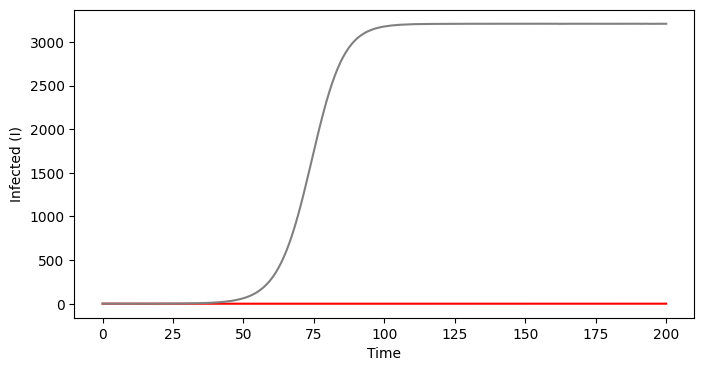}}
    \subfigure[]
    {\includegraphics[width=6.8cm]{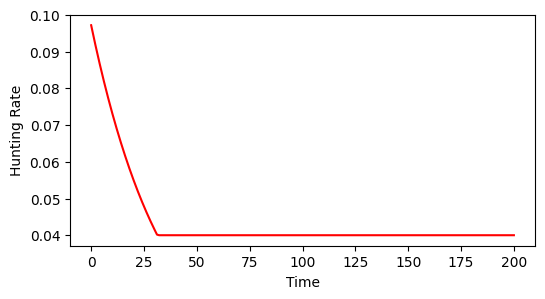}}
    \subfigure[]{\includegraphics[width=6.6cm]{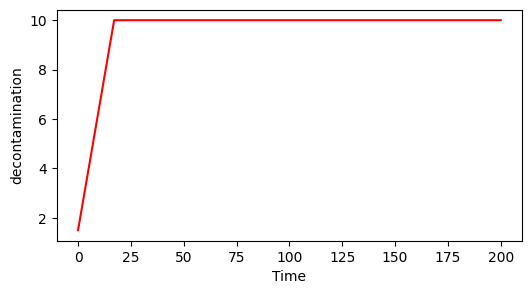}}
    \caption{(a) Comparison of uncontrolled (gray) and controlled (colored) dynamics under the combined hunting and environmental decontamination policy with $(\omega_1, \omega_2, \omega_3) = (5, 5, 2)$. (a) Susceptible population $S$ with uncontrolled (grey) vs.\ intervention (blue). (b) Infected population $I$ with uncontrolled (grey) vs.\ intervention (red). (c) Hunting rates applied during the testing phase. (d) Decontamination rates were applied during the testing phase.}
    \label{fig:DE3_SIUcompare}
\end{figure}

Table \ref{tab:intervention_comparison} compares outcomes across the three strategies under reward weights $(\omega_1, \omega_2, \omega_3) = (5, 5, 2)$. The hunting-only strategy reduces the infected population mainly through aggressive removal, but also greatly reduces the total population. Its apparent success in reducing infection thus comes at the cost of population loss. The decontamination strategy is highly effective at reducing infection and pathogen load, but yields only limited recovery of the total population. The combined strategy gives the strongest outcomes. It increases the total population by more than $60\%$ and the susceptible population by $208\%$, while still achieving substantial reductions in infection and contamination. It provides the most favorable trade-off between disease suppression and population persistence.

\begin{table}[htbp]
\centering
\footnotesize
\caption{Mean predicted values of population variables before and after intervention under the combined hunting and environmental decontamination policy with reward weights $(\omega_1, \omega_2, \omega_3) = (5, 5, 2)$.}
\label{tab:intervention_comparison}

\begin{tabular}{p{0.22\linewidth}ccccccc}
\hline
\textbf{Population} &
\textbf{No} &
\textbf{Hunting} &
\textbf{Change} &
\textbf{Environmental} &
\textbf{Change} &
\textbf{Combined} &
\textbf{Change} \\
\textbf{Characteristic} &
\textbf{Intervention} &
&
\textbf{(\%)} &
\textbf{Decontamination} &
\textbf{(\%)} &
&
\textbf{(\%)} \\
\hline

Total population ($N$)
& 4402.52 & 1856.27 & $-58$
& 4451.79 & $+3$
& 7412.42 & $+68$ \\

Susceptible population ($S$)
& 2378.52 & 1848.02 & $-22$
& 4451.72 & $+91$
& 7329.82 & $+208$ \\

Infected population ($I$)
& 2024.00 & 8.25 & $-100$
& 0.07 & $-100$
& 82.60 & $-96$ \\

Environmental pathogen load ($U$)
& 199.19 & 0.82 & $-100$
& 0.002 & $-100$
& 1.37 & $-96$ \\
\hline
\end{tabular}
\end{table}

\subsubsection{Initial Condition Analysis}

To test sensitivity to initial conditions, we use three initial hunting rates ($0.01$, $0.05$, and $0.15$) and vary the initial environmental pathogen load. Table \ref{tab:hunting_effects} summarizes the effect of these initial hunting rates relative to the uncontrolled baseline. Following the combined intervention, total population size increases by $76$-$86\%$, driven mainly by a $225$-$230\%$ rise in the susceptible population. The infected population declines by $70$-$86\%$ and the environmental pathogen load by $74$-$99\%$. Disease control is therefore robust across all initial conditions tested.

\begin{table}[htbp]
\centering
\footnotesize
\caption{Effect of the initial hunting rate on population dynamics and disease outcomes under the combined intervention policy.}
\label{tab:hunting_effects}

\begin{tabular}{p{0.22\linewidth}ccccccc}
\hline
\textbf{Population} &
\textbf{No} &
\textbf{Hunting} &
\textbf{Change} &
\textbf{Hunting} &
\textbf{Change} &
\textbf{Hunting} &
\textbf{Change} \\
\textbf{Characteristic} &
\textbf{Intervention} &
\textbf{Rate = 0.01} &
\textbf{(\%)} &
\textbf{Rate = 0.05} &
\textbf{(\%)} &
\textbf{Rate = 0.15} &
\textbf{(\%)} \\
\hline

Total population ($N$)
& 4639.85 & 8605.60 & $+85$
& 8184.26 & $+76$
& 8637.22 & $+86$ \\

Susceptible population ($S$)
& 2418.92 & 7954.82 & $+229$
& 7870.52 & $+225$
& 7971.78 & $+230$ \\

Infected population ($I$)
& 2220.94 & 650.78 & $-71$
& 313.74 & $-86$
& 665.44 & $-70$ \\

Environmental pathogen load ($U$)
& 221.44 & 58.62 & $-74$
& 3.28 & $-99$
& 21.35 & $-90$ \\
\hline
\end{tabular}
\end{table}

Table \ref{tab:env_effects} summarizes the effect of varying the initial environmental decontamination rate relative to the uncontrolled baseline.

\begin{table}[htbp]
\centering
\footnotesize
\caption{Effect of the initial environmental decontamination rate on population dynamics and disease outcomes under the combined intervention policy.}
\label{tab:env_effects}

\begin{tabular}{p{0.22\linewidth}ccccccc}
\hline
\textbf{Population} &
\textbf{No} &
\textbf{Environmental} &
\textbf{Change} &
\textbf{Environmental} &
\textbf{Change} &
\textbf{Environmental} &
\textbf{Change} \\
\textbf{Characteristic} &
\textbf{Intervention} &
\textbf{Rate = 0.001} &
\textbf{(\%)} &
\textbf{Rate = 5} &
\textbf{(\%)} &
\textbf{Rate = 10} &
\textbf{(\%)} \\
\hline

Total population ($N$)
& 4639.85 & 8897.99 & $+92$
& 8503.57 & $+83$
& 7950.61 & $+71$ \\

Susceptible population ($S$)
& 2418.92 & 8018.64 & $+231$
& 8278.00 & $+242$
& 7812.14 & $+223$ \\

Infected population ($I$)
& 2220.94 & 879.36 & $-60$
& 225.57 & $-90$
& 138.46 & $-94$ \\

Environmental pathogen load ($U$)
& 221.44 & 42.71 & $-81$
& 2.24 & $-99$
& 1.37 & $-99$ \\
\hline
\end{tabular}
\end{table}

We also compare performance across weightings. Increases in total and susceptible populations are similar across configurations. The largest reduction in the infected population is achieved under weights $(\omega_1, \omega_2, \omega_3) = (5, 5, 2)$. Reward weighting therefore meaningfully influences infection control, and different settings yield distinct epidemiological outcomes and environmental risk profiles.

Table \ref{tab:population_weights} presents the predicted mean values of $N$, $S$, $I$, and $U$ under different reward weight configurations. Across all schemes, $N$ and $S$ show similar proportional increases. The infected population $I$ shows its largest reduction under weights $(5, 5, 2)$, so this configuration most effectively suppresses prevalence. Together, these results confirm that reward weighting drives qualitatively distinct outcomes, with each configuration producing characteristic patterns of disease dynamics and environmental risk.

\begin{table}[htbp]
\centering
\footnotesize
\caption{Predicted mean population variables and percentage changes under different reward weight configurations.}
\label{tab:population_weights}

\begin{tabular}{p{0.20\linewidth}ccccccccc}
\hline
\textbf{Population} &
\textbf{No} &
\textbf{Weight} &
\textbf{Change} &
\textbf{Weight} &
\textbf{Change} &
\textbf{Weight} &
\textbf{Change} &
\textbf{Weight} &
\textbf{Change} \\
\textbf{Characteristic} &
\textbf{Intervention} &
\textbf{$(5,5,2)$} &
\textbf{(\%)} &
\textbf{$(5,5,5)$} &
\textbf{(\%)} &
\textbf{$(10,5,2)$} &
\textbf{(\%)} &
\textbf{$(5,10,2)$} &
\textbf{(\%)} \\
\hline

Total population ($N$)
& 4402.52 & 8970.52 & $+104$
& 8991.74 & $+104$
& 8906.45 & $+102$
& 9039.84 & $+105$ \\

Susceptible population ($S$)
& 2378.52 & 8511.54 & $+258$
& 8426.92 & $+254$
& 7757.68 & $+226$
& 8421.43 & $+254$ \\

Infected population ($I$)
& 2024.00 & 458.98 & $-77$
& 564.82 & $-72$
& 1148.77 & $-43$
& 618.41 & $-69$ \\

Environmental pathogen load ($U$)
& 199.19 & 5.88 & $-97$
& 13.21 & $-93$
& 941.98 & $+373$
& 8.81 & $-96$ \\
\hline
\end{tabular}
\end{table}

\subsection{Stochastic Optimal Control}
When both $\alpha_1$ and $\alpha_2$ are adjustable under stochastic dynamics, the control actions applied during testing are shown in Figure \ref{fig:SDE_SIandaction}. The agent gradually reduces the hunting rate, reaching minimal levels after about $25$ years and sustaining them thereafter. Meanwhile, it rapidly increases the decontamination rate at the outset and maintains it at its maximum for the remainder of the period.

This strategy is broadly consistent with the deterministic policy in Figure \ref{fig:DE3_SIUcompare} (c) and (d), but differs in one respect. Under stochasticity, the agent holds a higher hunting rate for a shorter initial period before switching to minimal culling, while sustaining maximum decontamination throughout. In the presence of environmental noise, the agent thus compensates for reduced predictability by front-loading hunting effort.

\begin{figure}[H]
    \centering
    \subfigure[]{\includegraphics[width=7cm]{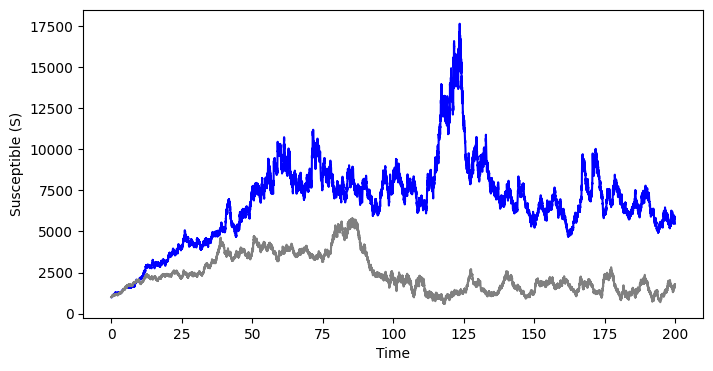}}
    \subfigure[]{\includegraphics[width=7cm]{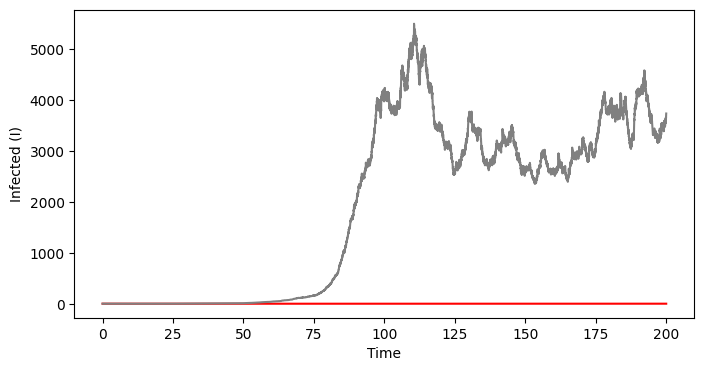}}
    \subfigure[]
    {\includegraphics[width=7cm]{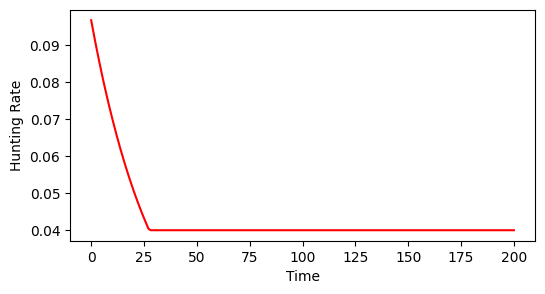}}
    \subfigure[]
    {\includegraphics[width=7cm]{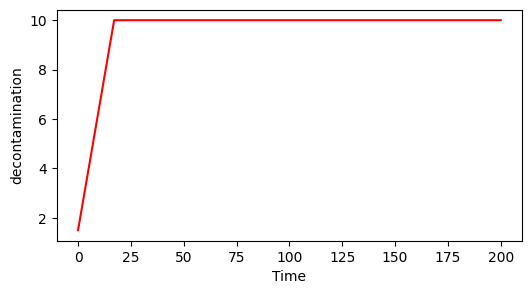}}    
    \caption{Comparison of uncontrolled (gray) and controlled (colored) dynamics under the combined hunting and environmental decontamination policy in the stochastic setting. (a) Susceptible population $S$ with uncontrolled (gray) vs.\ intervention (blue). (b) Infected population $I$ with uncontrolled (gray) vs.\ intervention (red). (c) Hunting rates applied during the testing phase. (d) Decontamination rates were applied during the testing phase.}
    \label{fig:SDE_SIandaction}
\end{figure}

\subsubsection{Robustness Analysis}
We now examine the effects of environmental fluctuations by varying the noise intensities and comparing outcomes in low- and high-contagion settings.

\textbf{Impact of Uncertainty.} A key concern for control policies in stochastic environments is robustness to uncertainty. To assess reliability, we apply the optimal policy across $30$ independent realizations with $\sigma_1 = \sigma_2 = \sigma_3 = 0.05$. As shown in Figure \ref{fig:boxplot} (b), the infected population stays very low in most runs, with the median near the lower quartile. Several outliers extend upward. Two cases show hundreds of infections, and three show outbreaks with mean infection counts above $1,000$. Stochastic variability can therefore substantially disrupt disease dynamics, leading to occasional failures of the control strategy.

To assess robustness across noise levels, we vary $\sigma_1, \sigma_2$, and $\sigma_3$ simultaneously. At $\sigma_1 = \sigma_2 = \sigma_3 = 0.03$, outbreaks are less frequent (Figure \ref{fig:boxplot}(a)). At $\sigma_1 = \sigma_2 = \sigma_3 = 0.07$, the infected population distribution becomes considerably more dispersed (Figure \ref{fig:boxplot}(c)), confirming that higher noise undermines control. At $\sigma_1 = \sigma_2 = \sigma_3 = 0.1$, large fluctuations and systematic instability emerge, highlighting the limits of the framework under high uncertainty. This raises the question of which reinforcement learning algorithms or robust control frameworks are best suited to environments with substantial stochastic perturbations.

\begin{figure}[H]
    \centering
    \subfigure[]{\includegraphics[width=5 cm]
    {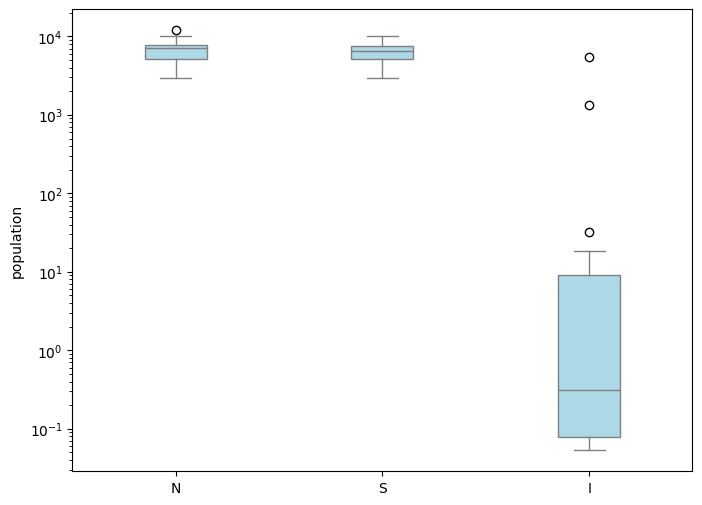}}
    \subfigure[]{\includegraphics[width=5 cm]
    {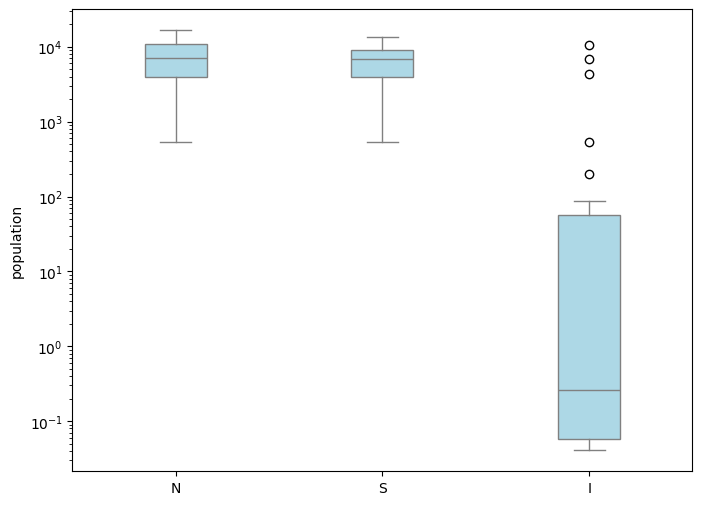}}
    \subfigure[]{\includegraphics[width=5 cm]{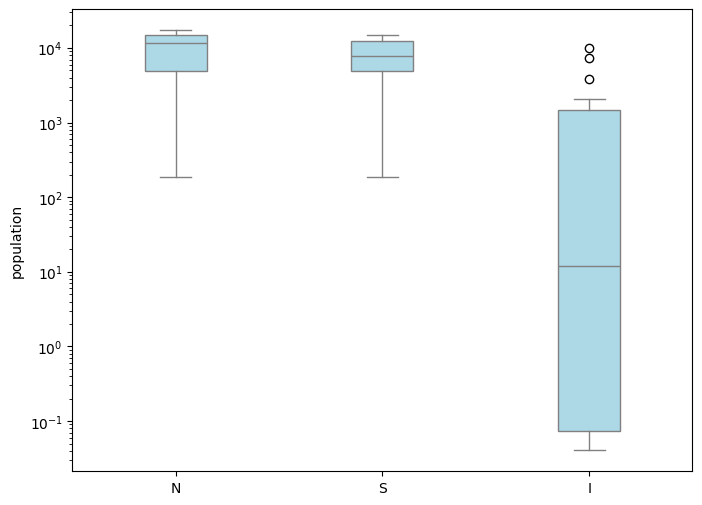}}      
    \caption{Box plots of total population $N$, susceptible population $S$, and infected population $I$ across $30$ independent stochastic realizations under the optimal combined intervention policy. (a) $\sigma_1 = \sigma_2 = \sigma_3 = 0.03$. (b) $\sigma_1 = \sigma_2 = \sigma_3 = 0.05$. (c) $\sigma_1 = \sigma_2 = \sigma_3 = 0.07$.}
    \label{fig:boxplot}
\end{figure}

\textbf{Impact of the transmission rates on policies.}
The basic reproduction number $R_0$ for CWD varies considerably, as it depends on model structure, deer species, and local environmental conditions \cite{delamater2019complexity}. Previous analyses of CWD in Alberta estimated $R_0$ for mule deer at between $2.2$ and $4.5$ during the initial phase of the epizootic \cite{Potapov2015}. Other work suggests $R_0$ may be much higher where environmental transmission occurs, since pathogens persisting in reservoirs can cause new infections even after the original host has died \cite{wallinga2007generation}. Reported values as high as $10$-$11$ underscore the role of indirect transmission and long-term prion contamination in the establishment and persistence of disease \cite{almberg2011modeling}.

Against this range, we study the policies learned under two epidemiological scenarios, defined by contrasting direct and environmental transmission rates. The first is a low-contagion environment with $\beta = 0.000052$ and $\beta_e = 0.0001$, yielding $R_{0} < 1$. Here, the disease poses a limited threat and is expected to fade out over time. The second is a highly contagious environment with $\beta = 0.001$ and $\beta_{e} = 0.003$, yielding $R_{0} \approx 22$, where intense transmission threatens population persistence. Comparing the two cases demonstrates the framework's robustness and adaptability across varying levels of disease severity and stochasticity.

\textbf{Case 1: Low-Contagion Disease}

When the reproduction number is below $1$, each case fails to replace itself, so the infected population declines and the disease eventually dies out. In the stochastic setting, random fluctuations can alter this deterministic prediction and let the disease persist longer than expected. For instance, random variation in birth, death, or environmental degradation rates can create temporary conditions favorable to transmission, sustaining infection beyond the deterministic forecast.

The experimental results support this. Figure \ref{fig:SDElowcontangious} shows the evolution of infections under different environmental control policies, while a similar harvesting strategy is maintained. As the hunting rate decreases from $10\%$ to $4\%$ across testing environments, the outcomes differ. When the decontamination level is held at $1$, the disease is eradicated (Figure \ref{fig:SDElowcontangious} (a)-(c)). As the environmental degradation rate increases to $2$ (Figure \ref{fig:SDElowcontangious} (d)-(f)) and then to $3$ (Figure \ref{fig:SDElowcontangious} (g)-(i)), the policy fails to control the infection.

\begin{figure}[H]
    \centering
    \subfigure[]{\includegraphics[width=4.3 cm]{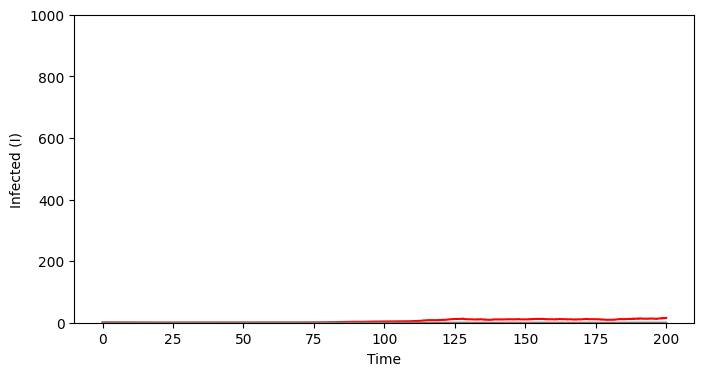}}
    \subfigure[]{\includegraphics[width=4.3 cm]{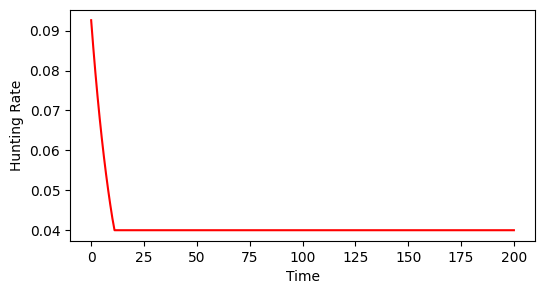}}  
    \subfigure[]{\includegraphics[width=4.2 cm]{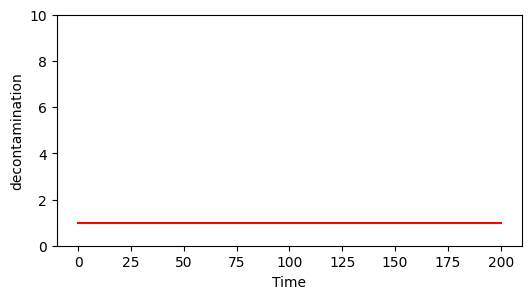}} 
    \subfigure[]{\includegraphics[width=4.3 cm]{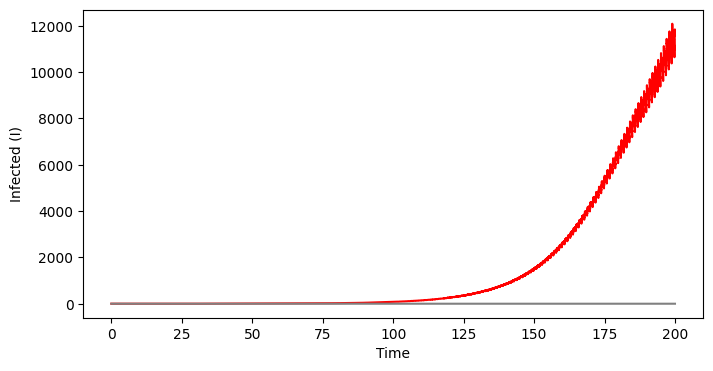}}
    \subfigure[]{\includegraphics[width=4.3 cm]{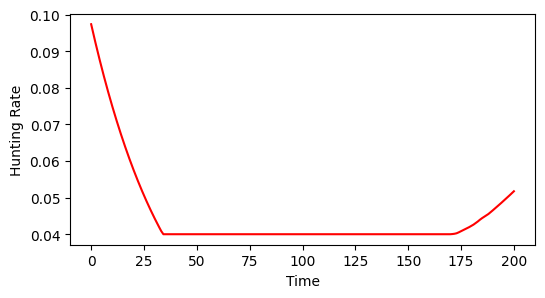}}  
    \subfigure[]{\includegraphics[width=4.2 cm]{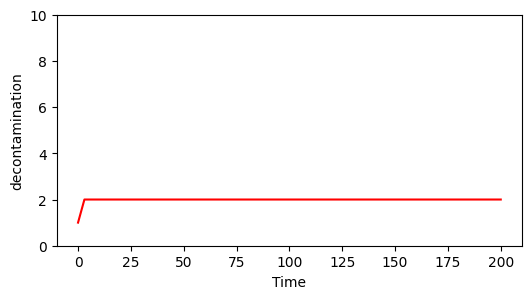}} 
    \subfigure[]{\includegraphics[width=4.3 cm]{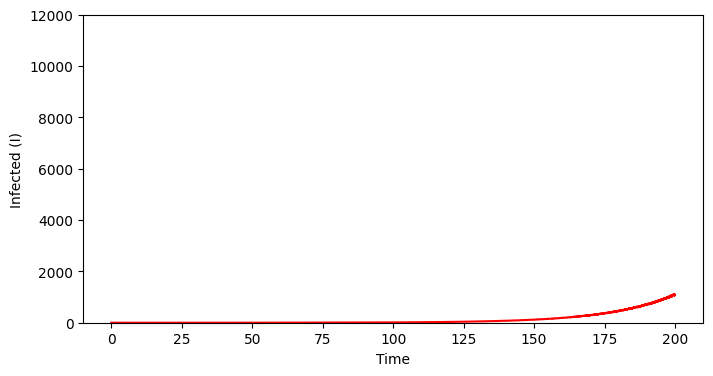}}
    \subfigure[]{\includegraphics[width=4.3 cm]{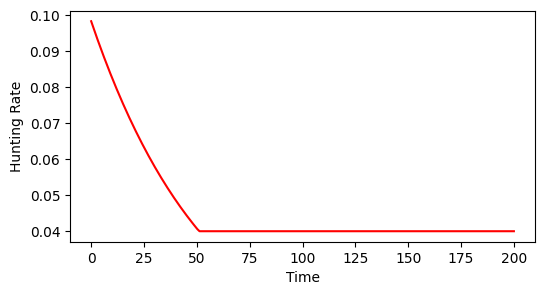}}
    \subfigure[]{\includegraphics[width=4.3 cm]{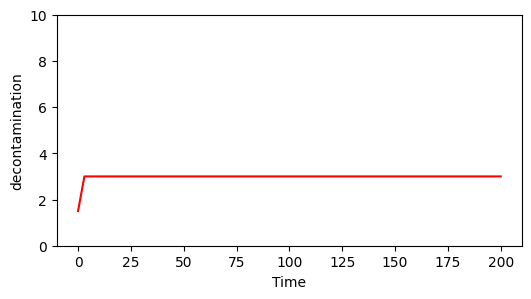}} 
    \caption{Panels (a), (d), and (g) show the evolution of the infected population $I$; panels (b), (e), and (h) show the corresponding harvesting rates; and panels (c), (f), and (i) show the environmental decontamination rates, which increase and stabilize at $1$, $2$, and $3$, respectively.}
    \label{fig:SDElowcontangious}
\end{figure}

\textbf{Case 2: High-Contagion Disease}

When the disease is highly contagious ($R_0 = 22$), it persists over long periods despite control efforts. Our experiments yield several observations. As shown in Figure \ref{fig:SDEhighcontagious} (a)-(c), adjusting the hunting rate without a sufficiently clean environment fails to suppress the infection, since the environmental load keeps sustaining transmission even under intensified harvesting. Holding the decontamination rate at its maximum of $10$ while gradually reducing the hunting rate to $4\%$ likewise fails to halt the disease (Figure \ref{fig:SDEhighcontagious} (d)-(f)). When the hunting rate is instead raised to its maximum, the decontamination rate oscillates, alternately improving and declining. This strategy achieves better control but still cannot eradicate the disease (Figure \ref{fig:SDEhighcontagious} (g)-(i)).

\begin{figure}[H]
    \centering
    \subfigure[]{\includegraphics[width=4.3 cm]{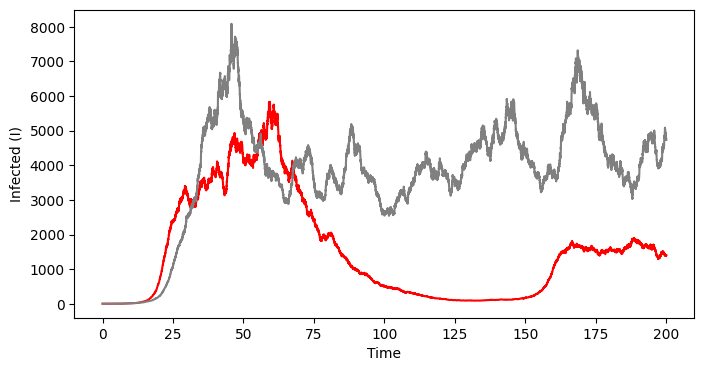}}
    \subfigure[]{\includegraphics[width=4.3 cm]{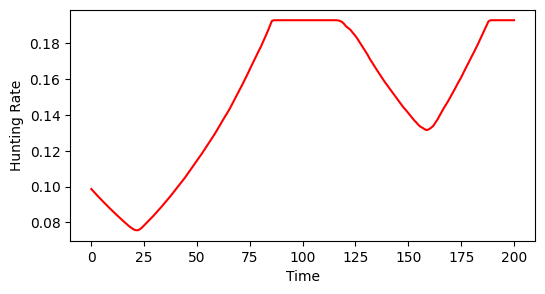}}  
    \subfigure[]{\includegraphics[width=4.3 cm]{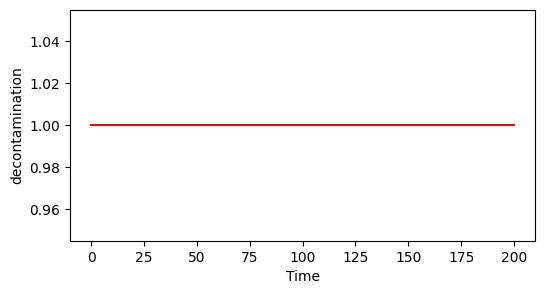}}
    \subfigure[]{\includegraphics[width=4.3 cm]{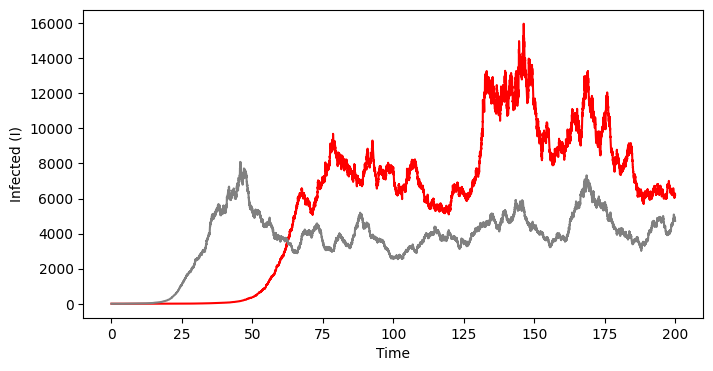}} 
    \subfigure[]{\includegraphics[width=4.3 cm]{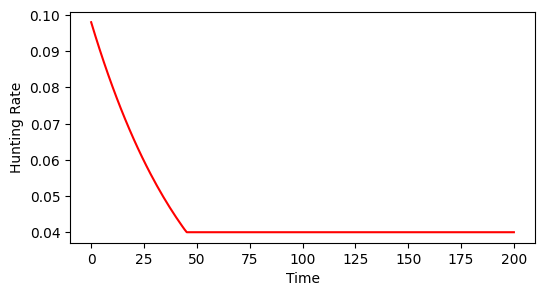}}  
    \subfigure[]{\includegraphics[width=4.3 cm]{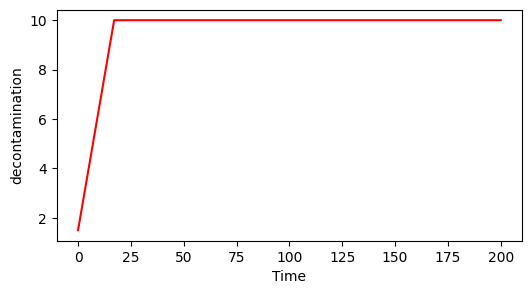}}
    \subfigure[]{\includegraphics[width=4.3 cm]{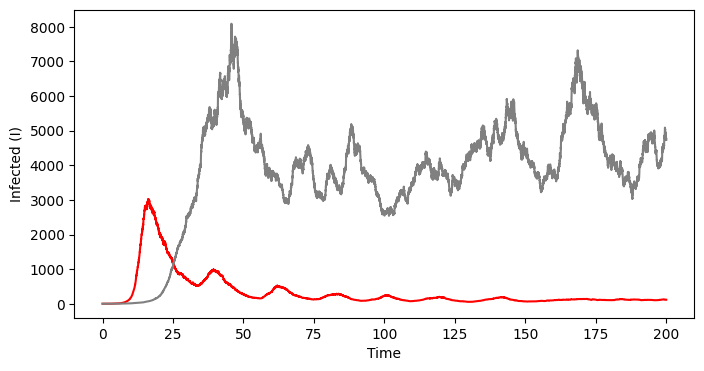}} 
    \subfigure[]{\includegraphics[width=4.3 cm]{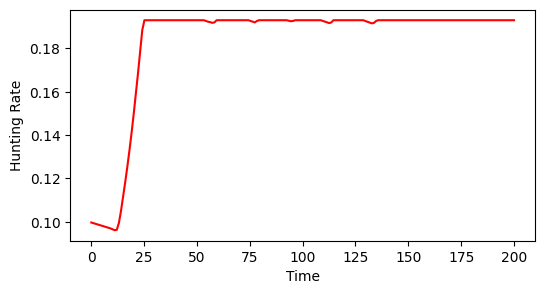}}  
    \subfigure[]{\includegraphics[width=4.3 cm]{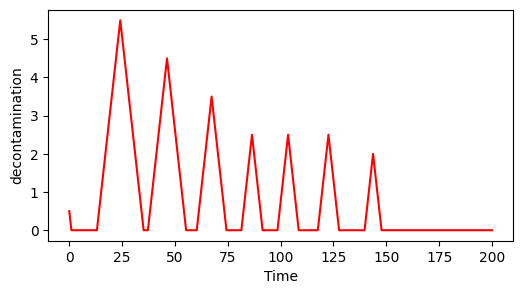}}
    \caption{Panels (a), (d), and (g) show the evolution of the infected population $I$ during the testing phase, before intervention (gray) and after intervention (red); panels (b), (e), and (h) show the corresponding harvesting rates; and panels (c), (f), and (i) show the corresponding decontamination rates.}
    \label{fig:SDEhighcontagious}
\end{figure}

The results show that, for a highly transmissible disease, increasing the hunting rate is effective only when combined with a high level of decontamination (Figure \ref{fig:highcontagiousnew}). This underscores the central role of environmental management in the effectiveness of population-based interventions, especially for highly infectious diseases. Across the three scenarios, the agent recommends different strategies depending on the conditions, but it consistently raises the decontamination rate and holds it high throughout. This reinforces the leading role of environmental management in reducing disease persistence.

\begin{figure}[H]
    \centering
    \subfigure[]{\includegraphics[width=7 cm]{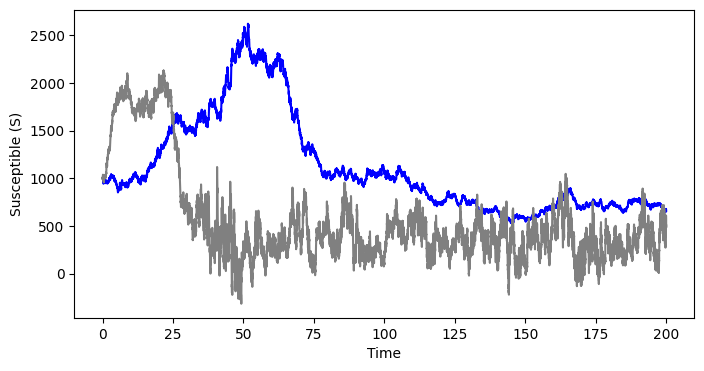}}
    \subfigure[]{\includegraphics[width=7 cm]{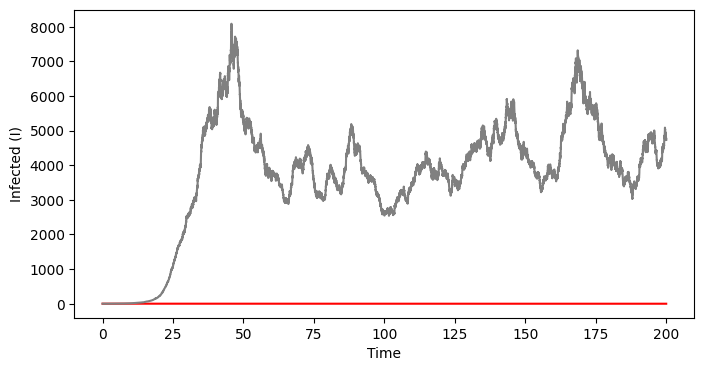}}  
    \subfigure[]{\includegraphics[width=7 cm]
    {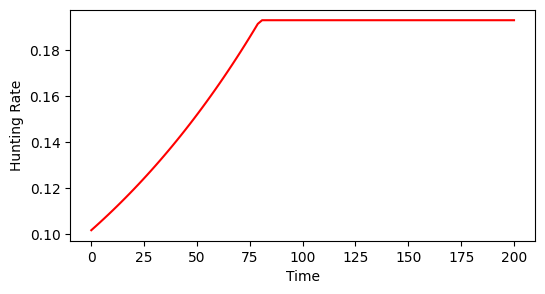}}
    \subfigure[]{\includegraphics[width=7 cm]
    {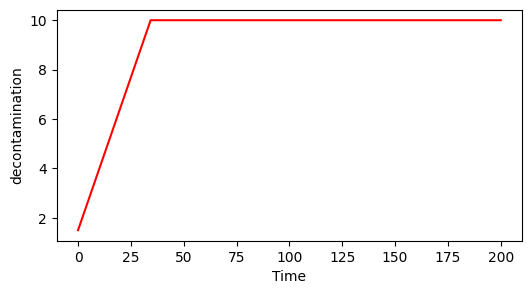}}
    \caption{Panels (a) and (b) show the evolution of the susceptible (blue) and infected (red) populations during a successful testing phase under intervention, compared with the non-intervention baseline (gray); panels (c) and (d) show the corresponding harvesting rates and environmental decontamination rates.}
    \label{fig:highcontagiousnew}
\end{figure}

\textbf{Impact on Effective Reproduction Number.}

We examine the effective reproduction number $R_e(t)$ under different strategies. Figure \ref{fig:RCWD} (a)-(c) shows its evolution under three policies: increasing the hunting rate alone, increasing the decontamination level alone, and decreasing the hunting rate while increasing decontamination (Figure \ref{fig:DE1_SIUcompare}, \ref{fig:DE2_SIUcompare}, \ref{fig:DE3_SIUcompare}). The hunting-only strategy gradually reduces $R_e(t)$. The decontamination-only strategy produces a smaller reduction and stays nearly constant. The combined policy yields a larger $R_e(t)$ (Figure \ref{fig:RCWD} (c)), yet earlier simulations show the disease remains controlled, since this strategy reduces both the infected population and environmental contamination (Figure \ref{fig:DE3_SIUcompare}). Thus $R_e(t)$ alone may not fully capture long-term dynamics under adaptive management and population-reduction strategies.

\begin{figure}[H]
    \centering
    \subfigure[]{\includegraphics[width=5 cm]{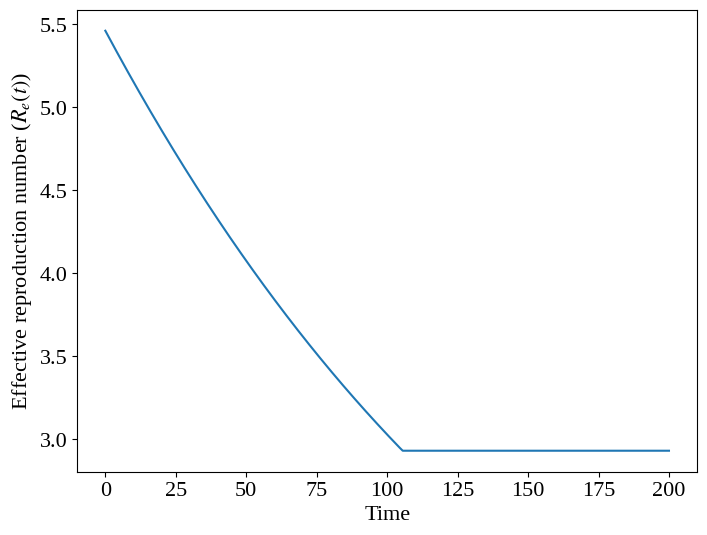}}
    \subfigure[]{\includegraphics[width=5 cm]{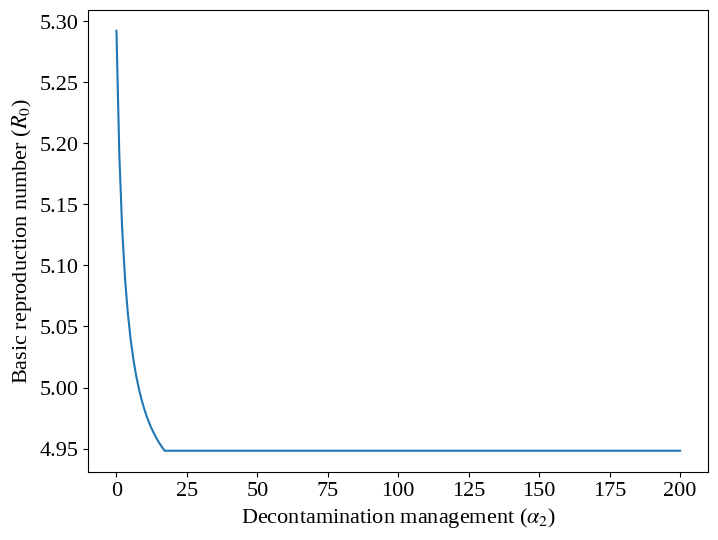}}  
     \subfigure[]{\includegraphics[width=5 cm]{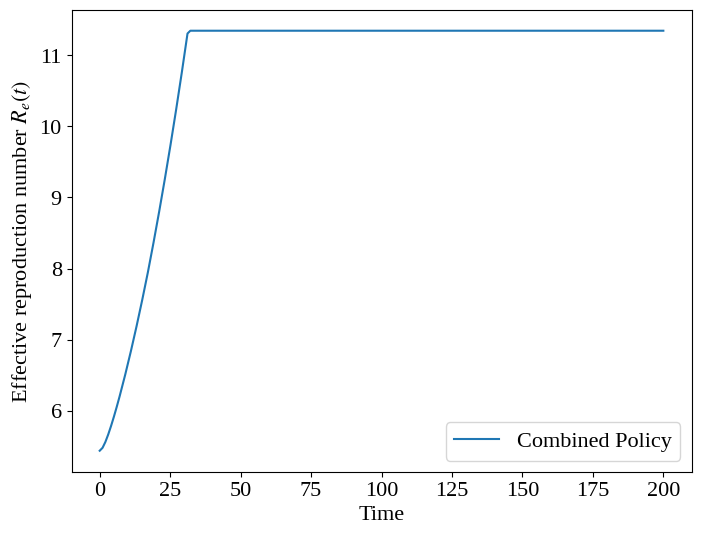}}  
    \caption{Evolution of the effective reproduction number $R_e(t)$ under three intervention strategies: (a) increasing the hunting rate only; (b) increasing the environmental decontamination level only; and (c) decreasing the hunting rate while increasing the decontamination level.}
    \label{fig:RCWD}
\end{figure}

Figure \ref{fig:highcontagiousR} shows $R_e(t)$ under two strategies when the disease initially has high transmission potential (cf. Figure \ref{fig:SDEhighcontagious} (e), (f) and \ref{fig:highcontagiousnew} (c), (d)). Increasing decontamination while reducing the hunting rate causes $R_e(t)$ to rise rapidly and stay high (Figure \ref{fig:highcontagiousR} (a)), yet the disease is controlled (Figure \ref{fig:SDEhighcontagious} (e), (f)). Increasing both hunting and decontamination substantially lowers $R_e(t)$ (Figure \ref{fig:highcontagiousR} (b) and \ref{fig:highcontagiousnew} (c), (d)), so this strategy is more effective against highly contagious outbreaks.

\begin{figure}[H]
    \centering
    \subfigure[]{\includegraphics[width=6 cm]{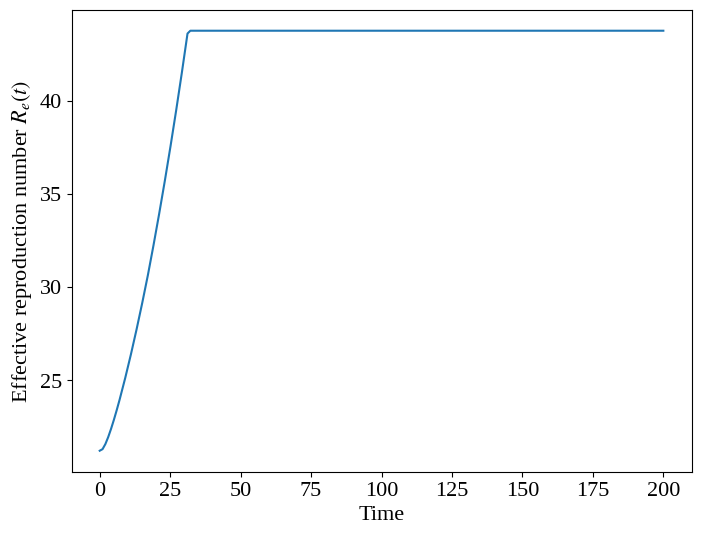}}
    \subfigure[]{\includegraphics[width=6 cm]{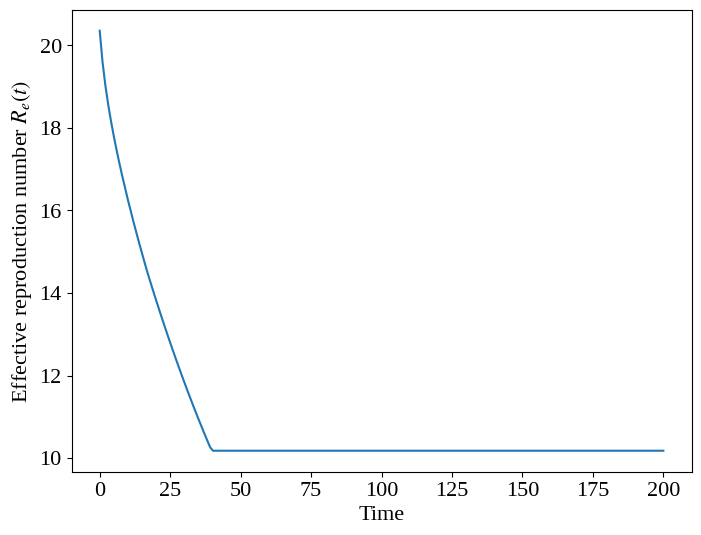}}  
    \caption{Evolution of the effective reproduction number $R_e(t)$ under two intervention strategies: (a) decreasing the hunting rate while increasing the environmental decontamination level; and (b) increasing both the hunting rate and the environmental decontamination level.}
    \label{fig:highcontagiousR}
\end{figure}

\section{Discussion}\label{Sec:Disc}
In this paper, we develop a stochastic ecological model that incorporates environmental and demographic stochasticity into disease dynamics. The framework evaluates three management strategies for controlling CWD, hunting, environmental decontamination, and a combination of the two, in both deterministic and stochastic settings, examines their performance across different transmission intensities, and demonstrates the robustness and adaptability of the learned policies. Because the model couples adaptive, learned decision-making with environmental transmission, it can also be extended to other infectious disease systems in which an environmental reservoir sustains transmission.

Beyond the control results, the analysis of the model itself helps explain why environmental management emerges as the dominant lever. The basic reproduction number decomposes into a direct-transmission contribution and an environmental contribution, $R_0 = R_{\mathrm{dir}} + R_{\mathrm{env}}$, and across most of the parameter space the environmental route dominates, $R_{\mathrm{env}} \gg R_{\mathrm{dir}}$, because the slow intrinsic disease dynamics render direct transmission relatively inefficient. Consequently, the $R_0=1$ threshold is far more responsive to changes in the decontamination rate $\alpha_2$ than to the hunting rate $\alpha_1$, which is consistent with the reinforcement learning agent's repeated preference for raising and sustaining decontamination. The deterministic analysis further shows that the disease-free and endemic equilibria exchange stability through a forward (supercritical) transcritical bifurcation at $R_0=1$, with no backward branch, so that hunting can drive infection to extinction before threatening population persistence, but only up to the demographic threshold $\alpha_1=\nu-\mu\approx 0.1931$, beyond which the same control that suppresses the disease extirpates the host. This threshold directly motivates the ceiling $\alpha_1^{\max}$ imposed on the agent's action space. For the reflected stochastic system, we established well-posedness and positivity and showed that the normal reflection at $\{S=0\}$ leaves both the disease-free invariant measure and the transverse infection linearization unchanged, so that the diffusion-corrected proxy $\mathcal{R}_0^s$ and the top Lyapunov exponent governing disease invasion coincide for the reflected and unreflected models.

CWD prions are exceptionally difficult to eliminate from the environment, and residual infectivity in contaminated sites appears to be an important factor in sustaining epizootics \cite{miller2004chronic}. Unlike conventional pathogens, prions resist virtually all standard inactivation methods, including fire, ionizing radiation, chemical disinfectants, and autoclaving, each of which can reduce but not fully eliminate prion infectivity \cite{trapotsis2016prion}. Prions remain infectious in external environments for years, so contaminated soil acts as a long-term reservoir that likely sustains CWD incidence in free-ranging cervids and complicates eradication in captive herds; soil therefore represents a plausible transmission route for deer and sheep \cite{miller2004environmental,almberg2011modeling}.

The challenge of environmental decontamination is compounded by the impracticality of conventional best practices such as incineration and autoclaving for most field applications, and by the limited number of effective alternatives reported. Several promising mitigation strategies have nonetheless emerged. Degradation of pathogenic prion protein by a synthetic analog of birnessite, a naturally occurring manganese oxide mineral (MnO\textsubscript{2}) common in soils, appears promising for environmental remediation: upon exposure to MnO\textsubscript{2}, prion levels decreased by at least four orders of magnitude, indicating that soil mineral composition can attenuate prions in some environments \cite{russo2009pathogenic}.

Soil humic acids have similarly been shown to degrade CWD prions and reduce infectivity levels \cite{kuznetsova2018soil}. At the farm scale, treating CWD-contaminated soil with 2 N sodium hydroxide (NaOH) achieved a greater than 100-fold reduction in prion infectivity, making it a practical candidate for field use \cite{kang2019sodium}. The situation differs substantially between wild and farmed cervid contexts. In wild populations, environmental remediation of CWD-contaminated landscapes is essentially impractical due to their scale, and once prions are established in soil, elimination is not currently achievable. Farmed operations, by contrast, are smaller, bounded, managed spaces that permit more cost-feasible and systematic decontamination once CWD is detected in the farm herd.

Regulatory agencies in both the United States and Canada have developed structured response protocols for CWD-positive farms. The U.S. Animal and Plant Health Inspection Service (APHIS) herd management plans for CWD-positive premises may require whole-herd depopulation, mandatory cervid-free periods following removal of positive or exposed animals, fencing requirements, restrictions on the movement of contaminated equipment, and premises cleaning and disinfection \cite{aphis2023program}. In Canada, Canadian Food Inspection Agency (CFIA)-validated decontamination of a CWD-affected farm in Ontario involves physically removing at least five centimeters of soil from high-traffic areas and replacing it with at least ten centimeters of clean soil sourced from sites that have never housed cervids \cite{morgan2023gslr}. All equipment, materials, and areas classified as minimally contaminated are also cleaned and disinfected. Hard surfaces and equipment are treated chemically with sodium hydroxide or sodium hypochlorite, while reusable surgical and handling instruments are sterilized by autoclaving at elevated temperatures and pressure \cite{trapotsis2016prion, mcdonnell2003challenge}. Farm decontamination also addresses one of the most consequential transmission pathways: the movement of CWD prions from high-density farmed populations to adjacent wild cervids across fence lines. Combined with biosecurity measures, decontamination substantially reduces this farm-to-wild transmission route and lowers the risk of spillover to free-ranging deer herds \cite{kincheloe2021chronic}. Although these strategies reflect current regulatory and field-based practice, determining the optimal timing, intensity, and combination of interventions across larger landscapes remains a complex management challenge. We therefore adopt a computational framework to evaluate how adaptive decision-making, integrating both population control and environmental pathogen management, can be optimized to minimize the spread of CWD and its long-term ecological consequences.

In the deterministic setting, the DRL agent learns effective strategies and, when compared against current practice, reveals that more effective disease-control actions were available. The three single- and combined-strategy comparisons expose a clear trade-off. Hunting alone can bring the outbreak under control, but only at the cost of the total population, which declines substantially over the subsequent 200 years. Decontamination alone can likewise control the disease, but requires sustaining a consistently high decontamination rate, whose labor, equipment, chemical, and long-term operating costs may limit its sustainability. Only the combined policy resolves this tension: the total cervid population more than doubles while infection prevalence and environmental contamination fall to near-zero levels, indicating that pairing population control with environmental decontamination is essential for both sustained disease control and long-term ecosystem recovery. The decision-maker's priorities, encoded in the reward weights, shape these outcomes; emphasizing environmental conditions leads the agent to enforce a cleaner environment, thereby reducing infection.

Under stochastic dynamics, the agent converges to a consistent policy, a gradual reduction in the hunting rate combined with a rapid increase and sustained high level of environmental decontamination, which contains the disease in roughly $80\%$ of $30$ independent realizations, with about $10\%$ producing large outbreaks. This underscores the substantial influence of environmental randomness on both disease dynamics and the effectiveness of control. The stochastic policy is broadly consistent with its deterministic counterpart but differs in one notable respect: under environmental noise, the agent front-loads hunting effort, maintaining a higher hunting rate for a shorter initial period before switching to minimal culling while sustaining maximal decontamination throughout, thereby compensating for reduced predictability. More broadly, whereas earlier CWD modeling has largely relied on population-level regression and compartment models, and CWD decision-making as a form of stochasticity has not been well integrated into these approaches, the present framework couples adaptive, learned decision-making with both environmental and demographic stochasticity within a single model. Examining low- and high-contagion regimes reinforces the primacy of environmental management: even under a fixed harvesting strategy, different decontamination levels are required to achieve eradication, and when the disease is highly contagious, raising the hunting rate while maintaining a high decontamination level proves more effective than reducing hunting effort. Across all scenarios, the agent consistently raises and sustains environmental decontamination, confirming its central role in reducing disease persistence.

Our examination of the effective reproduction number $R_e(t)$ highlights a subtlety in interpreting adaptive control. Under the combined policy, $R_e(t)$ can remain comparatively high, or even rise, while the disease is nonetheless brought under control, because the policy simultaneously reduces the infected population and the environmental contamination that would otherwise sustain transmission. In the highly contagious regime, decreasing hunting while increasing decontamination drives $R_e(t)$ upward yet still controls the disease, whereas increasing both hunting and decontamination substantially lowers $R_e(t)$. This indicates that $R_e(t)$ alone may not fully capture long-term disease dynamics under adaptive management and population-reduction strategies, and that control assessments should track the joint trajectory of the infected population and the environmental reservoir rather than the reproduction number in isolation.

By explicitly incorporating stochasticity and parameter variability, our framework shows potential for application to a broader class of infectious disease models beyond the one considered here. The various aspects of the resulting policies examined in this study may help guide future work on training similar systems in both deterministic and stochastic settings.

CWD management entails both direct and indirect costs, according to the National Institutes of Health. Direct expenses include monitoring and testing, which require funding for sample collection, laboratory analysis, and specialized equipment, as well as management activities such as culling infected or at-risk animals, with associated costs for personnel, equipment, and carcass disposal. Implementing and enforcing related regulations demands further resources, and personnel time for wildlife managers, researchers, and support staff accounts for a substantial share of overall expenditure. Indirectly, CWD can harm hunting tourism, reducing license sales and associated revenue for state and local economies. The cost of environmental decontamination, however, is not well established and varies considerably depending on the specific situation and decontamination approach. Reported costs include personnel time, sample processing, travel, equipment, management actions such as sharpshooting or culling, regulatory enforcement, outreach, veterinary expenses, and potential impacts on the hunting and farmed cervid industries \cite{chiavacci2022economic}. Furthermore, environmental decontamination is not always recommended as a primary strategy for CWD control because prions can persist in the environment for many years, making complete removal extremely difficult. Decontamination efforts are also often costly and challenging to implement across large geographic areas, and some chemical treatments effective against prions may harm ecosystems and non-target species. Incorporating these economic and practical considerations into the model is a promising direction for future work.

Another limitation of the proposed framework is that the reliability of the learned policy decreases as environmental noise increases. As $\sigma_1$, $\sigma_2$, and $\sigma_3$ increase from $0.03$ to $0.07$, disease control becomes more challenging, suggesting reduced policy robustness under stronger stochastic perturbations. Future work will investigate reinforcement learning algorithms and robust control frameworks better suited to highly stochastic environments.

A further limitation is that the model is non-spatial, treating the cervid population and its environmental reservoir as well-mixed, whereas CWD spreads geographically and concentrates at hotspots frequented by deer for water, food, or shelter; incorporating spatial structure and localized contamination is a natural extension. In addition, several epidemiological inputs remain uncertain; there is currently no universally accepted value for the environmental shedding rate $\xi$, and the environmental decay and demographic rates vary with species, age, climate, and habitat, so the calibrated values should be regarded as plausible baselines rather than precise estimates.

\section*{Declaration of generative AI and AI-assisted technologies in the manuscript preparation process}
During the preparation of this work, the author(s) used Claude and Grammarly to edit and correct the English language. After using these tools, the author(s) reviewed and edited the content as needed and took full responsibility for the content of the published article.
\newpage
\appendix
\section*{Appendix}
\section{The reflected model and its well-posedness}
\label{sec:reflected}

\subsection{The CWD stochastic model}
On the closed domain $D=\{(S,I,U):\,S\ge 0,\ I\ge 0,\ U\ge 0\}$ with $N=S+I$, the reflected process $(S,I,U)$ and its boundary regulator $L$ solve
\begin{subequations}\label{eq:reflecteda}
\begin{align}
\dif S &= f_{S}\,\dif t
        + \sigma_{1} N\,\dif W_{1}
        + \sigma_{2} S\,\dif W_{2}
        + \dif L, \label{eq:reflected-Sa}\\
\dif I &= f_{I}\,\dif t
        + \sigma_{2} I\,\dif W_{2}, \label{eq:reflected-Ia}\\
\dif U &= f_{U}\,\dif t
        + \sigma_{3} U\,\dif W_{3}, \label{eq:reflected-Ua}
\end{align}
\end{subequations}
with drifts
\begin{align*}
f_{S} &= \nu N\left(1-\frac{N}{K}\right)
        - \frac{\beta S I}{N} - \beta_{e} S U
        - \mu S\left(1-\frac{N}{K}\right) - \alpha_{1} S,\\
f_{I} &= \frac{\beta S I}{N} + \beta_{e} S U
        - \mu I\left(1-\frac{N}{K}\right) - \gamma I - \alpha_{1} I,\\
f_{U} &= \xi I - (\epsilon+\alpha_{2}) U,
\end{align*}
where $W_{1},W_{2},W_{3}$ are independent standard Brownian motions and $L$ is a continuous, nondecreasing, adapted process representing the local time of $S$ at $0$ to guarantee that $S\geq 0$ and such that
\begin{equation}\label{eq:localtimea}
L_{0}=0,\qquad \dif L_{t}\ge 0,\qquad
\int_{0}^{t}\indic_{\{S_{s}>0\}}\,\dif L_{s}=0 .
\end{equation}
The reflection direction is the inward normal $+\mathbf{e}_{S}$, so $L$ acts only when $S=0$ and injects the least susceptible flux consistent with $S\geq 0$. In the $(N,I,U)$ chart the regulator $L$ enters through $N=I+S$ alone,
\begin{equation}\label{eq:reflected-N}
\dif N=\Big[(\nu-\mu)N\big(1-\frac{N}{K}\big)-\alpha_{1}N-\gamma I\Big]\dif t +\sigma_{1}N\,\dif W_{1}+\sigma_{2}N\,\dif W_{2}+\dif L .
\end{equation}
With $S\ge 0$ enforced, the incidence terms $\beta SI/N$ and $\beta_{e}SU$ are genuine nonnegative fluxes and the ratio bound $S/N\le 1$ holds for all $t$.

\subsection{Simulation Algorithm: Reflected Euler Scheme}
Numerically, we integrate \eqref{eq:reflected} by a reflected (L\'epingle) Euler-Maruyama step on $S$, leaving the $I,U$ updates unreflected; the local time is recovered as the accumulated reflection increment, giving a consistent estimator of $\ell$. 

\begin{algorithm}[t]
\caption{Reflected Euler-Maruyama for the CWD SDE \eqref{eq:reflected}}
\label{alg:reflected-euler}
\begin{algorithmic}[1]
\Require initial state $(S_{0},I_{0},U_{0})\in D$; step $\Delta t$; horizon $T$
\State $L_{0}\gets 0$;\quad $n_{\max}\gets \lfloor T/\Delta t\rfloor$
\For{$n=0,1,\dots,n_{\max}-1$}
  \State $N_{n}\gets S_{n}+I_{n}$
  \State draw $\Delta W_{1},\Delta W_{2},\Delta W_{3}\stackrel{\text{iid}}{\sim}\mathcal{N}(0,\Delta t)$
  \State $\tilde S\gets S_{n}+f_{S}(S_{n},I_{n},U_{n})\,\Delta t
          +\sigma_{1}N_{n}\,\Delta W_{1}+\sigma_{2}S_{n}\,\Delta W_{2}$
  \State $S_{n+1}\gets \max(\tilde S,\,0)$
         \Comment{normal reflection at $\{S=0\}$}
  \State $\Delta L_{n}\gets \max(-\tilde S,\,0)$;\quad
         $L_{n+1}\gets L_{n}+\Delta L_{n}$
  \State $I_{n+1}\gets I_{n}+f_{I}(S_{n},I_{n},U_{n})\,\Delta t
          +\sigma_{2}I_{n}\,\Delta W_{2}$
  \State $U_{n+1}\gets U_{n}+f_{U}(S_{n},I_{n},U_{n})\,\Delta t
          +\sigma_{3}U_{n}\,\Delta W_{3}$
  \State clip $I_{n+1},U_{n+1}\gets\max(\cdot,0)$ to guard round-off
\EndFor
\State \Return $\{(S_{n},I_{n},U_{n},L_{n})\}_{n=0}^{n_{\max}}$
\end{algorithmic}
\end{algorithm}
The scheme has the usual order-$\frac12$ strong (order-$1$ weak) convergence for normally reflected SDEs, and $L_{n}$ is obtained as a by-product for estimating $\ell$. As a validation, $\Rzeros$ should agree with a Monte-Carlo top Lyapunov exponent $\hat\lambda=T^{-1}\log\|(I_{T},U_{T})\|$ computed on the subcritical face, where the reflection is inert.

\subsection{Parameters Values}

\begin{table}[H]
\centering
\caption{Model parameters, descriptions, baseline values, and corresponding references.}
\label{tab:parameters}

\begin{tabular}{|p{0.15\linewidth}|p{0.40\linewidth}|p{0.18\linewidth}|p{0.17\linewidth}|}
\hline
\textbf{Parameter} &
\textbf{Description} &
\textbf{Baseline Value} &
\textbf{Reference}
\\
\hline

$\nu$ & Birth rate & $0.2$ & Assumed \\
\hline

$\beta$ & Direct transmission rate & $0.00052$ & \cite{Noelle2024} \\
\hline

$\beta_e$ & Environmental transmission rate & $0.00065$ & \cite{Noelle2024} \\
\hline

$\gamma$ & Disease-induced mortality rate & $0.00015$ & \cite{Noelle2024} \\
\hline

$\mu$ & Natural mortality rate & $1/144$ & Approximated from \cite{islam2022modeling} \\
\hline

$\xi$ & Pathogen shedding rate into the environment & $0.1$ & Assumed \\
\hline

$\epsilon$ & Natural environmental decay rate & $1/120$ & Assumed \\
\hline

$\alpha_1$ & Deer harvesting rate & $0.1$ & \cite{tpwd_deer_management} \\
\hline

$\alpha_2$ & Environmental decontamination rate & $1$ & Assumed \\
\hline




$\sigma_i,\ i=1,2,3$ & Variability parameters & $(0.05,\,0.05,\,0.05)$ & Assumed \\
\hline

\end{tabular}
\end{table}

\subsection{Calibration}
The shedding rate $\xi$ describes the rate at which infected deer or elk release infectious prions into the environment through saliva, urine, feces, blood, or carcasses \cite{denkers2024temporal}. There is currently no universally accepted value for shedding rate in epidemiological models.

The environmental decay rate of CWD prions is believed to be very low, since prions can persist in soil for many years, and possibly far longer under favorable conditions \cite{miller2004environmental}. We therefore assume a small environmental decay rate of $\frac{1}{120}$.

The natural mortality rate $\mu$ for cervids varies with species, ages, climate, hunting pressure, and habitat quality  \cite{schuyler2019effects}. A previous study estimates the natural mortality rate for bulk at $\frac{1}{132}$  \cite{islam2022modeling}; we adopt $\mu = \frac{1}{144}$.

The birth rates $\nu$ likewise vary across species and environmental conditions. White-tailed deer females commonly produce one to two fawns annually, whereas mule deer and elk populations generally show lower effective recruitment owing to differences in survival and reproductive ecology \cite{miller2006dynamics, cross2025predictions}. We assume a birth rate of $0.2$ for the deer family.

\section{Auxiliary Lemmas and Definitions}
The Appendix presents several important definitions, lemmas, and theorems. Some are drawn from the existing literature; the remainder are proven here.
\begin{lemma}[Generalized L'Hospital rule \cite{Lee1977Generalizations}]\label{lHop}
For functions $f$ and $g$ satisfying the standard hypotheses,
\begin{equation} \label{gen_LHospital}
    \liminf_{t\to \infty} \dfrac{f'(t)}{g'(t)} \leq \liminf_{t\to \infty} \dfrac{f(t)}{g(t)} \leq \limsup_{t\to \infty} \dfrac{f(t)}{g(t)} \leq \limsup_{t\to \infty} \dfrac{f'(t)}{g'(t)}.
\end{equation}
\end{lemma}

Let $(\Omega,\mathcal{F},\{\mathcal{F}_t\}_{t\geq 0},\mathbb{P})$ be a complete probability space with filtration $\{\mathcal{F}_t\}_{t\geq 0}$, with respect to which all of the stochastic processes introduced below are defined.

\begin{lemma}[Ito's formula]\label{lem:ito}
Let $X(t)$ satisfy the one-dimensional stochastic differential equation
$$
    dX({t})=\mu_{t}dt+\sigma_{t}dW({t}).
$$
       If $f\left(t,x\right)$ is a twice-differentiable scalar function, then
$$
    df(t,X(t))=\left(\frac{\partial f}{\partial t}+\mu_{t}\frac{\partial f}{\partial x}+\frac{\sigma_{t}^{2}}{2}\frac{\partial^{2}f}{\partial x^{2}}\right)dt+\sigma_{t}\frac{\partial f}{\partial x}dW(t).
$$

\end{lemma}

Consider the $d$-dimensional stochastic differential equation
$$
dX(t)=F(t,X(t)) dt + G(t,X(t)) dW(t),
$$
where $F(t,x)=\left(f_{i}(t,x)\right)_{1\leq i\leq d}$ is a $d\times 1$ vector and $G(t,x)=\left(g_{i,j}(t,x)\right)_{1\leq i\leq d,1\leq j\leq n}$ is a $d\times n$ matrix, both defined on $[0,\infty]\times \mathbb{R}^d$ and locally Lipschitz in $x$, and $W(t)$ is an $n-$ dimensional Brownian motion. Let $V \in C^{1,2}([0,\infty]\times \mathbb{R}^d,[0,\infty)),$ and define the differential operator $L$ by $$LV(t,x)=\dfrac{\partial V(t,x)}{\partial t}+F(t,x)^T DV(t,x) +\frac12 \text{Trace}(G(t,x)^T  H(t,x) G(t,x) ),$$ where $DV(t,x)=\dfrac{\partial V(t,x)}{\partial x}$ is the gradient of $V$ and $H(t,x)=\dfrac{\partial^2 V(t,x)}{\partial x^2}$ is its Hessian. It\^o's formula for $X(t)$ and such a function $V$ states that $$dV(t,X(t))=LV(t,X(t))dt+DV(t,X(t))^T G(t,X(t)) dW(t)$$ for $t\geq 0$. The following lemma appears in \cite{mao2011stochastic}.

\begin{lemma}\label{Ap2}
Let $\{M_t:t\geq 0\}$ be a real-valued continuous local martingale with respect to $\{\mathcal{F}_t\}_{t\geq 0}$ with $M_0=0$. If $\limsup_{t \to \infty} \frac{\mathbb{E}(M_t^2)}{t}   < \infty,$ then $\lim_{t \to \infty} \frac{M_t}{t}=0$ almost surely.
\end{lemma}

The following lemma could be found in the appendix of \cite{balogh2025stochastic}. 
\begin{lemma}\label{Ap1}
Let $f\in C([0,\infty)\times \Omega,(0,\infty))$, and let $a(t)$ be a process satisfying $\lim_{t \to \infty} \dfrac{a(t)}{t}=a$ almost surely for some positive constant $a$.  Suppose there exists a positive constant $b$ such that
$$\log(f(t))\geq (\leq) \, a(t)-b\int_0^t f(s)ds $$
for all $t\geq 0$ almost surely. Then
$$\liminf_{t\to \infty} (\limsup_{t\to \infty})\, \frac{1}{t} \int_0^t f(s)ds\geq (\leq) \, \dfrac{a}{b} \quad a.s.$$
\end{lemma}

\begin{proof} We prove the "$\geq$" case; the "$\leq$" case follows by the same argument. Let $g(t):=b\,f(t),$ so that the inequality becomes
$$
\log(g(t))\geq (\leq) \, a(t)+\log(b)-\int_0^t g(s)ds. 
$$

Since $\lim_{t \to \infty} \dfrac{a(t)}{t}=a$ almost surely, for any sufficiently small $\epsilon>0$ there exists $\Omega_\epsilon\in \mathcal{F}$ with $\mathbb{P}(\Omega_\epsilon)\geq 1-\epsilon$ such that
$$
0<a-\epsilon \leq \dfrac{a(t)}{t}(\omega)\leq a+\epsilon \quad \text{for all} \,t\geq t_0(\omega) \, \text{and each} \, \omega \in \Omega_\epsilon.
$$
Let $G(t)=\int_0^t g(s) ds$ for all $t\geq 0$; then $dG(t)=g(t)dt$ and 
$$
G(t)\geq a(t)+\log(b) -\log(g(t)) \quad \text{for all} \,t\geq 0. 
$$
Therefore, 
$$
de^{G(t)}=e^{G(t)}g(t)dt\geq e^{a(t)+\log(b) -\log(g(t))} g(t)dt=b e^{a(t)} dt,
$$
and $e^{a(t)}\geq e^{(a-\epsilon)t}$ on $\Omega_\epsilon$. Thus, for $t>t_0$
$$
G(t)\geq \log\left(e^{G(t_0)}+\frac{b}{(a-\epsilon)}[e^{(a-\epsilon)t}-e^{(a-\epsilon)t_0}]\right).
$$ 
Hence, by L'Hospital's rule (Lemma \ref{lHop})
$$
\liminf_{n\to \infty} \dfrac{G(t)}{t}\geq a-\epsilon,
$$
and so
$$
\liminf_{t\to \infty} \frac{1}{t} \int_0^t f(s)ds\geq \dfrac{a-\epsilon}{b} \quad \text{ on }\Omega_\epsilon.
$$
Since $\epsilon>0$ is arbitrary, 
$$
\liminf_{t\to \infty}  \frac{1}{t} \int_0^t f(s)ds\geq  \, \dfrac{a}{b} \quad a.s.
$$
\end{proof}

\newpage
\bibliographystyle{acm}
\bibliography{CWD}

\end{document}